\documentclass[prx,twocolumn,superscriptaddress,nofootinbib,longbibliography]{revtex4-2}

\usepackage[pdftex]{graphicx}
\usepackage{amsmath,amsfonts,amsbsy,amssymb,amsthm,mathtools}
\usepackage{mathrsfs}
\usepackage{epstopdf}
\usepackage{ulem,color}
\usepackage{enumitem}
\usepackage{siunitx}
\usepackage{amssymb,amsmath}
\usepackage{ulem,color}
\usepackage{MnSymbol}
\usepackage{environ}
\usepackage{xcolor}
\usepackage{stmaryrd}
\usepackage{textcomp}

\usepackage{tikz}
\usetikzlibrary{arrows.meta, positioning, calc}

\usepackage{tikz}

\usepackage{tikz-cd}

\usepackage{pgfplots}
\pgfplotsset{compat=1.18}

\usepackage{pstool}
\usepackage{tikzsymbols}
\usetikzlibrary{plotmarks,spy}

\newtheorem{lemma}{Lemma}
\newtheorem{corollary}{Corollary}
\newtheorem{theorem}{Theorem}
\newtheorem{fact}{Fact}

\theoremstyle{definition}
\newtheorem{definition}{Definition}
\newtheorem{conjecture}{Conjecture}
\newtheorem{remark}{Remark}
\newtheorem{example}{Example}

\usepackage[colorlinks = true, citecolor=blue, linkcolor=blue, urlcolor=blue]{hyperref}
\usepackage[T1]{fontenc}
\usepackage{mwe}

\newcommand{\ketbra}[1]{|{#1}\>\mkern-4mu\<{#1}|}

\renewcommand{\>}{\rangle}
\newcommand{\<}{\langle}

\newcommand{\C}{{\mathbb{C}}}
\newcommand{\R}{{\mathbb{R}}}

\newcommand{\F}{{\mathbb{F}}}
\newcommand{\Z}{{\mathbb{Z}}}
\newcommand{\B}{{\mathbb{B}}}
\newcommand{\one}{\chi}
\newcommand{\CNOT}{{\mathrm{CNOT}}}

\newcommand{\spn}{\textrm{span}}
\newcommand{\wt}{{\textrm{wt}\,}}
\newcommand{\restr}{\textrm{res}}
\newcommand{\supp}{\mathrm{supp}}

\newcommand{\blank}{{-}}

\DeclarePairedDelimiter\ceil{\lceil}{\rceil}
\DeclarePairedDelimiter\floor{\lfloor}{\rfloor}

\newcommand{\CQT}{Centre for Quantum Technologies, National University of Singapore, 3 Science Drive 2, Singapore 117543\looseness=-1}
\newcommand{\NTU}{Nanyang Quantum Hub, School of Physical and Mathematical Sciences, Nanyang Technological University, Singapore 639673\looseness=-1}

\newcommand{\ihpc}{Institute of High Performance Computing (IHPC), Agency for Science, Technology and Research (A*STAR), 1 Fusionopolis Way, $\#$16-16 Connexis, Singapore 138632, Republic of Singapore}
\newcommand{\qinc}{
Quantum Innovation Centre (Q.InC), Agency for Science Technology and Research (A*STAR), 2 Fusionopolis Way, Innovis $\#$08-03, Singapore 138634, Republic of Singapore }
\newcommand{\infocomm}{Institute for Infocomm Research (I{\texttwosuperior}R), Agency for Science, Technology and Research (A*STAR), 1 Fusionopolis Way, \#21-01, Connexis South Tower, Singapore 138632, Republic of Singapore}

\begin{document}

\normalem
\newlength\figHeight 
\newlength\figWidth 

\title{Contextuality of Quantum Error-Correcting Codes}

\author{Derek Khu}
\email{derek_khu@a-star.edu.sg}
\affiliation{\infocomm}

\author{Andrew Tanggara}
\email{andrew.tanggara@gmail.com}
\affiliation{\CQT}
\affiliation{\NTU}

\author{Chao Jin}
\email{jin_chao@a-star.edu.sg}
\affiliation{\infocomm}

\author{Kishor Bharti}
\email{kishor.bharti1@gmail.com}
\affiliation{\qinc}
\affiliation{\ihpc}

\begin{abstract} 
Fault-tolerant quantum computation requires quantum error correction (QEC), which relies on entanglement to protect information from local noise. Achieving universality, however, demands overcoming the Eastin--Knill theorem. This is often accomplished through strategies like magic state distillation, a process that prepares computational resources---namely, magic states---whose power is now understood to be rooted in quantum contextuality, a fundamental nonclassical feature generalizing Bell nonlocality. Yet, the broader role of contextuality in enabling universality, including its significance as an inherent feature of QEC codes and protocols themselves, has remained largely unexplored.
In this work, we develop a rigorous framework for contextuality in QEC and prove three main results. Fundamentally, we show that subsystem stabilizer codes with two or more gauge qubits are strongly contextual in their partial closure, while others are noncontextual, establishing a clear criterion for identifying contextual codes. Mathematically, we unify Abramsky--Brandenburger’s sheaf-theoretic and Kirby--Love’s tree-based definitions of contextuality, resolving a conjecture of Kim and Abramsky. Practically, we prove that many widely studied code-switching protocols which admit universal transversal gate sets, such as the doubled color codes introduced by Bravyi and Cross, are necessarily strongly contextual in their partial closure.
Collectively, our results establish quantum contextuality as an intrinsic characteristic of fault-tolerant quantum codes and protocols, complementing entanglement and magic as resources for scalable quantum computation. For quantum coding theorists, this provides a new invariant: contextuality classifies which subsystem stabilizer codes can participate in universal fault-tolerant protocols. These findings position contextuality not only as a foundational concept but also as a practical guide for the design and analysis of future QEC architectures.

\end{abstract}

\maketitle

\section{Introduction} \label{sec:intro}
Quantum error correction (QEC) plays a crucial role in realizing fault-tolerant quantum computation by protecting fragile quantum states from noise~\cite{shor1995scheme,gottesman1997stabilizer,dennis2002topological,kitaev2003fault,lidar2013quantum,terhal2015quantum}. This raises a fundamental question: \emph{what quantum resources are required to enable QEC?} One prominent answer is quantum entanglement~\cite{bennett1996mixed,brun2006correcting,dur2007entanglement,bravyi2024much}, and in fact, the interplay between quantum error correction and entanglement theory has led to significant advancements~\cite{kitaev2003fault,anshu2023nlts,bravyi2024much}. Another answer is magic~\cite{bravyi2005universal,veitch2014resource,gidney2024magic}. Magic is intimately connected to the Eastin--Knill theorem~\cite{eastin2009restrictions}, which states the fundamental limitation that no quantum error-correcting code can transversely implement a universal set of gates. Overcoming this no-go theorem requires nontransversal operations, leading to the development of magic state distillation~\cite{bravyi2005universal}, now a cornerstone of fault-tolerant quantum computation. The QEC community has heavily focused on entanglement and magic, but are these the only resources, or should we explore others to deepen our understanding of fault tolerance?

Quantum contextuality~\cite{Koc,Specker1960-bg,budroni2022kochen,klyachko2008simple,raussendorf2013contextuality}, a fundamental nonclassical phenomenon that generalizes Bell nonlocality, is a natural candidate. Contextuality has been celebrated for its role in device certification~\cite{lapkiewicz2011experimental,um2013experimental,jerger2016contextuality,zhan2017experimental,malinowski2018probing,leupold2018sustained,zhang2019experimental,um2020randomness,wang2022significant,hu2023self,liu2023experimental,bharti2019robust,bharti2019local,xu2024certifying,saha2020sum,budroni2022kochen,wang2022significant,hu2023self,liu2023experimental,liu2023exploring,xu2024certifying,arora2024computational}, machine learning~\cite{bowles2023contextuality,gao2022enhancing}, and communication~\cite{hameedi2017communication,saha2019state,gupta2023quantum}, yet its significance for QEC remains unclear. 
Clarifying this relationship could not only deepen our theoretical understanding of fault tolerance but also guide the design of more efficient error-correcting schemes. Key open questions therefore include: What does contextuality mean for QEC codes? Can the diverse definitions of contextuality~\cite{Koc,mermin1990extreme,abramsky2011sheaf,abramsky2015contextuality,kirby2019contextuality} be unified, and how do they apply to QEC? Which codes exhibit contextuality? How does contextuality manifest in fault-tolerant protocols enabling universal computation, such as code-switching?

Our work addresses these gaps by establishing a rigorous framework for contextuality in QEC. Our contributions are threefold: fundamental, mathematical, and practical.

Fundamentally, we formulate a precise, operationally meaningful definition of contextuality for QEC codes and prove a criterion for subsystem qubit stabilizer codes: those with at least two gauge qubits are strongly contextual in their \emph{partial closure} (i.e., when considering the full set of observables whose outcomes can be inferred by multiplying commuting check measurements; see Definition~\ref{def:partial-closure}), whereas others are noncontextual (Corollary~\ref{cor:subsystem_code_contextual_iff_g_ge_2}). This finding provides a clear characterization of contextuality within these codes, establishing a foundational framework for exploring nonclassicality in QEC.

Mathematically, we unify different frameworks of contextuality under partial closure. This includes the partial algebra approach of Kochen and Specker~\cite{Koc}, the ``all-versus-nothing'' arguments as named by Mermin~\cite{mermin1990extreme}, the probabilistic and strong contextuality notions arising from the sheaf-theoretic formalism by Abramsky and Brandenburger ~\cite{abramsky2011sheaf, abramsky2015contextuality}, and the tree-based definition by Kirby and Love~\cite{kirby2019contextuality}. By demonstrating that these seemingly disparate approaches are unified under partial closure, we provide a common language to analyze contextuality across a wide range of physical systems. In doing so, we also resolve a conjecture by Kim and Abramsky~\cite{kim2023stateindependentallversusnothingarguments} (Corollary~\ref{cor:strongly-contextual-implies-siavn-closure}). This unification applies broadly to Pauli measurement sets, including the check measurements used for QEC codes.

Practically, we apply our framework to several of the most important code-switching protocols enabling universal transversal gates in the literature, such as the fault-tolerant code-switching between the $\llbracket 7,1,3 \rrbracket$ Steane code and $\llbracket 15,1,3 \rrbracket$ Reed–Muller code~\cite{butt2024fault}, and the three families of doubled color codes introduced by Bravyi and Cross~\cite{bravyi2015doubled}. All of these protocols exhibit strong contextuality in their partial closure (Corollaries~\ref{cor:code_switching_contextuality} and \ref{cor:doubled_color_code_contextuality}). We further prove under mild assumptions that strong contextuality under partial closure necessarily arises in protocols achieving universality (Theorem~\ref{thm:code-switching-contextual}), suggesting that contextuality underpins fault-tolerant universality in code-switching. Our results position contextuality alongside entanglement and magic as a core resource for scalable quantum computation~\cite{chitambar2019quantum}, providing a crucial new consideration for designing future fault-tolerant QEC systems.

\subsection{Comparison with Prior Art}

Our investigation into contextuality within quantum error-correcting codes is informed by a broad spectrum of related research, particularly studies exploring nonclassicality and its role in quantum computation. Foundational works on quantum nonclassicality, such as the seminal Bell's theorem~\cite{bell1964einstein} and the Kochen--Specker theorem~\cite{Koc}, established that no local hidden-variable model can fully account for quantum predictions. Subsequent work, including the Greenberger--Horne--Zeilinger (GHZ) argument \cite{greenberger1989going, greenberger1990bell} and its refinement by Mermin \cite{mermin1990extreme}, demonstrated ``all-versus-nothing'' (AvN) contradictions in certain quantum correlations, which cannot be explained by local realism. These insights have since been extended through studies that highlight nonlocal features within the structure of quantum codes~\cite{divincenzo1997quantum, cabello2008mermin} and provide detailed analyses of AvN contradictions in stabilizer states~\cite{abramsky2017complete}.

Building on these efforts, Howard \textit{et al.}~\cite{howard2014contextuality} demonstrated in their seminal work that contextuality serves as a critical resource for magic state distillation, a key subroutine in state injection protocols designed to circumvent the limitations imposed by the Eastin--Knill theorem~\cite{eastin2009restrictions} and enable fault-tolerant universal quantum computation. Their approach operates at the \emph{state} level: they showed, using the Cabello--Severini--Winter exclusivity graph framework~\cite{cabello2010noncontextuality,cabello2014graph}, that only contextual \emph{states} can be distilled into useful magic states, whereas noncontextual states cannot.

In contrast, our work takes a \emph{complementary} perspective by focusing on contextuality at the \emph{operator} level. Specifically, our approach leverages the sheaf-theoretic methods from~\cite{abramsky2011sheaf, abramsky2015contextuality}, which formalizes contextuality as the failure of locally consistent operator measurement outcomes to ``glue'' together into a global assignment. We also incorporate the tree-based contextuality formalism introduced by Kirby and Love~\cite{kirby2019contextuality}, which has proven effective for analyzing variational quantum algorithms~\cite{kirby2021contextual, kirby2019contextuality}. It is worth noting the recent work of Abramsky \textit{et al.}~\cite{abramsky2024commutation}, who independently establish a similar link between graph-theoretic properties and contextuality. While their work also leverages the sheaf-theoretic framework, it analyzes a more restrictive class of empirical models where local sections are required to be homomorphisms that respect the partial multiplication, an algebraic constraint not imposed in the general possibilistic models we consider.

This \emph{operator}-level lens is natural for QEC, since codes are fundamentally defined by their check \emph{operators}, and thus contextuality can be analyzed directly at these latter algebraic structures that govern error detection and correction. Using this operator-level approach, we define contextuality and strong contextuality for QEC codes based on their check measurements. We then prove that a wide range of code-switching protocols that enable universal computation via transversal gates must exhibit strong contextuality in the partial closure of the check operators used in the constituent QEC codes. This suggests that contextuality indeed underpins universal fault tolerance at the operator level, offering a novel framework that complements state-based approaches and enriches the understanding of nonclassical resources required for quantum computation.

\subsection{Paper Outline}

The paper is structured as follows. In Section~\ref{sec:prelims}, we review the necessary preliminaries on qubit stabilizer codes and subsystem qubit stabilizer codes (Section~\ref{sec:prelims:quantum_codes}), as well as several formalisms for quantum contextuality: a sheaf-theoretic formulation of contextuality (Section~\ref{sec:prelims:sheaf_contextuality}), a graph-based contextuality (Section~\ref{sec:prelims:graph_contextuality}), and the all-versus-nothing argument (Section~\ref{sec:prelims:avn}). In Section~\ref{sec:QEC_contextuality}, we discuss what it means for a QEC code to exhibit contextuality in an operational sense and give a precise definition based on sheaf-theoretic contextuality. We then present one of our main foundational results: the criterion that links the number of gauge qubits in a subsystem code to its contextuality (Corollary~\ref{cor:subsystem_code_contextual_iff_g_ge_2}). Sections~\ref{sec:kirby_love} and \ref{sec:equivalence_of_contextuality_notions_under_partial_closure} contain the technical proofs that establish the equivalence of various contextuality definitions under partial closure (Corollary~\ref{cor:strong-contextuality-equivalences}), culminating in the resolution of a conjecture by Kim and Abramsky~\cite{kim2023stateindependentallversusnothingarguments} (Corollary~\ref{cor:strongly-contextual-implies-siavn-closure}). In Section~\ref{sec:contextuality_code_switching}, we apply our framework to code-switching protocols, demonstrating through specific examples and a general theorem (Theorem~\ref{thm:code-switching-contextual}) that a wide range of code-switching protocols enabling a universal transversal gate set are necessarily strongly contextual in a partial closure. Finally, in Section~\ref{sec:discussion_open_problems}, we discuss the implications of our findings and suggest directions for future research.

This paper connects concepts from quantum error correction, quantum foundations, abstract algebra and sheaf theory. To aid readers from different backgrounds, we suggest the following reading paths, which are also illustrated in Fig.~\ref{fig:reading_guide}.
\begin{itemize}
    \item \textbf{For the QEC expert interested in the main results}: Focus on the QEC preliminaries (Section~\ref{sec:prelims:quantum_codes}), the definitions and results on contextuality of QEC codes (Section~\ref{sec:QEC_contextuality}), and the application to code-switching (Section~\ref{sec:contextuality_code_switching}). The technical proofs in Sections~\ref{sec:kirby_love} and \ref{sec:equivalence_of_contextuality_notions_under_partial_closure}, which establish the equivalence of different contextuality notions, can be taken as given on a first reading.
    \item \textbf{For the quantum foundations expert:} The various formulations of contextuality are defined in Sections~\ref{sec:prelims:sheaf_contextuality}, \ref{sec:prelims:graph_contextuality}, and \ref{sec:prelims:avn}, while the core technical contributions are in the characterization of the Kirby--Love property (Section~\ref{sec:kirby_love}) and the proof of the equivalence of different contextuality frameworks under partial closure (Section~\ref{sec:equivalence_of_contextuality_notions_under_partial_closure}), which resolves a conjecture from~\cite{kim2023stateindependentallversusnothingarguments}. The QEC sections (Sections~\ref{sec:prelims:quantum_codes}, \ref{sec:QEC_contextuality}, and \ref{sec:contextuality_code_switching}) provide the physical motivation and application for this framework.
    \item \textbf{For a general overview:} The introduction (Section~\ref{sec:intro}) and discussion (Section~\ref{sec:discussion_open_problems}) sections provide the high-level motivation, a summary of results, and our outlook. A brief review of the main QEC result in Section~\ref{sec:QEC_contextuality} (specifically Corollary~\ref{cor:subsystem_code_contextual_iff_g_ge_2}) will provide one of the central takeaways of the paper.
\end{itemize}

\begin{figure}
    \centering
    \begin{tikzpicture}[
        node distance=0.8cm and 0.6cm,
        box/.style={
            rectangle, 
            draw, 
            rounded corners, 
            text centered, 
            text width=8.2em,
            minimum height=2em
        },
        path_box/.style={
            minimum height=4.5em
        },
        line/.style={
            draw, 
            -{Stealth[length=2mm]}
        }
    ]

    \node[box, fill=purple!20, text width=17em] (intro) {Sec I: Introduction};

    \node[box, path_box, fill=blue!20, below left = of intro.south] (qec_prelim) {Sec II A:\\ QEC \\ Preliminaries};
    \node[box, path_box, fill=red!20, below right = of intro.south] (con_prelim) {Sec II B--D:\\ Contextuality Preliminaries};

    \node[box, path_box, fill=blue!20, below = of qec_prelim] (main_result) {Sec III:\\ Main QEC Result \\ (Corollary~1)};
    \node[box, path_box, fill=red!20, below = of con_prelim] (kl_proof) {Sec IV:\\ Kirby--Love Characterization};

    \node[box, path_box, fill=blue!20, below = of main_result] (app_switch) {Sec VI:\\ Application to Code-Switching};
    \node[box, path_box, fill=red!20, below = of kl_proof] (equiv_proof) {Sec V:\\ Equivalence \\ Proofs};

    \coordinate (mid_point) at ($(app_switch.south)!0.5!(equiv_proof.south)$);
    \node[box, fill=purple!20, text width=17em, below = of mid_point] (discuss) {Sec VII: Discussion};

    \path[line] (intro) -- (qec_prelim);
    \path[line] (intro) -- (con_prelim);
    \path[line] (qec_prelim) -- (main_result);
    \path[line] (con_prelim) -- (main_result);
    \path[line] (con_prelim) -- (kl_proof);
    \path[line] (kl_proof) -- (equiv_proof);
    \path[line] (main_result) -- (app_switch);
    \path[line, dashed] (equiv_proof) -- 
        node[pos=0.5, sloped, above, font=\scriptsize] {Technical}
        node[pos=0.5, sloped, below, font=\scriptsize] {Justification}
        (main_result);
    \path[line] (app_switch) -- (discuss);
    \path[line] (equiv_proof) -- (discuss);

    \node[font=\small\bfseries, color=blue!80!black, rotate=90] 
        at ([xshift=-0.4cm]main_result.west) {QEC Path};
    \node[font=\small\bfseries, color=red!90!black, rotate=-90] 
        at ([xshift=0.4cm]kl_proof.east) {Quantum Foundations Path};
    
    \end{tikzpicture}
    \caption{A reading guide for this paper. Readers primarily interested in the QEC results can follow the solid path on the left. Those interested in the formal proofs of contextuality can follow the path on the right. The dashed arrow indicates that the technical results of Section~\ref{sec:equivalence_of_contextuality_notions_under_partial_closure} provide the formal justification for the main QEC result in Section~\ref{sec:contextuality_code_switching}.}
    \label{fig:reading_guide}
\end{figure}

\section{Preliminaries}\label{sec:prelims}

Throughout this paper we focus on codes in the stabilizer formalism.
We denote the $n$-qubit Pauli group by $G_n = \<I,X,Y,Z\>^{\otimes n}$, where $\<a,b,c,\dots\>$ denotes the group generated by group elements $a,b,c,\dots$, and write $X_i$ (resp.\ $Y_i, Z_i$) for the $n$-qubit Pauli operator that is the tensor product of $X$ (resp.\ $Y, Z$) in the $i$th component and $I$ in all other components. (All tensor products are taken over $\C$.)
Also we let $\mathcal{P}_n\subseteq G_n$ be the set of all $n$-qubit Pauli operators with real global phase, i.e., $\mathcal{P}_n = \{\pm 1\} \times \{I, X, Y, Z\}^{\otimes n}$, for convenience when working with stabilizer codes.
Finally, we denote the \textit{weight} of $n$-qubit Pauli operator $p\in G_n$ by $\wt p$, which is the number of nonidentity terms among the tensor factors of $p$.
For example, $p=I\otimes Z\otimes X = Z_2 X_3$ has a weight of $\wt p=2$.

\subsection{Quantum Error-Correcting Codes}\label{sec:prelims:quantum_codes}

In general, a quantum error-correcting code is defined by a subspace $\mathcal{Q}$ of a quantum system described by a $d$-dimensional complex vector space $\C^d$.
A quantum state $|\psi\>$ that belongs to $\mathcal{Q}$ can be recovered after being inflicted by an error $E$ correctable by $\mathcal{Q}$.
This is generally done by first performing a set of check measurements $\mathcal{C}$ on the erroneous state $E|\psi\>$ which gives a set of error syndromes $o=\{o_1,\dots,o_s\}$, followed by a correcting operation $T_o$ depending on the error syndromes to recover the initial quantum state as $T_o E|\psi\> \propto |\psi\>$.

An important and perhaps most studied family of quantum error-correcting codes is the \textit{stabilizer codes}~\cite{gottesman1997stabilizer,gottesman2016surviving}.
A stabilizer code $\mathcal{Q}$ is defined by an abelian subgroup $\mathcal{S}\subseteq\mathcal{P}_n$ of the Pauli group $G_n$, which allows convenient group-theoretic analysis for its properties mainly due to its check measurements $\mathcal{C}$ being those that generate the group $\mathcal{S}$.
In this work, we focus on quantum error-correcting codes within the formalism of stabilizer codes.
Below, we formally define a stabilizer code.

\begin{definition}\label{def:stabilizer_code}
    Let $\mathcal{S}\subseteq\mathcal{P}_n$ be an abelian subgroup of $n$-qubit Pauli group $G_n$, with $-I \not\in \mathcal{S}$.
    The \textit{stabilizer code} $\mathcal{Q}$ fixed by $\mathcal{S}$ is a subspace of $\C^{2^n}$ defined by 
    \begin{equation}
        \mathcal{Q} = \spn\big\{ |\psi\> : \forall p\in\mathcal{S}, \, p|\psi\>=|\psi\> \big\} \;.
    \end{equation}
    In words, $\mathcal{Q}$ consists of quantum states in the $+1$ eigenspace of all $p\in\mathcal{S}$, where $\mathcal{S}$ is known as a \textit{stabilizer group} and an element $s$ of the stabilizer group $\mathcal{S}$ is known as a \textit{stabilizer} of $\mathcal{Q}$.
    A set of check measurements $\mathcal{C}$ of a stabilizer code $\mathcal{Q}$ can be taken as the set of independent generators $p_1,\dots,p_s$ of stabilizer group $\mathcal{S}$.
    An $n$-qubit stabilizer code $\mathcal{Q}$ is an $\llbracket n,k,d\rrbracket$ code whenever\footnote{All instances of $\log$ in this paper refer to the base-$2$ logarithm.}  $k=\log\dim\mathcal{Q}$ and $d=\min_{p\in N(\mathcal{S})\setminus\mathcal{S}}\wt p$, where $N(\mathcal{S})$ is the normalizer of $\mathcal{S}$ in $G_n$.
\end{definition}

Stabilizer codes belong to the family of \textit{subspace} quantum error-correcting codes where the entire subspace $\mathcal{Q}$ is being used to store logical information.
A more general family of codes called the \textit{subsystem codes}~\cite{kribs2005operator,kribs2005unified} offers more flexibility where the subspace $\mathcal{Q}$ of $d$-dimensional complex vector space $\C^d$ is further partitioned into two subsystems: (1) a \textit{code system} $\mathcal{A}$ where logical quantum information is encoded in and (2) a \textit{gauge system} $\mathcal{B}$.
In other words, the codespace $\mathcal{Q}$ is isomorphic to $\mathcal{A}\otimes\mathcal{B}$ and a codestate takes the form of $\ketbra{\psi}_\mathcal{A}\otimes\rho_\mathcal{B} \in\mathcal{L}(\mathcal{Q})$ for some arbitrary density operator $\rho_\mathcal{B}$.
When the subsystem $\mathcal{B}$ is one-dimensional, this gives us a subspace quantum error-correcting code.

In this work we focus on a class of subsystem codes, the \textit{stabilizer subsystem codes}~\cite{poulin2005stabilizer} which uses the stabilizer formalism to construct a subsystem code.
As opposed to the stabilizer code (Definition~\ref{def:stabilizer_code}) which is defined only by the abelian stabilizer group $\mathcal{S}\subseteq\mathcal{P}_n$, a subsystem stabilizer code is defined by both an abelian stabilizer group $\mathcal{S}\subseteq\mathcal{P}_n$ and a \textit{gauge group} $\mathcal{G}\subseteq G_n$ (which is not necessarily abelian) such that $Z(\mathcal{G}) = \< \mathcal{S}, iI \>$, where $Z(\mathcal{G})$ is the center of $\mathcal{G}$. 
From $\mathcal{S}$ and $\mathcal{G}$, we can define the corresponding codespace $\mathcal{Q}$ fixed by $\mathcal{S}$ and a corresponding decomposition $\mathcal{Q}\cong\mathcal{A}\otimes\mathcal{B}$ such that, under this isomorphism, each $g \in \mathcal{G}$ acts as $I_\mathcal{A} \otimes g_\mathcal{B}$, where $I_\mathcal{A}$ is the identity operator on $\mathcal{A}$ and $g_\mathcal{B}$ is some operator on $\mathcal{B}$.
A subsystem stabilizer code is formally defined as follows.

\begin{definition}\label{def:subsystem_code}
    Let $\mathcal{S}\subseteq\mathcal{P}_n$ be an abelian subgroup of $G_n$, with $-I \not\in \mathcal{S}$, and let $\mathcal{G}\subseteq G_n$ be a subgroup of $G_n$ such that $Z(\mathcal{G}) = \< \mathcal{S}, iI \>$.
    The associated \textit{subsystem qubit stabilizer code} $\mathcal{Q}\cong\mathcal{A}\otimes\mathcal{B}$ is a subspace of $\C^{2^n}$ defined by 
    \begin{equation}
        \mathcal{Q} = \spn\big\{ |\varphi\> : \forall p\in \mathcal{S}, \, p|\varphi\>=|\varphi\> \big\} \;,
    \end{equation}
    along with the data of $\mathcal{A}$ and $\mathcal{B}$ as described in \cite{poulin2005stabilizer}.
    The check measurements $\mathcal{C}$ of subsystem stabilizer code $\mathcal{Q}$ form a minimal set of generators $\{p_1,\dots,p_l\}\subseteq \mathcal{P}_n$ that, together with $iI$, generate $\mathcal{G}$.
    The subsystem stabilizer code $\mathcal{Q}$ fixed by $\mathcal{G}$ is an $\llbracket n,k,g,d\rrbracket$ code whenever $g=(\log|\mathcal{G} / \<\mathcal{S},iI\>|)/2$ and $k=n-g-\log |\mathcal{S}|$ and $d=\min_{p\in N(\mathcal{S})\setminus\mathcal{G}}\wt p$. 
\end{definition}

As shown in~\cite{poulin2005stabilizer}, given a subsystem qubit stabilizer code, one can construct $n$ pairs of ``virtual'' Pauli $X$ and $Z$ operators $\{X_j', Z_j'\}_{1 \le j \le n}$ from $G_n$ satisfying:
\begin{enumerate}
\item $[X_j', X_k'] = [Z_j', Z_k'] = [X_j', Z_k'] = 0$ for all distinct $j, k \in \{1, \cdots, n\}$,
\item $X_j' Z_j' = - Z_j' X_j'$ for all $j \in \{1, \cdots, n\}$, and
\item $(X_j')^2 = (Z_j')^2 = I$ for all $j \in \{1, \cdots, n\}$.
\end{enumerate}
The stabilizer group $\mathcal{S}$ is given by $\mathcal{S} = \langle Z_1', \cdots, Z_s' \rangle$ for some $s \ge 0$, and its normalizer in $G_n$ is $N(\mathcal{S}) = \langle iI, Z_1', \cdots, Z_n', X_{s+1}', \cdots, X_n' \rangle$. We can now define the set of check measurements to be $\mathcal{C} = \{Z_1', \cdots, Z_s', X_{s+1}', Z_{s+1}', \cdots, X_{s+g}', Z_{s+g}'\}$ for some $g \ge 0$ with $s+g \le n$. Together with $iI$, these generate the gauge group $\mathcal{G} = \langle \mathcal{C}, iI \rangle$.

\subsection{Sheaf-Theoretic Contextuality Framework}\label{sec:prelims:sheaf_contextuality}

We recall the sheaf-theoretic definition of contextuality from \cite{abramsky2011sheaf, abramsky2015contextuality, kim2023stateindependentallversusnothingarguments} in this section. For the reader who is less familiar with sheaf theory, we include a primer on sheaf-theoretic contextuality in Appendix \ref{sec:primer_sheaf_theory}. A more complete treatment may be found in \cite{abramsky2011sheaf}.

We also refer to the notion of \emph{semirings} in this section. Semirings are algebraic structures similar to rings\footnote{In this paper, the word ``ring'' (resp.\ ``semiring'') always refers to a commutative ring (resp.\ semiring) with unit.}, but without the requirement for an additive inverse. This means that, unlike in a ring, not every element needs a corresponding ``negative'' that results in zero when added. We will mainly be concerned with the semiring of all nonnegative real numbers $\R_{\ge 0}$ and the two-element boolean semiring $\B = \{0,1\}$ with $1+1=1$ (note that, like all semirings, $x+0=0+x=x$ and $x\cdot 0 = 0 \cdot x = 0$ for all $x \in \B$, and $1 \cdot 1 = 1$). For the reader unfamiliar with the concept of semirings, it will be sufficient to think of them as either $\R_{\ge 0}$ or the $\B$ for the purposes of this paper.

\begin{definition} \label{def:quantum_measurement_scenario}
    A \emph{quantum measurement scenario} (also known as a \emph{measurement scenario with a quantum representation} in \cite{abramsky2011sheaf}), is a triple $\langle X, \mathcal{M}, O \rangle$, where $X$ is a nonempty, finite set of observables on some Hilbert space, $\mathcal{M}$ is the set of maximally commuting subsets of $X$, and $O$ is a nonempty, finite set of outcomes for the measurement of each observable. The set $\mathcal{M}$ is also known as the \emph{measurement cover}, and each element in $\mathcal{M}$ is a \emph{measurement context}.
\end{definition}

\begin{definition} \label{def:emp-model-general}
    For a quantum measurement scenario $\langle X, \mathcal{M}, O \rangle$, we define the \emph{event sheaf} $\mathcal{E}$ on $X$ (as a discrete topological space) by taking the sections to be $\mathcal{E}(U) = O^U$ for all $U \subseteq X$ and the restriction maps to be $s \mapsto s|_{U}$ for all $U \subseteq U' \subseteq X$ and $s \in \mathcal{E}(U')$. For a semiring $R$, we can define the functor $\mathcal{D}_R \colon \mathbf{Set} \to \mathbf{Set}$ by setting $\mathcal{D}_R(X)$ to be the set of $R$-distributions on $X$ (i.e., functions $d \colon X \to R$ with $\sum_{x \in X} d(x) = 1$) and, for any map $f\colon X \to Y$, $\mathcal{D}_R(f) \colon \mathcal{D}_R(X) \to \mathcal{D}_R(Y)$ to be
    \[d \mapsto \left(y \mapsto \sum_{x \in f^{-1}(\{y\})} d(x)\right).\]
    A \emph{(no-signaling) $R$-valued empirical model} on $\langle X, \mathcal{M}, O \rangle$ is defined to be a family $e = \{e_C\}_{C \in \mathcal{M}}$ with $e_C \in \mathcal{D}_R\mathcal{E}(C)$ such that $e_C|_{C \cap C'} = e_{C'}|_{C \cap C'}$ for all $C, C' \in \mathcal{M}$. A \emph{possibilistic empirical model} on $\langle X, \mathcal{M}, O \rangle$ is defined to be a subpresheaf $\mathcal{S}$ of $\mathcal{E}$ satisfying the following properties:
    \begin{enumerate} [label={(E\arabic*)}]
        \item \label{def:emp-model-e1} $\mathcal{S}(C) \ne \emptyset$ for all $C \in \mathcal{M}$,
        \item \label{def:emp-model-e2} $\mathcal{S}(U') \to \mathcal{S}(U)$ is surjective for all $C \in \mathcal{M}$ and $U \subseteq U' \subseteq C$, and
        \item \label{def:emp-model-e3} for any  family $\{s_C\}_{C \in \mathcal{M}}$ with $s_C \in \mathcal{S}(C)$ satisfying $s_C|_{C\cap C'} = s_{C'}|_{C \cap C'}$ for all $C, C' \in \mathcal{M}$, there is a global section $s \in \mathcal{S}(X)$ satisfying $s|_{C} = s_C$ for all $C \in \mathcal{M}$. (This section is unique because $\mathcal{E}$ and hence $\mathcal{S}$ are separated.)
    \end{enumerate}
\end{definition}

For $R = \R_{\ge 0}$, $\mathcal{D}_R(X)$ is simply the set of probability distributions on $X$, for any set (of events) $X$. For $R = \B$, $\mathcal{D}_R(X)$ is the set of ``\emph{possibility} distributions'' on $X$, where we assign $1$ to an element of $X$ to indicate it is a ``possible'' event, and $0$ otherwise.

In the case where $R = \B$, there is a bijective correspondence between $\B$-valued empirical models and possibilistic empirical models. Given a $\B$-valued empirical model $\{e_C\}_{C \in \mathcal{M}}$, we obtain a possibilistic empirical model $\mathcal{S}$ defined by setting, for each $U \subseteq X$,
\[\mathcal{S}(U) \coloneqq \{s \in O^U \mid e_C|_{U \cap C}(s|_{U\cap C}) = 1 \text{ for all } C \in \mathcal{M}\}.\]
Conversely, given a possibilistic empirical model $\mathcal{S}$, the family $\{\one_{\mathcal{S}(C)}\}_{C \in \mathcal{M}}$ forms a $\B$-valued empirical model, where we use $\one_{A}$ to denote the indicator function of the set $A$. These two operations are inverse to each other.

Given a unital additive map $f \colon R \to S$ between semirings $R$ and $S$ (i.e., $f$ satisfies $f(1) = 1$ and $f(x+y) = f(x) + f(y)$ for all $x, y \in R$), there is a natural transformation of functors from $\mathcal{D}_R$ to $\mathcal{D}_S$ given by sending each $d \in \mathcal{D}_R(X)$ to $f \circ d \in \mathcal{D}_S(X)$, and thus any $R$-valued empirical model $\{e_C\}_{C \in \mathcal{M}}$ gives rise to an $S$-valued empirical model $\{f \circ e_C\}_{C \in \mathcal{M}}$.

For any semiring $R$, we can define its \emph{support map} $f \colon R \to \B$ by setting $f(r) = 0$ if and only if $r = 0$. This map is unital and additive if and only if $R$ is a nonzero zerosumfree semiring (i.e., a nonzero semiring in which the sum of any two nonzero elements is nonzero). Examples of zerosumfree semirings include $\R_{\ge 0}$ and $\B$. For nonzero zerosumfree semirings $R$, we get a $\B$-valued empirical model, or equivalently a possibilistic empirical model, associated to any $R$-valued empirical model in the manner described above. We formalize this in the definition below.

\begin{definition}
    Let $R$ be a nonzero zerosumfree semiring, and $f \colon R \to \B$ be its support map. The \emph{possibilistic empirical model $\mathcal{S}$ associated to an $R$-valued empirical model} $e = \{e_C\}_{C \in \mathcal{M}}$ on $\langle X, \mathcal{M}, O \rangle$ is the one corresponding to (under the bijective correspondence described earlier) the $\B$-valued empirical model $\{f \circ e_C\}_{C \in \mathcal{M}}$. Equivalently, it is defined by setting, for each $U \subseteq X$,
    \[\mathcal{S}(U) \coloneqq \{s \in O^U \mid s|_{U\cap C} \in \mathrm{supp}(e_C|_{U \cap C}) \text{ for all } C \in \mathcal{M}\}.\]
\end{definition}

The nonzero zerosumfree semiring of most interest to us is $\R_{\ge 0}$. In this case, $\R_{\ge 0}$-distributions on a finite set $U$ are precisely probability distributions on $U$.

\begin{definition}
   Given a quantum measurement scenario $\langle X, \mathcal{M}, O \rangle$ on a Hilbert space $\mathcal{H}$ as well as a quantum state $\rho$ on $\mathcal{H}$, the \emph{empirical model defined by $\rho$} is the $\R_{\ge 0}$-valued empirical model $e$ on $\langle X, \mathcal{M}, O \rangle$ obtained by setting each $e_C$ to be the probability distribution on $O^C$ of the measurement outcomes from making the measurements in $C$ on the state $\rho$. The \emph{possibilistic empirical model defined by $\rho$} is the possibilistic empirical model associated to the empirical model defined by $\rho$.
\end{definition}

\begin{definition}\label{def:quantum_measurement_scenario_contextuality}
    We say a possibilistic empirical model $\mathcal{S}$ on a quantum measurement scenario $\langle X, \mathcal{M}, O \rangle$ is \emph{strongly contextual} if it has no global section, i.e., if $\mathcal{S}(X) = \emptyset$. We say an $R$-valued empirical model $e = \{e_C\}_{C \in \mathcal{M}}$ on $\langle X, \mathcal{M}, O \rangle$ is \emph{contextual} if there is no global section $d \in \mathcal{D}_R\mathcal{E}(X)$ satisfying $d|_C = e_C$ for all $C \in \mathcal{M}$, and \emph{strongly contextual} if the possibilistic empirical model associated to $e$ is strongly contextual. If $e$ is not contextual, we say it is \emph{noncontextual}.
\end{definition}

\begin{remark}
    In \cite{abramsky2011sheaf}, contextuality of an $\R_{\ge 0}$-valued empirical model is known as \emph{probabilistic nonextendability}, while contextuality of a $\B$-valued empirical model is known as \emph{possibilistic nonextendability}. In \cite{abramsky2015contextuality}, contextuality of the $\B$-valued empirical model associated to a possibilistic empirical model is known as \emph{logical contextuality}.
\end{remark}

As observed in \cite[][Sections 4 and 6]{abramsky2011sheaf}, strong  contextuality of an $\R_{\ge 0}$-valued empirical model implies contextuality of its corresponding $\B$-valued empirical model, and this in turn implies contextuality of the $\R_{\ge 0}$-valued empirical model. (It was also noted that the converse to either of these implications may not be true.) More generally, we have the following.

\begin{theorem} \label{thm:strong-contexuality-implies-contextuality}
    Let $R$ be a nonzero zerosumfree semiring. Let $e$ be an $R$-valued empirical model on $\< X, \mathcal{M}, O\>$, with $e'$ its corresponding $\B$-valued empirical model. We then have the following.
    \begin{enumerate}
        \item If $e$ is strongly contextual, $e'$ is contextual.
        \item If $e'$ is contextual, $e$ is also contextual.
    \end{enumerate}
    Therefore, $e$ is contextual if it is strongly contextual.
\end{theorem}
\begin{proof}
    The final statement follows from the first two.

    For the first statement, suppose that $e'$ is not contextual. We need to show that $e$ is not strongly contextual, i.e., the possibilistic empirical model $\mathcal{S}$ associated to $e$ (and equivalently, to $e'$) is not strongly contextual. Since $e' = \{\one_{\mathcal{S}(C)}\}_{C \in \mathcal{M}}$ is not contextual, there is some $d \in \mathcal{D}_{\B} \mathcal{E} (X)$ such that $d|_C = \one_{\mathcal{S}(C)}$ for all $C \in \mathcal{M}$. Since $d \in \mathcal{D}_{\B} \mathcal{E} (X)$, there is some $s \in \mathcal{E}(X)$ for which $d(s) = 1$. It follows that $\one_{\mathcal{S}(C)}(s|_C) = d|_C(s|_C) = 1$ and thus $s|_C \in \mathcal{S}(C)$ for all $C \in \mathcal{M}$. Hence, $s \in \mathcal{S}(X)$ by condition \ref{def:emp-model-e3} and separatedness of $\mathcal{S}$, proving that $\mathcal{S}(X) \ne \emptyset$, as desired.

    The second statement follows immediately from the definition of contextuality and the fact that the map $\mathcal{D}_R \to \mathcal{D}_{\B}$ inducing the map from $e$ to $e'$ is a natural transformation. More precisely, if $f: R \to \B$ is the support map and $e = \{e_C\}_{C \in \mathcal{M}}$ is such that there is some $d \in \mathcal{D}_R \mathcal{E} (X)$ such that $d|_C = e_C$ for all $C \in \mathcal{M}$ (i.e., $e$ is noncontextual), then $f \circ d \in \mathcal{D}_{\B} \mathcal{E} (X)$ satisfies $(f \circ d)|_C = f \circ e_C$ for all $C \in \mathcal{M}$ (so $e'$ is also noncontextual).
\end{proof}

\subsection{Kirby--Love Property}\label{sec:prelims:graph_contextuality}

Following \cite{kirby2019contextuality}, we now define the compatibility graph of a set of observables and what it means for a graph to have the Kirby--Love property.

\begin{definition}
    Given a set $X$ of observables, the \emph{compatibility graph} of $X$ is the graph\footnote{In this paper, the word ``graph'' always refers to a simple, undirected graph.} whose vertices are labeled by the operators in $X$, and where an edge exists between two distinct vertices if and only if the operators labeling those vertices commute.
\end{definition}

\begin{definition} \label{def:kirby-love-property}
    We say a graph has the \emph{Kirby--Love property} if it contains vertices $a,b,c,d$ such that $\{a,b\}$ and $\{a,c\}$ are edges but $\{a,d\}$ and $\{b,c\}$ are not edges (see Fig.\ \ref{fig:kirby-love-property-definition}(a)).
\end{definition}

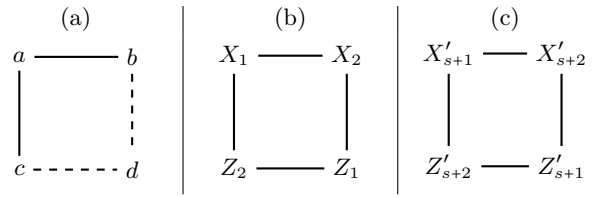
\begin{figure}
    \centering
    \begin{tabular}{>{\centering\arraybackslash}p{2.7cm}|>{\centering\arraybackslash}p{2.7cm}|>{\centering\arraybackslash}p{2.7cm}}
        (a) & (b) & (c) \\
        \begin{tikzpicture}
            \node (a) at (0,1.5) {$a$};
            \node (b) at (1.5,1.5) {$b$};
            \node (d) at (1.5,0) {$d$};
            \node (c) at (0,0) {$c$};
            \draw[thick] (a) -- (b);
            \draw[thick] (a) -- (c);
            \draw[dashed, thick] (c) -- (d);
            \draw[dashed, thick] (b) -- (d);
        \end{tikzpicture}
        & 
        \begin{tikzpicture}
            \node (a) at (0,1.5) {$X_1$};
            \node (b) at (1.5,1.5) {$X_2$};
            \node (d) at (1.5,0) {$Z_1$};
            \node (c) at (0,0) {$Z_2$};
            \draw[thick] (a) -- (b);
            \draw[thick] (a) -- (c);
            \draw[thick] (c) -- (d);
            \draw[thick] (b) -- (d);
        \end{tikzpicture}
        & 
        \begin{tikzpicture}
            \node (a) at (0,1.5) {$X_{s+1}'$};
            \node (b) at (1.5,1.5) {$X_{s+2}'$};
            \node (d) at (1.5,0) {$Z_{s+1}'$};
            \node (c) at (0,0) {$Z_{s+2}'$};
            \draw[thick] (a) -- (b);
            \draw[thick] (a) -- (c);
            \draw[thick] (c) -- (d);
            \draw[thick] (b) -- (d);
        \end{tikzpicture}
    \end{tabular}
    \caption{(a) All $4$-vertex graphs with the Kirby--Love property (up to relabeling of vertices). A solid (resp.\ no) line between two vertices indicates the presence (resp.\ absence) of an edge between the two vertices. A dashed line between two vertices indicates that the edge may or may not be present. In general, a graph has the Kirby--Love property if and only if it contains such a $4$-vertex graph as an induced subgraph. (b) Compatibility graph of $\{X_1, X_2, Z_1, Z_2\} \subseteq \mathcal{P}_2$. This graph has the Kirby--Love property. (c) Subgraph induced by the vertices representing $X_{s+1}', X_{s+2}', Z_{s+1}', Z_{s+2}'$ in the compatibility graph of the set of check measurements $\mathcal{C}$ of a code with $g \ge 2$ gauge qubits, as in the proof of Lemma \ref{lem:comp_graph_kl_iff_at_least_two_gauge_qubits}.}
    \label{fig:kirby-love-property-definition}
\end{figure}

In the language of compatibility graphs, the set $\mathcal{M}$ in a quantum measurement scenario $\langle X, \mathcal{M}, O \rangle$ is in fact the set of all maximal cliques in the compatibility graph of $X$.

\subsection{AvN property of empirical models}\label{sec:prelims:avn}

In \cite{abramsky2015contextuality}, the authors pointed out that strong contextuality can be established for a large family of quantum states using ``all-versus-nothing'' arguments like the one introduced in \cite{mermin1990extreme}, and proposed the following definitions to capture the essence of such arguments.

\begin{definition}\label{def:possibilistic_model_AvN}
    Let $\mathcal{S}$ be a possibilistic empirical model on a quantum measurement scenario $\langle X, \mathcal{M}, R \rangle$ (with event sheaf $\mathcal{E}$), where $R$ is a ring.

    An \emph{$R$-linear equation} is a triple $\phi = \langle C, a, b \rangle$ with $C \in \mathcal{M}$, $a \colon C \to R$ (representing the coefficients) and $b \in R$ (representing the constant), and we write $V_\phi \coloneqq C$. For each $C \in \mathcal{M}$, we say a section $s \in \mathcal{E}(C)$ \emph{satisfies} $\phi = \langle C, a, b \rangle$, written $s \models \phi$, if \[\sum_{m \in C} a(m)s(m) = b.\] The \emph{$R$-linear theory of $\mathcal{S}$} is the set $\mathbb{T}_R(\mathcal{S})$ of all $R$-linear equations $\phi$ such that every $s \in \mathcal{S}(V_\phi)$ satisfies $\phi$. We say a section $s \in \mathcal{S}(X)$ \emph{satisfies} $\mathbb{T}_R(\mathcal{S})$ if $s|_{V_\phi}$ satisfies $\phi$ for all $\phi \in \mathbb{T}_R(\mathcal{S})$, and say $\mathcal{S}$ is \emph{AvN} if $\mathbb{T}_R(\mathcal{S})$ is inconsistent, i.e., there is no section $s \in \mathcal{S}(X)$ that satisfies $\mathbb{T}_R(\mathcal{S})$. We will also say an empirical model is AvN if the associated possibilistic empirical model is AvN.
\end{definition}

The following is proven in \cite[][Proposition 7]{abramsky2015contextuality}.

\begin{fact} \label{fact:avn-implies-strongly-contextual}
    A possibilistic empirical model on a quantum measurement scenario which is AvN is also strongly contextual.
\end{fact}

We now specialize to the quantum measurement scenario $\langle X, \mathcal{M}, R \rangle$ where $R = \Z_2 = \{0, 1\}$ is the ring of integers modulo $2$ and the observables $x$ in $X$ only have eigenvalues $1$ and/or $-1$ (or equivalently, satisfy $x^2 = I$ since $x$ is diagonalizable), corresponding to the measurement outcomes $0$ and $1$ respectively. (In other words, each outcome $b \in \Z_2$ corresponds to the eigenvalue $(-1)^b$.)

Given a multiplicative relation\footnote{It makes sense for $a(x)$ to take values in $\Z_2$ because $x^2 = I$.} of commuting observables \[\prod_{x \in C} x^{a(x)} = (-1)^b I\] for some $C \in \mathcal{M}$, $b \in \Z_2$ and $a\colon C \to \Z_2$, we should expect, after 
measurement of a joint eigenstate of the observables whose eigenvalues are given by $(-1)^{s(x)}$ for $s \colon C \to \Z_2$ (and observing the outcome $s$), that the following relation holds:
\[\prod_{x \in C} (-1)^{s(x)a(x)} = (-1)^b,\]
or equivalently
\[\sum_{x \in C} a(x) s(x) = b.\]
This motivates the following definition of state-independent AvN in \cite{kim2023stateindependentallversusnothingarguments}.

\begin{definition}\label{def:quantum_measurement_schenario_AvN}
    Suppose $\langle X, \mathcal{M}, \Z_2 \rangle$ is a quantum measurement scenario where every observable in $X$ is self-inverse. The \emph{linear theory of $X$} is the set $\mathbb{T}_{\Z_2}(X)$ of all $\phi = \langle C, a, b \rangle$ (for $C \in \mathcal{M}$, $a \colon C \to \Z_2$, $b \in \Z_2$) such that the following relation holds in the space of observables on the Hilbert space:
    \[\prod_{x \in C} x^{a(x)} = (-1)^b I.\]
    We say a section $s \in \mathcal{E}(X)$ \emph{satisfies} $\mathbb{T}_{\Z_2}(X)$ if $s|_{V_\phi}$ satisfies $\phi$ for all $\phi \in \mathbb{T}_{\Z_2}(X)$, and say $X$ is \emph{state-independently AvN} if $\mathbb{T}_{\Z_2}(X)$ is inconsistent, i.e., there is no section $s \in \mathcal{E}(X)$ that satisfies $\mathbb{T}_{\Z_2}(X)$.
\end{definition}

The next result is proven in \cite{kim2023stateindependentallversusnothingarguments} for the case where $X$ is a subset of the set of observables in the Pauli $n$-group (which are self-inverse), but the same proof also applies to the more general result below, so we omit its proof.

\begin{theorem} \label{thm:siavn-implies-avn}
    Suppose $\langle X, \mathcal{M}, \Z_2 \rangle$ is a quantum measurement scenario, where the operators in $X$ have eigenvalues $1$ and/or $-1$ corresponding to the measurement outcomes $0$ and $1$ respectively. If $X$ is state-independently AvN, then any possibilistic empirical model on $\langle X, \mathcal{M}, \Z_2 \rangle$ that is defined by some quantum state is AvN.
\end{theorem}

\section{Contextuality of Quantum Error-Correcting Codes}\label{sec:QEC_contextuality}

The primary operations performed on a quantum system used for a quantum error-correcting code $\mathcal{Q}$ are the check measurements $\mathcal{C}$ whose outcomes give some information about errors to perform an appropriate correction operation.
In the case of a subspace stabilizer code, the check measurements $\mathcal{C}$ consist of repeated measurements of the Pauli observables that generate its stabilizer group $\mathcal{S}$ (see e.g., Section IV of~\cite{dennis2002topological}).
In contrast, for a subsystem stabilizer code, the check measurements $\mathcal{C}$ consist of measurements of the Pauli observables that generate its gauge group $\mathcal{G}$ in a particular sequence such that the eigenvalues of its stabilizers can be inferred from the check outcomes (see e.g.,~\cite{suchara2011constructions,higgott2021subsystem}).

As in other tasks in quantum information and computation, the statistics of the outcomes from the check measurements $\mathcal{C}$ of a quantum error-correcting code is where contextuality can manifest.
Thus we see that a quantum error-correcting code exhibits contextuality whenever its corresponding check measurements $\mathcal{C}$ exhibit contextual behaviors, as defined in the sheaf-theoretic framework of contextuality in Section~\ref{sec:prelims:sheaf_contextuality}.
In particular, we characterize the contextuality of quantum error-correcting codes in the stabilizer formalism (see Section~\ref{sec:prelims:quantum_codes}) by its check measurements $\mathcal{C}$ consisting of Pauli observables.

Our results regarding contextuality of quantum error-correcting codes in this section revolves around the partial closure of the check measurements $\mathcal{C}$ of that code (see Definition \ref{def:partial-closure} in Section~\ref{sec:equivalence_of_contextuality_notions_under_partial_closure}).
We explain why considering the partial closure of check measurements is operationally justified for quantum error-correcting codes in Remark~\ref{rem:partial_closure_motivation}.

We now formally define contextuality and strong contextuality for quantum error-correcting codes whose check measurements lie in $\mathcal{P}_n$.

\begin{definition}\label{def:QECC_contextuality}
    Let $\mathcal{C} \subseteq \mathcal{P}_n$ be the set of check measurements of a quantum error-correcting code, and $\mathcal{M}$ (resp.\ $\bar{\mathcal{M}}$) be the set of maximally commuting subsets of $\mathcal{C}$ (resp.\ $\bar{\mathcal{C}}$, where $\bar{\mathcal{C}}$ is the partial closure of $\mathcal{C}$ in $\mathcal{P}_n$). We say this quantum error-correcting code is \emph{contextual} (resp.\ \emph{contextual in a partial closure}) if the empirical model on $\langle \mathcal{C}, \mathcal{M}, \Z_2 \rangle$ (resp.\ $\langle \bar{\mathcal{C}}, \bar{\mathcal{M}}, \Z_2 \rangle$) defined by \emph{some} quantum state on the system is contextual. We can also say this code is \emph{strongly contextual} (resp.\ \emph{strongly contextual in a partial closure}) if the empirical model on $\langle \mathcal{C}, \mathcal{M}, \Z_2 \rangle$ (resp.\ $\langle \bar{\mathcal{C}}, \bar{\mathcal{M}}, \Z_2 \rangle$) defined by \emph{some} quantum state on the system is strongly contextual.
\end{definition}

\begin{remark}\label{rem:partial_closure_motivation}
    Performing any two commuting Pauli measurements $p$ and $q$ consecutively in any order is equivalent to performing the Pauli measurement $pq$ (or $qp$), in that we can infer the outcome from measuring $pq$ by measuring $p$ and then $q$ (or measuring $q$ and then $p$).
    Hence, if commuting measurements $p$ and $q$ are in the set of checks $\mathcal{C}$, we effectively also obtain the statistics of $pq$ from measuring all checks in $\mathcal{C}$.
    Thus, in determining which observables should be considered in a contextuality scenario for a given code with checks $\mathcal{C}$, it is natural to also include the product of any two commuting Pauli observables.
    This leads to the notion of the \textit{partial closure} $\Bar{\mathcal{C}}$ of check measurements $\mathcal{C}$ which contains all Pauli observables whose values can be inferred from $\mathcal{C}$ (by iteratively taking the product of any pair of commuting observables, starting from those in $\mathcal{C}$).
    The partial closure for Pauli measurements is defined formally and discussed in Definition \ref{def:partial-closure} in Section~\ref{sec:equivalence_of_contextuality_notions_under_partial_closure}.
\end{remark}

Recall from Section~\ref{sec:prelims:quantum_codes} that for an $\llbracket n,k,g,d\rrbracket$ subsystem qubit stabilizer code, its gauge group $\mathcal{G}$, stabilizer group $\mathcal{S}$, and its normalizer $N(\mathcal{S})$ can be described in terms of virtual Pauli operators $\{X_j', Z_j'\}_{1 \le j \le n}$ from $G_n$.
Its check measurements $\mathcal{C}$ can therefore be identified by the independent generators of $\mathcal{G}$, i.e., $\mathcal{C} = \{Z_1', \cdots, Z_s', X_{s+1}', Z_{s+1}', \cdots, X_{s+g}', Z_{s+g}'\}$ for some $g \ge 0$ with $s+g \le n$.

\begin{lemma} \label{lem:comp_graph_kl_iff_at_least_two_gauge_qubits}
The compatibility graph of the set $\mathcal{C} \subseteq \mathcal{P}_n$ of check measurements in the subsystem qubit stabilizer code has the Kirby--Love property if and only if $g \ge 2$.
\end{lemma}
\begin{proof}
If $g=0$, all the operators in $\mathcal{C}$ commute with one another, so the compatibility graph clearly does not have the Kirby--Love property. If $g=1$, there is only one pair of anticommuting operators in $\mathcal{C}$, namely $X_{s+1}'$ and $Z_{s+1}'$, and so the compatibility graph again does not have the Kirby--Love property. If $g \ge 2$, it is straightforward to verify that the condition for Kirby--Love property of the compatibility graph of $\mathcal{C}$ in Definition \ref{def:kirby-love-property} is met by taking $a = X_{s+1}', b = X_{s+2}', c = Z_{s+2}', d = Z_{s+1}'$ (see Fig.\ \ref{fig:kirby-love-property-definition}(c)).
\end{proof}

\begin{theorem} \label{thm:contextuality_in_terms_of_gauge_qubits}
    Let $\mathcal{C} \subseteq \mathcal{P}_n$ be the set of check measurements in the subsystem qubit stabilizer code, and $\bar{\mathcal{M}}$ be the set of maximally commuting subsets of $\bar{\mathcal{C}}$, where $\bar{\mathcal{C}}$ is the partial closure of $\mathcal{C}$ in $\mathcal{P}_n$. Let  $\mathcal{S}_\rho$ be the empirical model on $\langle \bar{\mathcal{C}}, \bar{\mathcal{M}}, \Z_2 \rangle$ defined by some quantum state $\rho$. Then, $\mathcal{S}_\rho$ is noncontextual if $g < 2$, while $\mathcal{S}_\rho$ is strongly contextual if $g \ge 2$.
\end{theorem}
\begin{proof}
    This follows from Lemma \ref{lem:comp_graph_kl_iff_at_least_two_gauge_qubits} and the equivalence between the Kirby--Love property, (probabilistic) contextuality, and strong contextuality in a partial closure, which will be established later in Corollary \ref{cor:strong-contextuality-equivalences}.
\end{proof}

\begin{corollary} \label{cor:subsystem_code_contextual_iff_g_ge_2}
    The subsystem qubit stabilizer code is strongly contextual in a partial closure if and only if there are at least $2$ gauge qubits. It is noncontextual if and only if there are less than $2$ gauge qubits.
\end{corollary}
\begin{proof}
    This simply follows from Theorem \ref{thm:contextuality_in_terms_of_gauge_qubits} and the definition of contextuality of a quantum error-correcting code.
\end{proof}

This criterion establishes a sharp threshold at two gauge qubits for the emergence of strong contextuality in subsystem stabilizer codes under partial closure. Codes with zero or one gauge qubit admit a noncontextual hidden-variable model, allowing consistent classical value assignments to their check operators across all measurement contexts. In contrast, those with two or more gauge qubits exhibit strong contextuality, precluding such assignments and introducing a form of quantum nonclassicality essential for fault-tolerant universal computation, as demonstrated in code-switching protocols (Section~\ref{sec:contextuality_code_switching}).

\begin{example}
Stabilizer codes are noncontextual.
\end{example}
\begin{proof}
This follows from Corollary \ref{cor:subsystem_code_contextual_iff_g_ge_2}, since there are $g = 0$ gauge qubits in stabilizer codes.
\end{proof}

\begin{example} \label{eg:noncontextual_6113_code}
The $\llbracket 6,1,1,3 \rrbracket$ six-qubit error-correcting subsystem code \cite{shaw2008encoding} is noncontextual.
\end{example}

\begin{example}
The measurement set for the $\llbracket 9,1,4,3 \rrbracket$ Bacon--Shor code \cite{eczoo_bacon_shor_9} is strongly contextual in a partial closure.
\end{example}

\begin{example}
    When regarded as a subsystem code, the honeycomb code described in \cite{hastings2021dynamically} is strongly contextual in a partial closure, because there are $n - 1$ gauge qubits where $n \ge 3$ is the number of hexagonal plaquettes.
\end{example}

\section{Characterization of the Kirby--Love property}\label{sec:kirby_love}

We now give an alternative characterization of the Kirby--Love property. (This is related to the property described in \cite[][Theorem 3]{kirby2019contextuality}, but our proof is purely graph-theoretic and does not depend on contextuality arguments.)
We then use this characterization to establish that the graph corresponding to a measurement scenario possesses the Kirby--Love property whenever the empirical model on this scenario is strongly contextual (Theorem~\ref{thm:strongly-contextual-implies-comp-graph-kl}) or even merely contextual (Theorem~\ref{thm:contextual-implies-comp-graph-kl}).

\begin{lemma} \label{lem:disjoint-max-cliques-iff-three-vertex-property}
    Let $G$ be a graph. Then the maximal cliques of $G$ are disjoint if and only if $G$ does not contain vertices $a,b,c$ such that $\{a,b\}$ and $\{a,c\}$ are edges but $\{b,c\}$ is not an edge.
\end{lemma}
\begin{proof}
    Assume first that the maximal cliques of $G$ are disjoint, and let $a,b,c$ be any vertices such that $\{a,b\}$ and $\{a,c\}$ are edges. Then each of $\{a,b\}$ and $\{a,c\}$ belongs to a maximal clique, which must be the same one since they are not disjoint. In particular, $\{b,c\}$ is also an edge. This proves the ``only if'' direction.

    Assume instead that $G$ does not contain  vertices $a,b,c$ such that $\{a,b\}$ and $\{a,c\}$ are edges but $\{b,c\}$ is not an edge. It follows that the neighborhood $N(a)$ of any vertex $a$ in $G$ must form a clique. Consider any two maximal cliques $C_1, C_2$ of $G$ that are not disjoint, i.e., they share some vertex $a$. Then the clique $N(a)$ contains $(C_1 \cup C_2) \setminus \{a\}$, which must also thus be a clique. Then $C_1 \cup C_2$ is a clique containing both $C_1$ and $C_2$. By maximality, $C_1 = C_2$, as desired.
\end{proof}

\begin{lemma} \label{lem:kl-iff-not-disjoint-max-cliques}
    Let $G$ be a graph with vertex set $V$, and $U$ be the set of universal vertices in $G$, i.e., vertices in $G$ which are adjacent to all other vertices. Then $G$ has the Kirby--Love property if and only if the maximal cliques of the induced subgraph of $V \setminus U$ are not disjoint.
\end{lemma}
\begin{proof}
    Assume first that $G$ does not have the Kirby--Love property. We need to show that the maximal cliques of the induced subgraph of $V \setminus U$ are disjoint, or equivalently by Lemma \ref{lem:disjoint-max-cliques-iff-three-vertex-property} that it does not contain vertices $a,b,c$ such that $\{a,b\}$ and $\{a,c\}$ are edges but $\{b,c\}$ is not an edge. Indeed, suppose on the contrary there were three such vertices. Since $a \not\in U$, then there must be some vertex $d$ such that $\{a, d\}$ is not an edge. This means $G$ has the Kirby--Love property, a contradiction.

    Assume instead that $G$ has the Kirby--Love property, i.e., there are vertices $\{a,b,c,d\}$ such that $\{a,b\}$ and $\{a,c\}$ are edges but $\{a,d\}$ and $\{b,c\}$ are not edges. The latter condition implies that $a,b,c \in V \setminus U$. It then follows from Lemma \ref{lem:disjoint-max-cliques-iff-three-vertex-property} that the maximal cliques of the induced subgraph of $V \setminus U$ are not disjoint.
\end{proof}

The following lemma simplifies the calculation of the maximal cliques of a graph.

\begin{lemma} \label{lem:max-cliques-calculation}
    Let $G$ be a graph with vertex set $V$, and $U$ be the set of universal vertices in $G$. Let $\mathcal{M}$ be the set of maximal cliques of the induced subgraph of $V \setminus U$. Then the maximal cliques of $G$ are exactly the sets $U \cup C$ for $C \in \mathcal{M}$.
\end{lemma}
\begin{proof}    
    It is clear that $U \cup C$ is a clique for each $C \in \mathcal{M}$, since each vertex in $U$ is universal and $C$ is a clique. Each such $U \cup C$ is also maximal, because any vertex outside of it must lie in some $C' \in \mathcal{M}$ different from $C$ and so must not be adjacent to some vertex in $C$ by maximality of the clique $C$ in the induced subgraph of $V \setminus U$.

    It remains to be proven that these are all the maximal cliques in $G$. Let $K$ be any maximal clique in $G$. Then $K \setminus U$ is a clique in the induced subgraph of $V \setminus U$, so is contained in some $C$. Thus $K \subseteq U \cup C$, and equality holds by maximality of $K$.
\end{proof}

Finally, the theorem below proves that strong contextuality implies the Kirby--Love property of the compatibility graph.

\begin{theorem} \label{thm:strongly-contextual-implies-comp-graph-kl}
    The compatibility graph of any strongly contextual possibilistic empirical model on a quantum measurement scenario has the Kirby--Love property.
\end{theorem}
\begin{proof}
    Let $\mathcal{S}$ be a possibilistic empirical model on a quantum measurement scenario $\langle X, \mathcal{M}, O \rangle$ whose compatibility graph $G$ does not have the Kirby--Love property. We need to show that $\mathcal{S}$ is not strongly contextual, i.e., $\mathcal{S}(X) \ne \emptyset$.

    Let $V$ be the vertex set of $G$, and $U$ be the set of universal vertices in $G$. Taking $\mathcal{M}'$ to be the (necessarily nonempty\footnote{If $V = U$, then $\mathcal{M}' = \{\emptyset\}$.}) set of maximal cliques of the induced subgraph of $V \setminus U$, we have \[\mathcal{M} = \{U \cup C : C \in \mathcal{M}'\},\] by Lemma \ref{lem:max-cliques-calculation}. As $G$ does not have the Kirby--Love property, Lemma \ref{lem:kl-iff-not-disjoint-max-cliques} says that the sets $C \in \mathcal{M}'$ are disjoint and so the intersection of any two distinct maximal contexts is $U$. (See Fig.\ \ref{fig:graph-structure-not-kirby-love} for a diagram of the structure of $G$.)

    Fix any $C_0 \in \mathcal{M}' \ne \emptyset$. Since $\mathcal{S}(U \cup C_0) \ne \emptyset$ by property \ref{def:emp-model-e1}, we may choose some $s_{C_0} \in \mathcal{S}(U \cup C_0)$. This restricts to $s_{C_0}|_{U} \in \mathcal{S}(U)$. For each $C \in \mathcal{M}$ distinct from $C_0$, we may choose some $s_C \in \mathcal{S}(U \cup C)$ such that $s_C|_{U} = s_{C_0}|_{U}$, by property \ref{def:emp-model-e2} which says that $\mathcal{S}(U \cup C) \to \mathcal{S}(U)$ is surjective. Thus, for any distinct $C, C' \in \mathcal{M}$,
    \[s_C|_{(U \cup C) \cap (U \cup C')} = s_C|_U = s_{C_0}|_U = s_{C'}|_U = s_{C'}|_{(U \cup C) \cap (U \cup C')}.\] The equality between the first and last terms is also (trivially) true if $C = C'$. Property \ref{def:emp-model-e3} then ensures the existence of some section $s \in \mathcal{S}(X)$, as desired.
\end{proof}

\begin{figure}
    \centering
    \begin{tikzpicture}[
            vertex/.style={circle, draw, minimum size=1.1cm}
        ]
            \node[vertex] (u) at (0,0) {$U$};

            \node[vertex] (c0) at (0,1.5) {$C_0$};

            \node[vertex] (c1) at (-1.5,0.75) {$C_1$};
            \node[vertex] (c2) at (-1.5,-0.75) {$C_2$};

            \node at (0,-1) {$\cdots$};

            \node[vertex] (cm-1) at (1.5,-0.75) {$C_{m-1}$};
            \node[vertex] (cm) at (1.5,0.75) {$C_m$};

            \draw[double, line width=0.5mm] (u) -- (c0);
            \draw[double, line width=0.5mm] (u) -- (c1);
            \draw[double, line width=0.5mm] (u) -- (c2);
            \draw[double, line width=0.5mm] (u) -- (cm-1);
            \draw[double, line width=0.5mm] (u) -- (cm);
        \end{tikzpicture}
    \caption{Structure of a graph $G$ that does not have the Kirby--Love property, as in the proof of Theorem \ref{thm:strongly-contextual-implies-comp-graph-kl}. Each circle represents a clique in the graph, while each thick double line connecting two cliques indicates that every vertex in one clique is connected to every vertex in the other (i.e., the subgraph induced by the two cliques is a clique). Here, $U$ is the set of universal vertices in $G$, and the maximal cliques are $U \cup C_i$ for $i=0, \cdots, m$. For $i \ne j$, any vertex in $C_i$ is not connected to any vertex in $C_j$.}
    \label{fig:graph-structure-not-kirby-love}
\end{figure}
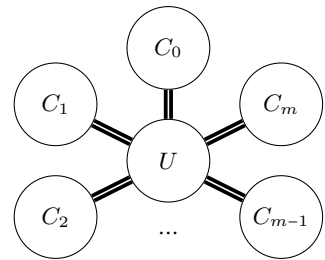

In fact, in the case where $R$ is a zerosumfree semifield\footnote{It is not hard to see that zerosumfree semifields are in fact exactly the semifields which are not fields.} (i.e., a nonzero semiring where all nonzero elements have a multiplicative inverse and the sum of any two nonzero elements is nonzero), we can say more using a similar line of argument. (For the reader unfamiliar with the concept of semirings or semifields, it will be sufficient to consider the case $R = \R_{\ge 0}$ in the proof below.)

\begin{theorem} \label{thm:contextual-implies-comp-graph-kl}
    If $R$ is a zerosumfree semifield, the compatibility graph of any contextual $R$-valued empirical model on a quantum measurement scenario has the Kirby--Love property. In particular, this applies to any contextual empirical model defined by some quantum state $\rho$.
\end{theorem}
\begin{proof}
    The second part of the theorem follows from the first, because the semiring $\R_{\ge 0}$ is a zerosumfree semifield. We now prove the first part of the theorem.
    
    Let $R$ be a zerosumfree semifield and $e = \{e_C\}_{C \in \mathcal{M}}$ be an $R$-valued empirical model on a quantum measurement scenario $\langle X, \mathcal{M}, O \rangle$ whose compatibility graph $G$ does not have the Kirby--Love property. We need to show that $e$ is noncontextual, i.e., that there is some global section $d \in \mathcal{D}_R\mathcal{E}(X)$ satisfying $d|_C = e_C$ for all $C \in \mathcal{M}$.

    As in the proof of Theorem \ref{thm:strongly-contextual-implies-comp-graph-kl}, we take $V$ to be the vertex set of $G$, $U$ to be the set of universal vertices in $G$, and $\mathcal{M}'$ to be the set of maximal cliques of the induced subgraph of $V \setminus U$, giving us \[\mathcal{M} = \{U \cup C' : C' \in \mathcal{M}'\},\] with the sets $C' \in \mathcal{M}'$ disjoint and so the intersection of any two distinct maximal contexts is $U$. For the rest of this proof, for each $C \in \mathcal{M}$, we will write $C'$ to denote the maximal context $C \setminus U$ in $\mathcal{M}'$. For each $x \in O^U$ and $y_C \in O^{C'}$, we will also use $(x,y_C) \in O^C$ to denote the unique element $g \in O^C$ such that $g|_{U} = x$ and $g|_{C'} = y_C$, and $(x, \{y_C\}_{C \in \mathcal{M}})$ to denote the unique element $h \in O^X$ such that $h|_{U} = x$ and $h|_{C'} = y_C$ for all $C \in \mathcal{M}$.

    For \emph{distinct} $C_1, C_2 \in \mathcal{M}$ and $x \in O^U$, the condition that $e_{C_1}|_{C_1 \cap C_2} = e_{C_2}|_{C_1 \cap C_2}$ translates to saying
    \[\sum_{y \in O^{C_1'}} e_{C_1}(x, y) = \sum_{y \in O^{C_2'}} e_{C_2}(x, y).\]
    Thus, for each fixed $x \in O^U$, the sum
    \[\sum_{y \in O^{C'}} e_{C}(x, y)\]
    is the same for all $C \in \mathcal{M}$, and we let this common value be $d_x$. Since each $e_C$ is an $R$-distribution, we have $\sum_{x \in O^U} d_x = 1$.
    
    Now, we can define a function $d$ on $O^X$ by setting, for each $x \in O^U$ and $y_C \in O^{C'}$,
    \[d(x, \{y_{C}\}_{C \in \mathcal{M}}) = d_x^{-(|\mathcal{M}|-1)} \prod_{\substack{C \in \mathcal{M}}} e_C(x, y_C)\]
    if $d_x \ne 0$ (so $d_x$ is invertible, by the assumption that $R$ is a semifield), or $d(x, \{y_{C}\}_{C \in \mathcal{M}}) = 0$ if $d_x = 0$.

    We claim that $d$ is an $R$-distribution on $\mathcal{E}(X)$ such that $d|_C = e_C$ for all $C \in \mathcal{M}$. To this end, we fix some $C \in \mathcal{M}$, $x \in O^U$ and $y_C \in O^{C'}$. Then, if $d_x \ne 0$,
    \begin{align*}
        & \sum_{\substack{K \in \mathcal{M} \setminus \{C\} \\ y_K \in O^{K'}}} d(x,y_C, \{y_{K}\}_{K \in \mathcal{M} \setminus \{C\}}) \\
        = \,\, & \sum_{\substack{K \in \mathcal{M} \setminus \{C\} \\ y_K \in O^{K'}}} \left( d_x^{-(|\mathcal{M}|-1)} \prod_{\substack{L \in \mathcal{M}}} e_L(x, y_L) \right) \\
        = \,\, & d_x^{-(|\mathcal{M}|-1)} e_C(x,y_C) \sum_{\substack{K \in \mathcal{M} \setminus \{C\} \\ y_K \in O^{K'}}} \left(  \prod_{\substack{L \in \mathcal{M} \setminus \{C\} }} e_L(x, y_L) \right) \\
        = \,\, & d_x^{-(|\mathcal{M}|-1)} e_C(x,y_C) \prod_{\substack{K \in \mathcal{M} \setminus \{C\} }} \left(  \sum_{\substack{y_K \in O^{K'}}}  e_K(x, y_K) \right) \\
        = \,\, & d_x^{-(|\mathcal{M}|-1)} e_C(x,y_C) \prod_{\substack{K \in \mathcal{M} \setminus \{C\} }} d_x \\
        = \,\, & e_C(x,y_C).
    \end{align*}
    If $d_x = 0$, this equality
    \begin{equation} \label{eq:pf-thm-contextual-implies-comp-graph-kl}
        \sum_{\substack{K \in \mathcal{M} \setminus \{C\} \\ y_K \in O^{K'}}} d(x,y_C, \{y_{K}\}_{K \in \mathcal{M} \setminus \{C\}}) = e_C(x,y_C)
    \end{equation}
    is still true (with both sides equal to zero), by using the definition of $d$ and also the fact that $R$ is a \emph{zerosumfree} semiring. (The latter implies that $e_C(x, y_C)$, which appears as one of the summands in the definition of $d_x$, must be zero.)

    Summing Equation~\eqref{eq:pf-thm-contextual-implies-comp-graph-kl} over all $x \in O^U$ and $y_C \in O^{C'}$ gives $1$ since $e_C$ is an $R$-distribution, so $d$ is indeed an $R$-distribution. Since Equation~\eqref{eq:pf-thm-contextual-implies-comp-graph-kl} holds for all $x \in O^U$ and $y_C \in O^{C'}$, then we also have $d|_C = e_C$. Thus $d \in \mathcal{D}_R\mathcal{E}(X)$ is a global section that witnesses the noncontextuality of $e$.
\end{proof}

\begin{remark}
    Theorem~\ref{thm:contextual-implies-comp-graph-kl} actually holds for all semifields (including fields) and rings. To see this for general semifields, we may modify the proof above by defining $d(x, \{y_{C}\}_{C \in \mathcal{M}})$ differently if $d_x = 0$. Instead of setting its value to zero, we can fix distinguished elements $y_C^\ast \in O^{C'}$ and define $d(x, \{y_{C}\}_{C \in \mathcal{M}})$ to be:
    \begin{itemize}
        \item $0$, if $y_C \ne y_C^\ast$ for at least two distinct $C \in \mathcal{M}$,
        \item $e_{C_0}(x, y_{C_0})$, if $y_{C_0} \ne y_{C_0}^\ast$ but $y_{C} = y_{C}^\ast$ for all $C \ne C_0$ (for some $C_0 \in \mathcal{M}$), and
        \item $\sum_{C \in \mathcal{M}} e_C(x, y_C^\ast)$, if $y_C = y_C^\ast$ for all $C \in \mathcal{M}$.
    \end{itemize}
    (This definition coincides with the one in the proof above, if $R$ is a zerosumfree semifield.) It is straightforward to check that Equation~\eqref{eq:pf-thm-contextual-implies-comp-graph-kl} continues to hold under this definition.
    
    For general rings, we can modify the proof by using this modified definition of $d(x, \{y_{C}\}_{C \in \mathcal{M}})$ for any $x \in O^U$ (not just the ones with $d_x = 0$), except that we set \[d(x, \{y_{C}\}_{C \in \mathcal{M}}) = \sum_{C \in \mathcal{M}} e_C(x, y_C^\ast) - (|\mathcal{M}|-1)d_x\] (instead of $\sum_{C \in \mathcal{M}} e_C(x, y_C^\ast)$) in the case where $y_C = y_C^\ast$ for all $C \in \mathcal{M}$. Again, a similar check would show that Equation~\eqref{eq:pf-thm-contextual-implies-comp-graph-kl} holds.
\end{remark}

\section{Equivalence of notions of contextuality under partial closure} \label{sec:equivalence_of_contextuality_notions_under_partial_closure}

\begin{figure*}
    \centering
    \begin{tabular}{c|c}
    Empirical model $\mathcal{S}_\rho$ on $\<X,\mathcal{M},\mathbb{Z}_2\>$ & Empirical model $\mathcal{S}_\rho$ on $\<\Bar{X},\Bar{\mathcal{M}},\mathbb{Z}_2\>$ (under partial closure) \\
    \hline
    \begin{tikzcd}
        |[alias=SIavn]| \text{State-independent AvN (Definition~\ref{def:state_independent_AvN_partial_closure})}
        \arrow[Rightarrow, to=dependentavn, "\text{Theorem~\ref{thm:siavn-implies-avn}}"] \\
        |[alias=dependentavn]| \text{State-dependent AvN (Definition~\ref{def:possibilistic_model_AvN})} 
        \arrow[Rightarrow, to=strong, "\text{Fact~\ref{fact:avn-implies-strongly-contextual}}"] \\
        |[alias=strong]| \text{Strongly contextual (Definition~\ref{def:quantum_measurement_scenario_contextuality}}) 
        \arrow[Rightarrow, to=contextual, "\text{Theorem \ref{thm:strong-contexuality-implies-contextuality}}"] \\
        |[alias=contextual]| \text{Contextual (Definition~\ref{def:quantum_measurement_scenario_contextuality})} 
        \arrow[Rightarrow, to=KL, "\text{Theorem~\ref{thm:contextual-implies-comp-graph-kl}, Fact~\ref{fact:kl-contextual-iff-comp-graph-kl}}"] \\
        |[alias=KL]| \text{KL-contextual (Definition~\ref{def:KL_contextual})} 
    \end{tikzcd}
    &
    \begin{tikzcd}
        |[alias=SIavn]| \begin{array}{c}
            \text{State-independent AvN} \\
            \text{(Definition~\ref{def:state_independent_AvN_partial_closure})} 
        \end{array}
        \arrow[Rightarrow, to=dependentavn, "\text{Theorem~\ref{thm:siavn-implies-avn}}"] 
        \arrow[Leftrightarrow, to=KL, "\text{Theorem~\ref{thm:siavn-iff-kl-contextual}}"] 
        \arrow[Leftrightarrow, to=KS, "\text{Theorem~\ref{thm:ks-contextuality-iff-siavn}}"]
        && |[alias=KS]| \begin{array}{c}
        \text{Kochen--Specker} \\
            \text{(Definition~\ref{def:kochen_specker})} 
        \end{array} \\
        |[alias=dependentavn]| \begin{array}{c}
        \text{State-dependent AvN} \\
            \text{(Definition~\ref{def:possibilistic_model_AvN})} 
        \end{array}
        \arrow[Rightarrow, to=strong, "\text{Fact~\ref{fact:avn-implies-strongly-contextual}}"]
        && |[alias=KL]| \begin{array}{c}
            \text{KL-contextual} \\
            \text{(Definition~\ref{def:KL_contextual})} 
        \end{array} \\
        |[alias=strong]| \begin{array}{c}
        \text{Strongly contextual} \\
            \text{(Definition~\ref{def:quantum_measurement_scenario_contextuality})} 
        \end{array}
        \arrow[Rightarrow, to=contextual, "\text{Theorem \ref{thm:strong-contexuality-implies-contextuality}}"] 
        && |[alias=contextual]| \begin{array}{c}
        \text{Contextual} \\
            \text{(Definition~\ref{def:quantum_measurement_scenario_contextuality})} 
        \end{array}
        \arrow[Rightarrow, to=KL, "\text{Theorem~\ref{thm:contextual-implies-comp-graph-kl}, Fact~\ref{fact:kl-contextual-iff-comp-graph-kl}}"] \\
    \end{tikzcd}
    \end{tabular}
    \caption{
    Left: Relationship between different definitions of contextuality for an empirical model $\mathcal{S}_\rho$ on a quantum measurement scenario $\<X,\mathcal{M},\mathbb{Z}_2\>$, for an arbitrary set of observables $X$.
    Right: Different definitions of contextuality are equivalent under partial closure, i.e., are equivalent for an empirical model $\mathcal{S}_\rho$ on a quantum measurement scenario $\<\Bar{X},\Bar{\mathcal{M}},\mathbb{Z}_2\>$, where $\Bar{X}$ is a set of observables under partial closure (see Corollary~\ref{cor:strong-contextuality-equivalences}).
    }
    \label{fig:contextuality_relationships}
\end{figure*}
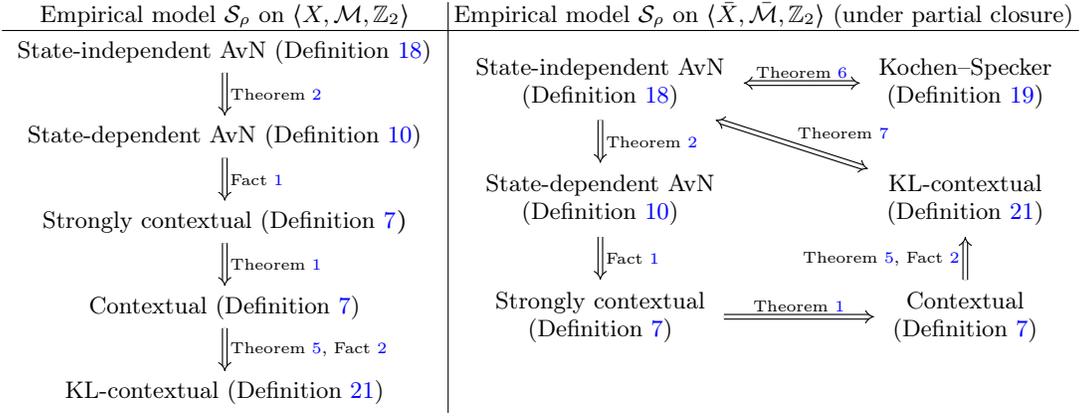

Now we introduce and formalize the notion of a \emph{partial group} in general, and then discuss partial groups in the context of Pauli operators (see Appendix \ref{sec:primer_partial_groups} for an informal overview tailored to this work). As we have discussed, the \emph{partial closure} of a given set of Pauli measurements intuitively corresponds to the largest set of Pauli measurements whose outcomes can be inferred from the given set, and this can be obtained by incrementally taking products of commuting Pauli operators in that set. The resulting set is a \emph{partial group} of Pauli operators, which roughly means that taking the product of any two \emph{commuting} Pauli operators in this set yields a Pauli operator that is still in the set. We will then proceed to consider contextuality scenarios with a set of Pauli measurements in a partial closure and show that the notion of contextuality in this scenario is rather unambiguous.
Particularly, we show in Corollary~\ref{cor:strong-contextuality-equivalences} that the three main notions of contextuality introduced in Section~\ref{sec:prelims}, namely sheaf-theoretic strong contextuality, graph-based Kirby--Love property, and AvN property, are equivalent under partial closure.
To prove the equivalences in Corollary~\ref{cor:strong-contextuality-equivalences}, we use the relationship between the Kirby--Love property and sheaf-theoretic contextuality in Theorem~\ref{thm:contextual-implies-comp-graph-kl} as well as equivalences between other notions of contextuality that have been studied in the literature, namely state-independent AvN property, Kochen--Specker type contextuality, and KL-contextuality, which we show in Section~\ref{sec:equivalence_of_contextuality_notions_under_partial_closure:main}.
The relationship between these different definitions of contextuality in general is summarized in the left-hand side of Fig.~\ref{fig:contextuality_relationships}.
The right-hand side of Fig.~\ref{fig:contextuality_relationships} illustrates the equivalence between all of these definitions whenever the quantum measurement scenario is closed under taking partial closure.

\subsection{Partial Group of Pauli Operators} \label{sec:partial_group_pauli_operators}

For any set $S$, we denote the free monoid on $S$ by $S^*$ and the multiplication (i.e., concatenation) operation on $S^*$ by $\circ$. Observe that we may naturally regard $S^*$ as the union $\bigcup_{i \ge 0} S^i$, and in this way identify elements of $S$ with words of length $1$ in $S^*$. Note also that any involution $f$ on $S$ naturally extends to an involution on $S^*$ given by
\[s_1\circ \cdots \circ s_r \mapsto f(s_r) \circ \cdots \circ f(s_1)\]
for all $s_1, \cdots, s_r \in S$.

We now recall the definition of a partial group introduced in \cite{chermak2013fusion}.

\begin{definition}[\cite{chermak2013fusion}, Definition 2.1] \label{def:partial-group}
    A \emph{partial group} consists of a nonempty set $\mathcal{P}$ together with the data of a set $\mathbf{D} \subseteq \mathcal{P}^*$, a product operation $\Pi \colon \mathbf{D} \to \mathcal{P}$ and an inversion operation $(\blank)^{-1}$ on $\mathcal{P}$ (and its natural extension to $\mathcal{P}^*$) satisfying the following conditions:
    \begin{enumerate}
        \item $\mathcal{P} \subseteq \mathbf{D}$,
        \item $u, v \in \mathbf{D}$ for all $u, v \in \mathcal{P}^*$ such that $u \circ v \in \mathbf{D}$,
        \item $\Pi$ restricts to the identity map on $\mathcal{P}$,
        \item $u \circ (\Pi (v)) \circ w \in \mathbf{D}$ and $\Pi (u \circ v \circ w) = \Pi (u \circ (\Pi v) \circ w) $ for all $u, v, w \in \mathcal{P}^*$ such that $u \circ v \circ w \in \mathbf{D}$, 
        \item $(\blank)^{-1}$ is an involution on $\mathcal{P}$, and
        \item $u^{-1} \circ u \in \mathbf{D}$ and $\Pi (u^{-1} \circ u) = \mathbf{1}$ for all $u \in \mathbf{D}$, where $\mathbf{1}$ denotes the image of the empty word under $\Pi$ (also called the \emph{identity element}).
    \end{enumerate}
    We will call the elements in $\mathbf{D}$ \emph{allowed words} and $\mathbf{D}$ the \emph{set of allowed words}. For $a, b \in \mathcal{P}$ such that $a \circ b \in \mathbf{D}$, we write $a \cdot b = \Pi(a \circ b)$. This means $\cdot$ is a partial binary operation from $\mathcal{P} \times \mathcal{P}$ to $\mathcal{P}$ with domain $(\mathcal{P} \times \mathcal{P}) \cap \mathbf{D}$.
\end{definition}

Informally, a partial group is almost like a group, except the product $\Pi$ is only defined for certain ``allowed words'' $\mathbf{D} \subseteq \mathcal{P}^*$ (finite sequences from $\mathcal{P}$). Many results about groups carry over to partial groups with appropriate modifications. We refer the reader to \cite[][Section 2]{chermak2013fusion} for a list of such results, which we will use without comment for the rest of this paper. As an example, we show that a modified version of associativity holds in partial groups.

\begin{lemma} \label{lem:partial-binary-operation-associative}
    Let $\mathcal{P}$ be a partial group with $\mathbf{D}$ the set of allowed words.
    Then the partial binary operation $\cdot$ on $(\mathcal{P} \times \mathcal{P}) \cap \mathbf{D}$ is associative, i.e.,
    \[(u \cdot v) \cdot w = u \cdot (v \cdot w)\]
    whenever $u, v, w \in \mathcal{P}$ satisfies $u \circ v \circ w \in \mathbf{D}$.
\end{lemma}
\begin{proof}
    This is a consequence of Lemma 2.2(b) of \cite{chermak2013fusion}.
\end{proof}

\begin{definition} \label{def:abelian-partial-group}
    Let $\mathcal{P}$ be a partial group, with  set of allowed words $\mathbf{D}$ and product $\Pi$. We say $\mathcal{P}$ is an \emph{abelian partial group} if $v \circ u \in \mathbf{D}$ and $\Pi(u\circ v) = \Pi(v \circ u)$ for all $u, v \in \mathcal{P}^*$ satisfying $u \circ v \in \mathbf{D}$.
\end{definition}

\begin{remark}
    The notion of an abelian partial group was also introduced in \cite{kim2023stateindependentallversusnothingarguments}. Our definition is slightly different; in particular, we require the existence of inverses (as is usual for the notion of a group) but do not impose reflexivity (i.e., $u \in \mathbf{D}$ need not imply $u \circ u \in \mathbf{D}$) or closure of $\mathbf{D}$ under $\Pi$ (i.e., $a \circ b \in \mathbf{D}$ and $c \in \mathbf{D}$ need not imply $\Pi(a \circ b) \circ c \in \mathbf{D}$).
\end{remark}

\begin{example}[\cite{chermak2013fusion}, Example 2.4(1)] \label{eg:groups-are-partial-groups}
    Given a group $G$, we may form a partial group by setting $\mathcal{P} = G$, $\mathbf{D} = \mathcal{P}^*$, and $\Pi$ and $(\blank)^{-1}$ to be the regular (multivariable) product and inverse operations on $G$. Conversely, given a partial group with set of allowed words $\mathbf{D} = \mathcal{P}^*$, product $\Pi$ and inverse $(\blank)^{-1}$, the set $\mathcal{P}$ with the binary operation $\cdot$ induced from $\Pi$ forms a group with the inverse operation given by $(\blank)^{-1}$.

    These two operations of transforming a group $G$ into a partial group $\mathcal{P}$ with $\mathbf{D} = \mathcal{P}^*$ and vice versa are inverse to each other, so groups (resp.\ abelian groups) are exactly partial groups (resp.\ abelian partial groups) where all words are allowed.
\end{example}

\begin{example} \label{eg:pauli-partial-group}
    An abelian partial group may be constructed by taking:
    \begin{itemize}
        \item $\mathcal{P} = G_n$ to be the Pauli $n$-group as a set (i.e., the set of operators of the form $\alpha P_1 \otimes \cdots \otimes P_n$ with $\alpha \in \{\pm 1, \pm i\}$ and $P_i \in \{I, X, Y, Z\}$),
        \item $\mathbf{D}$ to be the set of concatenations of (zero or more) pairwise commuting operators in $G_n$ (under the regular group operation on $G_n$),
        \item $\Pi$ to be the restriction of the regular (multivariable) group operation on $G_n$ to $\mathbf{D}$, and
        \item $(\blank)^{-1}$ to be the regular inversion operation on $G_n$.
    \end{itemize}
    It is straightforward to check that all the conditions of an abelian partial group holds in this setup.
\end{example}

Following \cite{chermak2013fusion}, we also have the notion of a \emph{partial subgroup}.

\begin{definition}
    Let $\mathcal{P}$ be a partial group, with set of allowed words $\mathbf{D}$, product $\Pi$ and inverse $(\blank)^{-1}$. We say a nonempty subset $\mathcal{Q}$ of $\mathcal{P}$ is a \emph{partial subgroup} of $\mathcal{P}$ if $\mathcal{Q}$ is closed under inversion (i.e., $u \in \mathcal{Q}$ implies $u^{-1} \in \mathcal{Q}$) and closed with respect to products (i.e., $u \in \mathcal{Q}^* \cap \mathbf{D}$ implies $\Pi (u) \in \mathcal{Q}$). (In this case, $\mathcal{Q}$ is a partial group with set of allowed words being $\mathcal{Q}^* \cap \mathbf{D}$ and the product and inverse operations being the restriction of the same operations on $\mathcal{P}$ to $\mathcal{Q}$.)
\end{definition}

\begin{example} \label{eg:pauli-realphase-partial-group}
    Let $\mathcal{P} = G_n$ be Pauli $n$-group viewed as a partial group as described in Example \ref{eg:pauli-partial-group}, and let $\mathcal{P}_n \subseteq G_n$ be the set of $n$-qubit Pauli operators along with a phase of $\pm 1$. It is not hard to see that $\mathcal{P}_n$ is closed under both inversion and (commutative) products, so it is a partial subgroup of $G_n$.
\end{example}

It is clear from the definition of a partial subgroup that the intersection of partial subgroups of a partial group $\mathcal{P}$ is still a partial subgroup of $\mathcal{P}$ and that $\mathcal{P}$ is a partial subgroup of itself. Consequently, the following definition makes sense.

\begin{definition} \label{def:partial-closure}
    Let $\mathcal{P}$ be a partial group and $S \subseteq \mathcal{P}$. The \emph{partial subgroup $\bar{S}$ of $\mathcal{P}$ generated by $S$}, also known as the \emph{partial closure} $\bar{S}$ of $S$ in $\mathcal{P}$, is the smallest partial subgroup of $\mathcal{P}$ that contains $S$. Equivalently, it is the intersection of all partial subgroups of $\mathcal{P}$ containing $S$.
\end{definition}

\begin{remark} \label{remark:constructive-def-partial-closure}
    The abelian partial group $\mathcal{P} = \mathcal{P}_n$ of $n$-qubit Pauli operators with phase $\pm 1$ as described in Example \ref{eg:pauli-realphase-partial-group} (along with the induced partial binary operation) is also an abelian partial group under the definition in \cite{kim2023stateindependentallversusnothingarguments}. In fact, for any subset $S \subseteq \mathcal{P}$, our definition of the partial closure of $S$ in $\mathcal{P}$ coincides with the one in \cite{kim2023stateindependentallversusnothingarguments}, since every element in $\mathcal{P}$ being self-inverse means the closure-under-inversion requirement is satisfied for any subset of $\mathcal{P}$. We can constructively define $\bar{S}$ from $S$ by iteratively constructing a sequence of sets $S_0 \subseteq S_1 \subseteq S_2 \subseteq \cdots$ such that $S_0 = S$ and, for $i \ge 1$, $S_{i}$ is the set of products of (any number of) elements in $S_{i-1}$, i.e., $S_i$ is the image under $\Pi$ of $\mathbf{D} \cap S_{i-1}^*$ (with $\Pi, \mathbf{D}$ as defined in Example \ref{eg:pauli-partial-group}). This sequence must stabilize after at most $|\mathcal{P}_n|$ steps (since each $S_i$ is a subset of $\mathcal{P}_n$), and the set it stabilizes to will be $\bar{S}$. For an example of such a construction, see Appendix \ref{sec:primer_partial_subgroups_closure}.
\end{remark}

\begin{definition}[\cite{chermak2013fusion}, Definition 3.1] \label{def:partial-group-hom}
    Suppose that, for $i = 1,2$, $\mathcal{P}_i$ is a partial group with set of allowed words $\mathbf{D}_i$ and product $\Pi_i$. Let $\beta \colon \mathcal{P}_1 \to \mathcal{P}_2$ be a mapping and $\beta^* \colon \mathcal{P}_1^* \to \mathcal{P}_2^*$ be the induced mapping. Then $\beta$ is a \emph{partial group homomorphism} if
    \begin{enumerate}
        \item $\beta^*(\mathbf{D}_1) \subseteq \mathbf{D}_2$, and
        \item $\beta(\Pi_1(w)) = \Pi_2(\beta^*(w))$ for all $w \in \mathbf{D}_1$.
    \end{enumerate}
\end{definition}

\begin{remark} \label{rem:partial-group-hom-equiv-def}
    In the case where $\mathcal{P}_2$ above is a group (i.e., $\mathbf{D}_2 = \mathcal{P}_2^*$) with binary operation $\cdot$, the first condition is automatically satisfied and we can replace the second condition with
    \[\beta(\Pi_1(u \circ v)) = \beta(u) \cdot \beta(v).\]
    for all $u, v \in \mathcal{P}_1$ such that $u \circ v \in \mathbf{D}_1$. To see this, we can first apply the above relation to $u = v = \mathbf{1}_1$ (identity element of $\mathcal{P}_1$) to get $\beta(\mathbf{1}_1) = \beta(\mathbf{1}_1) \cdot \beta(\mathbf{1}_1)$, giving us $\beta(\mathbf{1}_1) = \mathbf{1}_2$ (identity element of $\mathcal{P}_2$). This establishes the second condition for the case where $w$ is the empty word, and the general case can then be easily established by induction on the length of $w$.
\end{remark}

\subsection{Equivalence Between Notions of Contextuality Under Partial Closure}\label{sec:equivalence_of_contextuality_notions_under_partial_closure:main}

Following \cite{kim2023stateindependentallversusnothingarguments}, we now extend the definition of state-independent AvN for subsets of operators in $\mathcal{P}_n$ to the partial closure, and state the definition of Kochen--Specker type contextuality.

\begin{definition}\label{def:state_independent_AvN_partial_closure}
    A subset $X \subseteq \mathcal{P}_n$ is \emph{state-independently AvN in a partial closure} if $\bar{X}$ is state-independently AvN, i.e., if $\mathbb{T}_{\Z_2}(\bar{X})$ is inconsistent. (Here, $\bar{X}$ is the partial closure of $X$ in $\mathcal{P}_n$.)
\end{definition}

\begin{definition}\label{def:kochen_specker}
    A set $X \subseteq \mathcal{P}_n$ has Kochen--Specker type contextuality if there is no partial group homomorphism $\lambda \colon \bar{X} \to \{\pm 1\}$ that satisfies $\lambda(-I) = -1$ whenever $-I \in \bar{X}$. Here, the partial closure of $X$ is taken in $\mathcal{P}_n$, and $\{\pm 1\}$ has the partial group structure induced from the multiplicative group structure on $\{\pm 1\}$. By Remark \ref{rem:partial-group-hom-equiv-def}, this is equivalent to the nonexistence of any mapping $\lambda \colon \bar{X} \to \{\pm 1\}$ satisfying $\lambda(x \cdot y) = \lambda(x)\lambda(y)$ for all commuting $x, y \in \bar{X}$, and additionally satisfying $\lambda(-I) = -1$ if $-I \in \bar{X}$.
\end{definition}

The following result is stated in \cite{kim2023stateindependentallversusnothingarguments}, although a proof is not given. We provide the proof below.

\begin{theorem}[\cite{kim2023stateindependentallversusnothingarguments}, Corollary 3.10] \label{thm:ks-contextuality-iff-siavn}
    A set $X \subseteq \mathcal{P}_n$ has Kochen--Specker type contextuality if and only if $X$ is state-independently AvN in a partial closure.
\end{theorem}
\begin{proof}
    Suppose first that $X$ does not have Kochen--Specker type contextuality, so there is some partial homomorphism $\lambda \colon \bar{X} \to \{\pm 1\}$ such that $\lambda(-I) = -1$ if $-I \in \bar{X}$. Define $s \colon \bar{X} \to \Z_2$ to be the unique function satisfying $\lambda(x) = (-1)^{s(x)}$. Any equation $\phi \in \mathbb{T}_{\Z_2}(\bar{X})$ is of the form $\phi = \langle C, a, b \rangle$ (for $C \in \bar{\mathcal{M}}$, $a \colon C \to \Z_2$, $b \in \Z_2$), where $\bar{\mathcal{M}}$ is the measurement cover associated to $\bar{X}$ and
    \[\prod_{x \in C} x^{a(x)} = (-1)^b I.\]
    Applying $\lambda$ to this and using the properties of a partial group homomorphism gives us
    \[\prod_{x\in C} \lambda(x)^{a(x)} = (-1)^b,\] or equivalently,
    \[\sum_{x \in C} a(x) s(x) = b,\]
    i.e., $s \models \phi$. As this holds for all $\phi \in \mathbb{T}_{\Z_2}(\bar{X})$, we see that $s$ satisfies $\mathbb{T}_{\Z_2}(\bar{X})$, so $X$ is not state-independently AvN in a partial closure.

    Conversely, suppose that $X$ is not state-independently AvN in a partial closure, so there is some $s \in \mathcal{S}(\bar{X})$ satisfying $\mathbb{T}_{\Z_2}(\bar{X})$. Define $\lambda \colon \bar{X} \to \{\pm 1\}$ by $\lambda(x) = (-1)^{s(x)}$. If $-I \in \bar{X}$, then the equation $-I = (-1)^1 I$ holds in any context containing $-I$ (which must exist), so the equation $s(-I) = 1$ must hold, which means that $\lambda(-I) = -1$. We claim that $\lambda$ is a partial group homomorphism. Take any operators $a,b \in \bar{X}$ which commute, and let $c = a \cdot b$; we need to show that $\lambda (c) = \lambda(a) \lambda(b)$, or equivalently $\lambda(a) \lambda(b) \lambda(c^{-1}) = 1$. Since $a,b,c^{-1}$ clearly commute with each other, they must belong to some common context and we have the equation $a \cdot b \cdot c^{-1} = I$ which holds in the space of observables. As $s \in \mathbb{T}_{\Z_2} (\bar{X})$, we get $s(a) + s(b) + s(c^{-1}) = 0$. It immediately follows that $\lambda(a) \lambda(b) \lambda(c^{-1}) = 1$, as desired. Therefore, $X$ does not have Kochen--Specker type contextuality.
\end{proof}

As noted in \cite{kim2023stateindependentallversusnothingarguments}, there is an alternative definition of contextuality proposed by Kirby and Love in \cite{kirby2019contextuality}.

\begin{definition}[\cite{kirby2019contextuality}, Definitions 1, 2]
    A \emph{determining tree} for $x \in \mathcal{P}_n$ over a set $X \subseteq \mathcal{P}_n$ is a rooted tree whose vertices are elements of $\mathcal{P}_n$ and whose leaves are elements of $X$ such that
    \begin{enumerate}
        \item the root is $x$,
        \item the children of every parent are pairwise commuting operators, and
        \item every parent is the operator product of its children.
    \end{enumerate}
    (It is easy to see that if a determining tree for $x$ over $X$ exists, then $x \in \bar{X}$.)
    
    For a determining tree $\tau$, the \emph{determining set} $D(\tau)$ is the set containing one copy of each operator with odd multiplicity as a leaf in $\tau$.
\end{definition}

\begin{definition}[\cite{kirby2019contextuality}, Definition 3]\label{def:KL_contextual}
    We say $X \subseteq \mathcal{P}_n$ is \emph{KL-contextual} if there exist $x \in \mathcal{P}_n$ and determining trees $\tau_x$ and $\tau_{-x}$ for $x$ and $-x$ over $X$ (respectively) such that $D(\tau_x) = D(\tau_{-x})$.
\end{definition}

\begin{figure}
    \centering
    \begin{tikzpicture}[
        level distance=1cm,
        level 1/.style={sibling distance=2cm}, 
        level 2/.style={sibling distance=1cm}
    ]
        \node (left_root) {$Y_1 Y_2$}
            child {node {$X_1 Z_2$}
            child {node {$X_1$}}
            child {node {$Z_2$}}
            }
            child {node {$Z_1 X_2$}
            child {node {$Z_1$}}
            child {node {$X_2$}}
            };
        
        \node[xshift=4cm] (right_root) {$-Y_1 Y_2$}
            child {node {$X_1 X_2$}
            child {node {$X_1$}}
            child {node {$X_2$}}
            }
            child {node {$Z_1 Z_2$}
            child {node {$Z_1$}}
            child {node {$Z_2$}}
            };
    \end{tikzpicture}
        \caption{Determining trees for $\pm Y_1Y_2$ over $\{X_1, X_2, Z_1, Z_2\}$.}
    \label{fig:determining-tree-example-simple}
\end{figure}
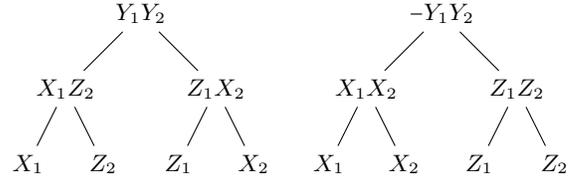

We have the following characterization of KL-contextuality in terms of the Kirby--Love property of the compatibility graph, as proven in \cite{kirby2019contextuality}.

\begin{fact}[\cite{kirby2019contextuality}, Theorem 2] \label{fact:kl-contextual-iff-comp-graph-kl}
    A set $X \subseteq \mathcal{P}_n$ is KL-contextual if and only if its compatibility graph has the Kirby--Love property.
\end{fact}

\begin{example} \label{eg:kl-contextual-example-simple}
    Consider the subset of two-qubit Pauli operators given by $\{X_1, X_2, Z_1, Z_2\} \subseteq \mathcal{P}_2$. Its compatibility graph has the Kirby--Love property (see Fig.\ \ref{fig:kirby-love-property-definition}(b)). By Fact \ref{fact:kl-contextual-iff-comp-graph-kl}, this set of Pauli operators must be KL-contextual. Indeed, we can construct determining trees for both $Y_1Y_2$ and $-Y_1Y_2$ with the same determining set $\{X_1, X_2, Z_1, Z_2\}$, as shown in Fig.\ \ref{fig:determining-tree-example-simple} \cite[Fig.\ 1]{kirby2019contextuality}.
\end{example}

The following theorem is recorded in \cite{kim2023stateindependentallversusnothingarguments}, although the forward direction of the theorem was left as a conjecture. We complete the proof below.

\begin{theorem}[\cite{kim2023stateindependentallversusnothingarguments}, Theorem 3.12] \label{thm:siavn-iff-kl-contextual}
    A set $X \subseteq \mathcal{P}_n$ is state-independently AvN in a partial closure if and only if $X$ is KL-contextual.
\end{theorem}
\begin{proof}
    We first prove the forward direction. Assume that $X$ is state-independently AvN in a partial closure, i.e., $\mathbb{T}_{\Z_2}(\bar{X})$ is inconsistent. Note that all equations in $\mathbb{T}_{\Z_2}(\bar{X})$ are of the form
    \[\sum_{x \in \bar{X}} a(x)s(x) = b,\]
    where each $a(x), b \in \Z_2$ and $s(x)$ are the unknowns which must take values in $\Z_2$. Since $\bar{X}$ is a finite set (say of size $n$) and thus the set of equations in $\mathbb{T}_{\Z_2}(\bar{X})$ is also finite (say of size $m$), we may label the equations $E_1, \cdots, E_m$ and the elements $x_1, \cdots, x_n$, and write the equation $E_i$ as
    \[\sum_{j=1}^n a_{ij} s_j = b_i\]
    for $a_{ij}, b_i \in \Z_2$ and the unknowns $s_j$ (representing $s(x_j)$) taking values in $\Z_2$. This may be represented by the matrix equation $A\mathbf{s} = \mathbf{b}$, with $A = (a_{ij})$ an $m \times n$ matrix over $\Z_2$ (which is a field), $\mathbf{s} = (s_j) \in \Z_2^n$ and $\mathbf{b} = (b_i) \in \Z_2^m$. The inconsistency of this system implies that $\mathbf{b}$ is not in the column space of $A$. The column space of $A$ is the orthogonal complement of the null space of $A^\top$, so there must exist some $\mathbf{y} = (y_i) \in \Z_2^m$ in the null space of $A^\top$ such that $\mathbf{y} \cdot \mathbf{b} \ne 0$, i.e., $\mathbf{y} \cdot \mathbf{b} = 1$.

    Let $K = \{i : y_i = 1\} \subseteq \{1, \cdots, m\}$ be the set of indices $i$ such that the $i$th entry of $\mathbf{y}$ is $1$. As $\mathbf{y}$ is in the null space of $A^\top$, i.e., $\mathbf{y}^\top A = \mathbf{0}$, the sum of the rows of $A$ whose index is in $K$ must be equal to the zero vector. In other words, for each $j \in \{1, \cdots, n \}$, there are an even number of equations $E_i$ (for $i \in K$) where $a_{ij} = 1$. Meanwhile, the equation $\mathbf{y} \cdot \mathbf{b} = 1$ says that $\sum_{i \in K} b_i = 1$. Now, we can use each equation $E_i$ to construct a depth-$1$ determining tree $\tau_i$ for $(-1)^{b_i} I$ over $\bar{X}$, by setting the leaves to be the $x_j$ for which $a_{ij} = 1$: this is a valid determining tree, as equation $E_i$ occurring in $\mathbb{T}_{\Z_2}(\bar{X})$ means \[\prod_{j: a_{ij} = 1} x_j = (-1)^{b_i} I.\] We can then construct a depth-$2$ determining tree of $-I$ over $\bar{X}$, by taking the $|K|$ children of the root $-I$ to be $(-1)^{b_i} I$ over all $i \in K$ and the children of each of these $(-1)^{b_i} I$ to be the $x_j$ such that $a_{ij} = 1$ (i.e., each of the $|K|$ children is the root of the subtree $\tau_i$): this is a valid determining tree because $\sum_{i \in K} b_i = 1$ implies that the product of the $(-1)^{b_i} I$ (over $i \in K$) is indeed $-I$. In this way, we get a determining tree $\sigma$ of $-I$ over $\bar{X}$ where every $x_j \in \bar{X}$ appears as a leaf an even number of times. Finally, we note that every $x_j \in \bar{X}$ in fact has a determining tree $\tau_j'$ over $X$ (by the constructive definition of $\bar{X}$ given in Remark \ref{remark:constructive-def-partial-closure}), so we may extend each leaf $x_j$ of $\sigma$ by the tree $\tau_j'$ to get a determining tree $\sigma'$ of $-I$ over $X$ where each $x \in X$ appears as a leaf of $\sigma'$ with even multiplicity, i.e., $D(\sigma') = \emptyset$. By \cite[][Corollary 3.2]{kirby2019contextuality}, this implies that $X$ is KL-contextual.

    We now establish the converse direction of the theorem. Assume $X$ is KL-contextual. By Theorem \ref{thm:ks-contextuality-iff-siavn}, it suffices to prove that $X$ has Kochen--Specker type contextuality. Suppose on the contrary that this is not true, so there is some partial group homomorphism $\lambda: \bar{X} \to \{\pm 1\}$ such that $\lambda(-I) = -1$ if $-I \in \bar{X}$. Since $X$ is KL-contextual, then there is a determining tree $\tau$ for $-I$ over $\bar{X}$ whose determining set is empty, by \cite[][Corollary 3.2]{kirby2019contextuality}. It follows from $\lambda$ being a partial group homomorphism that the product of $\lambda(y)$ over all leaves $y$ of $\tau$ is equal to $\lambda(-I) = -1$. However, the determining set of $\tau$ is empty, so every $x \in X$ appears as a leaf of $\tau$ an even number of times, which means that the product of the $\lambda(y)$ must be $1$ (considering that $1^2 = (-1)^2 = 1$). We arrive at a contradiction, so it must have been that $X$ has Kochen--Specker contextuality, as desired.
\end{proof}

\begin{corollary} \label{cor:strong-contextuality-equivalences}
    Let $X \subseteq \mathcal{P}_n$, $\bar{X}$ be its partial closure in $\mathcal{P}_n$, and $\bar{\mathcal{M}}$ be the set of maximally commuting subsets of $\bar{X}$. Let $\mathcal{S}_{\rho}$ be the empirical model on $\langle \bar{X}, \bar{\mathcal{M}}, \Z_2 \rangle$ defined by the quantum state $\rho$. The following are equivalent:
    \begin{enumerate}[label=(\arabic*)]
        \item The compatibility graph of $X$ has the Kirby--Love property. \label{cor-item:equiv-comp-graph-kl-regular}
        \item The compatibility graph of $\bar{X}$ has the Kirby--Love property. \label{cor-item:equiv-comp-graph-kl-closure}
        \item $X$ is KL-contextual. \label{cor-item:equiv-kl-contextual-regular}
        \item $\bar{X}$ is KL-contextual. \label{cor-item:equiv-kl-contextual-closure}
        \item $X$ is state-independently AvN in a partial closure. \label{cor-item:equiv-siavn-closure-regular}
        \item $\bar{X}$ is state-independently AvN in a partial closure. \label{cor-item:equiv-siavn-closure-closure}
        \item $X$ has Kochen--Specker contextuality. \label{cor-item:equiv-ks-contextuality-regular}
        \item $\bar{X}$ has Kochen--Specker contextuality. \label{cor-item:equiv-ks-contextuality-closure}
        \item $\mathcal{S}_{\rho}$ is AvN for all quantum states $\rho$. \label{cor-item:equiv-avn-all-rho}
        \item $\mathcal{S}_{\rho}$ is AvN for some quantum state $\rho$. \label{cor-item:equiv-avn-some-rho}
        \item $\mathcal{S}_{\rho}$ is strongly contextual for all quantum states $\rho$. \label{cor-item:equiv-strongly-contextual-all-rho}
        \item \label{cor-item:equiv-strongly-contextual-some-rho} $\mathcal{S}_{\rho}$ is strongly contextual for some quantum state $\rho$. 
        \item $\mathcal{S}_{\rho}$ is contextual for all quantum states $\rho$. \label{cor-item:equiv-contextual-all-rho}
        \item $\mathcal{S}_{\rho}$ is contextual for some quantum state $\rho$. \label{cor-item:equiv-contextual-some-rho}
    \end{enumerate}
\end{corollary}
\begin{proof}
    The following implications are sufficient to prove the corollary:
    \begin{itemize}
        \item {[\ref{cor-item:equiv-comp-graph-kl-regular} $\iff$ \ref{cor-item:equiv-kl-contextual-regular}]} and {[\ref{cor-item:equiv-comp-graph-kl-closure} $\iff$ \ref{cor-item:equiv-kl-contextual-closure}]}: These follow from Fact \ref{fact:kl-contextual-iff-comp-graph-kl}.
        \item {[\ref{cor-item:equiv-kl-contextual-regular} $\iff$ \ref{cor-item:equiv-siavn-closure-regular}]} and {[\ref{cor-item:equiv-kl-contextual-closure} $\iff$ \ref{cor-item:equiv-siavn-closure-closure}]}: These follow from Theorem \ref{thm:siavn-iff-kl-contextual}.
        \item {[\ref{cor-item:equiv-siavn-closure-regular} $\iff$ \ref{cor-item:equiv-ks-contextuality-regular}]} and {[\ref{cor-item:equiv-siavn-closure-closure} $\iff$ \ref{cor-item:equiv-ks-contextuality-closure}]}: These follow from Theorem \ref{thm:ks-contextuality-iff-siavn}.
        \item \ref{cor-item:equiv-siavn-closure-regular} $\iff$ \ref{cor-item:equiv-siavn-closure-closure}: This is clear from the definitions, since the partial closure of $\bar{X}$ is $\bar{X}$.
        \item \ref{cor-item:equiv-siavn-closure-regular} $\implies$ \ref{cor-item:equiv-avn-all-rho}: This follows from Theorem \ref{thm:siavn-implies-avn}.
        \item {[\ref{cor-item:equiv-avn-all-rho} $\implies$ \ref{cor-item:equiv-strongly-contextual-all-rho}]} and {[\ref{cor-item:equiv-avn-some-rho} $\implies$ \ref{cor-item:equiv-strongly-contextual-some-rho}]}: These follow from Fact \ref{fact:avn-implies-strongly-contextual}.
        \item {[\ref{cor-item:equiv-strongly-contextual-all-rho} $\implies$ \ref{cor-item:equiv-contextual-all-rho}]} and {[\ref{cor-item:equiv-strongly-contextual-some-rho} $\implies$ \ref{cor-item:equiv-contextual-some-rho}]}: These follow from Theorem \ref{thm:strong-contexuality-implies-contextuality}.
        \item {[\ref{cor-item:equiv-avn-all-rho} $\implies$ \ref{cor-item:equiv-avn-some-rho}]} and {[\ref{cor-item:equiv-strongly-contextual-all-rho} $\implies$ \ref{cor-item:equiv-strongly-contextual-some-rho}]} and {[\ref{cor-item:equiv-contextual-all-rho} $\implies$ \ref{cor-item:equiv-contextual-some-rho}]}: These are obvious.
        \item \ref{cor-item:equiv-contextual-some-rho} $\implies$ \ref{cor-item:equiv-comp-graph-kl-regular}: This follows from Theorem \ref{thm:contextual-implies-comp-graph-kl}.
    \end{itemize}
\end{proof}

It is conjectured in \cite[][Conjecture 3.14]{kim2023stateindependentallversusnothingarguments} that any measurement cover $\mathcal{M}$ that realizes a contextual empirical model for some quantum state is state-independently AvN in a partial closure. We resolve this in the corollary below.

\begin{corollary} \label{cor:strongly-contextual-implies-siavn-closure}
    Let $X \subseteq \mathcal{P}_n$, and $\mathcal{M}$ be the set of maximally commuting subsets of $X$. Let $\mathcal{S}$ be an empirical model on $\langle X, \mathcal{M}, \Z_2 \rangle$ defined by some quantum state. If $\mathcal{S}$ is contextual, then $X$ is state-independently AvN in a partial closure.
\end{corollary}
\begin{proof}
    Suppose $\mathcal{S}$ is contextual. By Theorem \ref{thm:contextual-implies-comp-graph-kl}, the compatibility graph of $X$ has the Kirby--Love property. By the equivalence of \ref{cor-item:equiv-comp-graph-kl-regular} and \ref{cor-item:equiv-siavn-closure-regular} in Corollary \ref{cor:strong-contextuality-equivalences}, we see that $X$ is state-independently AvN in a partial closure.
\end{proof}

\section{Contextuality in Code-Switching Protocols}\label{sec:contextuality_code_switching}

Code-switching is a protocol that enables the transfer of logical information encoded in one stabilizer code $\mathcal{Q}_1$, with stabilizer group $\mathcal{S}_1$, into another stabilizer code $\mathcal{Q}_2$, with stabilizer group $\mathcal{S}_2$, and vice versa. One can transition between the two codes by performing specific measurements, followed by unitary correction operations based on their outcomes. This protocol can be viewed as a single subsystem code, where the stabilizer group is $\mathcal{S}_1 \cap \mathcal{S}_2$ and the gauge group $\mathcal{G}$ contains both $\mathcal{S}_1$ and $\mathcal{S}_2$ as subgroups.

In this section, we will first recall the formalism of code-switching protocols as presented in~\cite{butt2024fault}, and formally define what it means for a code-switching protocol to be contextual. We then examine a few examples of code-switching protocols from~\cite{anderson2014fault,butt2024fault,bravyi2015doubled} and show that, by analyzing the structure of their underlying gauge groups, these protocols exhibit strong contextuality in a partial closure (Corollary~\ref{cor:code_switching_contextuality}, Remark~\ref{rem:RM_code_switching}, Remark~\ref{rem:10qubit_code_switching}, Corollary~\ref{cor:doubled_color_code_contextuality}). Notably, in each of these protocols, the two stabilizer codes, $\mathcal{Q}_1$ and $\mathcal{Q}_2$, together admit a universal transversal gate set. Due to the strong contextuality of these protocols in a partial closure, we see this as an evidence of how contextuality may be a valuable resource for achieving universal quantum computation via code-switching and transversal gates. This motivates the central question we investigate in this section: must any code-switching protocol that achieves a universal transversal gate set be strongly contextual in a partial closure? As we will show, it is likely that the answer is yes. To support this conjecture, we prove in Theorem~\ref{thm:code-switching-contextual} that, under certain mild assumptions on the two stabilizer codes, it must be that the subsystem code associated to the code-switching protocol has at least $3$ gauge qubits and thus is strongly contextual in a partial closure.

\subsection{Code-switching protocols} \label{sec:code_switching_intro}

As discussed in~\cite{butt2024fault}, a code-switching protocol between two stabilizer codes is possible if both stabilizer codes can be obtained from the same subsystem stabilizer code by fixing appropriate gauge qubits. We recall the formalism of this protocol and slightly generalize it to allow for code-switching between two \emph{subsystem} stabilizer codes, as seen in~\cite{bravyi2015doubled}.

Suppose $\mathcal{Q}$ is an $\llbracket n,k,g,d \rrbracket$ subsystem stabilizer code. There are $n$ pairs of ``virtual'' Pauli $X$ and $Z$ operators $\{X_j', Z_j'\}_{1 \le j \le n}$ from $\mathcal{P}_n$ (satisfying the usual Pauli relations between them) such that the stabilizer group is $\mathcal{S} = \langle Z_1', \cdots, Z_s' \rangle$ for some $s \ge 0$ and the gauge group is $\mathcal{G} = \langle iI, \mathcal{S}, X_{s+1}', Z_{s+1}', \cdots, X_{s+g}', Z_{s+g}'\rangle$. We can form subsystem codes $\mathcal{Q}_1$ and $\mathcal{Q}_2$ by fixing a subset of the gauge operators that are not also stabilizers. Formally, we choose subsets of indices $I_1, I_2 \subseteq \{s+1, \cdots, s+g\}$. We can then define a subsystem stabilizer code $\mathcal{Q}_1$ with stabilizer group $\mathcal{S}_1 = \langle \mathcal{S}, Z_j' :  j \in I_1\rangle$ and gauge group $\mathcal{G}_1 = \langle iI, \mathcal{S}_1, X_j', Z_j' :  j \in \{s+1, \cdots, s+g\} \setminus I_1\rangle$. We can also similarly define a subsystem stabilizer code $\mathcal{Q}_2$ with stabilizer group $\mathcal{S}_2 = \langle \mathcal{S}, X_j' :  j \in I_2\rangle$ and gauge group $\mathcal{G}_2 = \langle iI, \mathcal{S}_2, X_j', Z_j' :  j \in \{s+1, \cdots, s+g\} \setminus I_2\rangle$. (For code-switching between \emph{stabilizer codes} $\mathcal{Q}_1$ and $\mathcal{Q}_2$, $I_1 = I_2 = \{s+1, \cdots, s+g\}$.) Code-switching from $\mathcal{Q}_1$ to $\mathcal{Q}_2$ can be carried out by measuring the stabilizers $X_j'$ for $j \in I_2$ and performing any corrections if necessary. Code-switching from $\mathcal{Q}_2$ to $\mathcal{Q}_1$ can be performed in similar way, except we measure $Z_j'$ for $j \in I_1$ instead. This defines a code-switching protocol between subsystem stabilizer codes $\mathcal{Q}_1$ and $\mathcal{Q}_2$, which are both derived from the same parent subsystem code $\mathcal{Q}$. Note that $\mathcal{S} = \mathcal{S}_1 \cap \mathcal{S}_2$.

It now makes sense to discuss contextuality of a code-switching protocol under the framework we have developed.

\begin{definition} \label{def:code-switching-contextual}
    Consider a code-switching protocol between subsystem stabilizer codes $\mathcal{Q}_1$ and $\mathcal{Q}_2$, both of which are derived from the same parent subsystem code $\mathcal{Q}$. We say this code-switching protocol is \emph{contextual} (resp.\ \emph{contextual in a partial closure}, resp.\ \emph{strongly contextual}, resp.\ \emph{strongly contextual in a partial closure}) if the underlying quantum error-correcting code $\mathcal{Q}$ has the same property.    
\end{definition}

\subsection{Example 1: Code-Switching Between $\llbracket7,1,3\rrbracket$ and $\llbracket15,1,3\rrbracket$ Color Code} \label{sec:code_switching_examples}

We first study the code-switching protocol between the (extended) $\llbracket7,1,3\rrbracket$ color code and $\llbracket15,1,3\rrbracket$ color code, as presented in~\cite[Section IV.A]{butt2024fault}, and show that it is strongly contextual in a partial closure (Corollary~\ref{cor:code_switching_contextuality}). To analyze this, we consider the $15$-qubit subsystem code in which code-switching between these two codes takes place.
In particular, we show that this subsystem code has $3$ gauge qubits (Lemma~\ref{lem:gauge_qubits_of_code_switching}), which implies that it is strongly contextual in a partial closure (by Corollary~\ref{cor:subsystem_code_contextual_iff_g_ge_2}).

This protocol involves a 15-qubit subsystem code with gauge group $\mathcal{G}$ and stabilizer group $\mathcal{S}$. 
These qubits are arranged on the vertices of a partitioned tetrahedron as illustrated in Fig.~\ref{fig:code_switching}.
Fourteen of the 15 qubits are placed on the surface of a tetrahedron: one on each of its four vertices (qubits 0, 4, 6, 10), one on the midpoint of each of the six edges connecting the vertices (qubits 1, 3, 5, 7, 8, 9, 12), and one on the middle of each of its four faces (qubits 2, 11, 12, 14).
The last qubit is placed on the middle of the interior of the tetrahedron (qubit 13).
The tetrahedron is then partitioned into four three-dimensional cells (hexahedra), each of which has eight vertices and one of four distinct colors assigned to it: red, green, blue, or yellow (e.g., the red cell has vertices 0, 1, 2, 3, 7, 12, 13, 14). The $\llbracket 15,1,3 \rrbracket$ code involves all the $15$ qubits. While the $\llbracket 7,1,3 \rrbracket$ color code is supported only on qubits 0--6, we may extend it to a $\llbracket 15,1,3 \rrbracket$ code by including qubits 7--14 as eight additional physical qubits encoding zero logical qubits. The system involving these latter eight qubits is referred to as the \emph{bulk} in~\cite{butt2024fault}. Each weight-4 $X$- or $Z$-operator whose support is a face of the yellow cell is a stabilizer of the bulk and of the extended $\llbracket 7,1,3 \rrbracket$ color code. These operators generate the stabilizer group of the bulk and we will call these the \emph{bulk stabilizers}.

\begin{figure}
    \centering
    \includegraphics[width=1\columnwidth]{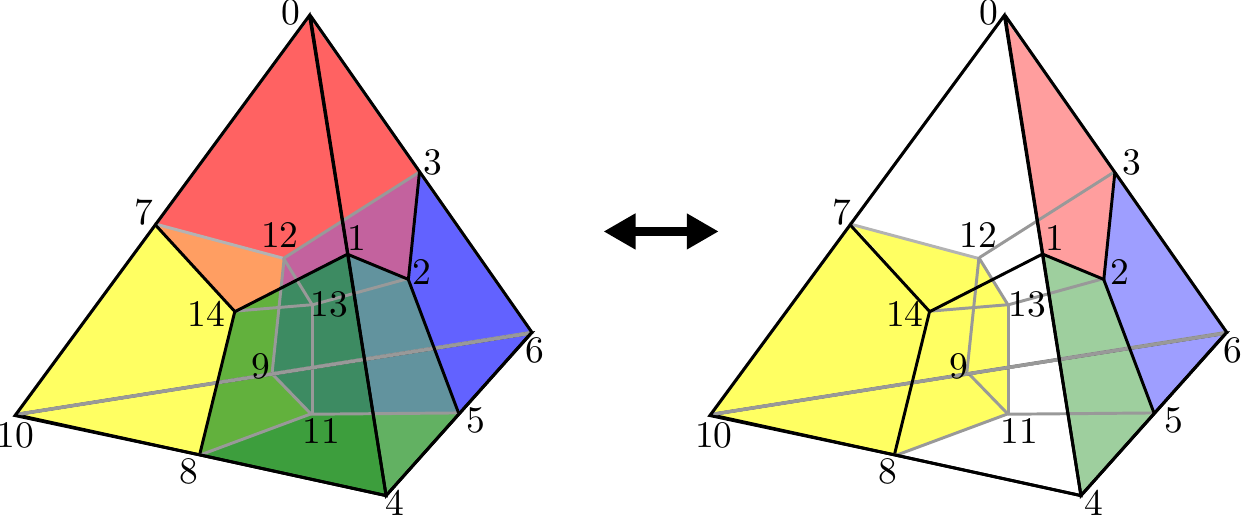}
    \caption{Code-switching protocol between $\llbracket7,1,3\rrbracket$ code and $\llbracket15,1,3\rrbracket$ code using 15 qubits on the vertices of partitioned tetrahedron, which we show to be contextual (Corollary~\ref{cor:code_switching_contextuality}).
    The $\llbracket7,1,3\rrbracket$ code is supported only on the 7 qubits on the right-hand side face of the tetrahedron (qubits 0--6), which are decoupled from the other 8 qubits (qubits 7--14) after switching from the $\llbracket15,1,3\rrbracket$ code.
    Switching back from the $\llbracket7,1,3\rrbracket$ code to the $\llbracket15,1,3\rrbracket$ code couples back qubits 0--6 with the other 8 qubits.}
    \label{fig:code_switching}
\end{figure}

The stabilizer group of the extended $\llbracket7,1,3\rrbracket$ code is generated by the bulk stabilizers as well as the red, green, and blue weight-4 $X$- and $Z$-stabilizers on the surface with qubits 0--6 in its support (which we call the \emph{Steane surface}). (Note that we can replace the latter with weight-8 $X$- and $Z$-stabilizers acting on the red, green, and blue cells, due to the presence of the bulk stabilizers in the generating set.)
The stabilizer group of the $\llbracket15,1,3\rrbracket$ code is generated by the weight-8 $X$- and $Z$-operators acting on the qubits of each of the four cells, together with all weight-4 $Z$-stabilizers on all the faces of the tetrahedron. This includes faces on the surface of the tetrahedron (e.g., $Z_0Z_7Z_{13}Z_{14}$) and faces in the interior that intersect two cells, labeled by $RG, RB, RY, GB, GY, BY$ (e.g., $Z_{RG} = Z_1Z_2Z_{13}Z_{14}$).

Switching from the $\llbracket15,1,3\rrbracket$ code to the $\llbracket7,1,3\rrbracket$ code involves measuring the red, green, and blue weight-4 $X$-stabilizers on the Steane surface, i.e., $X_R \coloneqq X_0X_1X_2X_3, X_G \coloneqq X_1X_2X_4X_5, X_B \coloneqq X_2X_3X_5X_6$.
Conditioned on the measurement outcomes, we perform a correction operation by applying zero or more 4-qubit $Z$ unitaries on the intersecting faces connecting the yellow cell and the surface, i.e., $Z_{RG} \coloneqq Z_1Z_2Z_{13}Z_{14}, Z_{RB} \coloneqq Z_2Z_3Z_{12}Z_{13}, Z_{GB} \coloneqq Z_2Z_5Z_{11}Z_{13}$.
For colors $\{c,c'c''\}=\{R,G,B\}$, a correction unitary on the face that intersects color-$c$ cell and color-$c'$ cell is being applied whenever measurement $X_{c''}$ gives a $-1$ outcome, e.g., $Z_{RG}$ is applied if measurement outcome of $X_B$ is $-1$.
On the other hand, switching from the $\llbracket7,1,3\rrbracket$ code to the $\llbracket15,1,3\rrbracket$ code involves measuring $Z_{RG}, Z_{RB}, Z_{GB}$.
Zero or more correction unitaries among $X_R, X_G , X_B$ are then applied according the following rule: apply $X_{c''}$ if the measurement outcome of $Z_{cc'}$ is $-1$, for colors $\{c,c'c''\}=\{R,G,B\}$.

\begin{lemma}\label{lem:gauge_qubits_of_code_switching}
    The subsystem code for the code-switching protocol between the $\llbracket7,1,3\rrbracket$ code and the $\llbracket15,1,3\rrbracket$ code (as described above) has $g=3$ gauge qubits.
\end{lemma}
\begin{proof}
    The stabilizer group $\mathcal{S}$ of the subsystem code, which is the intersection of the stabilizer groups of the extended $\llbracket7,1,3\rrbracket$ code and the $\llbracket15,1,3\rrbracket$ code, is generated by the four weight-8 $X$-type and four weight-8 $Z$-type operators acting on each cell as well as the $Z$-type bulk stabilizers. The $Z$-type bulk stabilizers are not independent. Any of the three pairs of face operators on opposing faces of the yellow cell (namely, the pair $Z_{RY}, Z_8Z_9Z_{10}Z_{11}$, the pair $Z_{GY}, Z_7Z_9Z_{10}Z_{12}$, and the pair $Z_{BY},Z_7Z_8Z_{10}Z_{14}$) has both operators supported on disjoint qubits, and their product is the $Z$-type weight-$8$ stabilizer acting on the entire cell. This means that the $Z$-type bulk stabilizers may be generated by $Z_{RY}, Z_{GY}, Z_{BY}$ along with the weight $Z$-type weight-$8$ stabilizer acting on the entire cell. Thus, $\mathcal{S}$ may be generated by $11$ independent generators, namely the four $X$-type and four $Z$-type weight-$8$ operators acting on the qubits of each cell, as well as the three $Z$-type face operators $Z_{RY}, Z_{GY}, Z_{BY}$.

    The gauge group $\mathcal{G}$ is generated, up to phase (i.e., with the addition of $iI$), by the stabilizers for each of the two codes. This means it is generated up to phase by the weight-8 $X$- and $Z$-operators on each cell, the weight-4 $Z$-type operators on all the faces of all the cells, and the weight-4 $X$-type operators on the Steane surface and on the bulk. As in the preceding paragraph, these gauge operators are not independent (even if we restrict to considering only the $Z$- or $X$-type ones to avoid anticommutativity issues). Following the same argument as before, the $Z$-type gauge operators listed above may be independently generated by the four weight-8 $Z$-operators on each cell, together with the six weight-4 $Z$-operators on faces in the interior of the tetrahedron (namely, $Z_{RG}, Z_{RB}, Z_{RY}, Z_{GB}, Z_{GY}, Z_{BY}$). The $X$-type gauge operators may be independently generated by the four weight-8 $X$-operators on each cell, together with the three weight-4 $X$-operators on faces in the interior of the tetrahedron that are adjacent to the yellow cell (namely, $X_{RY}, X_{GY}, X_{BY}$).

    Thus, the minimal number of Pauli operators needed to generate $\mathcal{G}$ up to phase is $17$, while that needed to generate $\mathcal{S}$ is $11$. Hence, the number of gauge qubits is $g = \frac{1}{2}(17-11)=3$. (The three pairs of anticommuting gauge operators are in fact $(Z_{RG}, X_{BY})$, $(Z_{RB}, X_{GY})$, and $(Z_{GB}, X_{RY})$.)

\end{proof}

\begin{corollary}\label{cor:code_switching_contextuality}
    The code-switching protocol between the $\llbracket7,1,3\rrbracket$ code and the $\llbracket15,1,3\rrbracket$ code is strongly contextual in a partial closure.
\end{corollary}
\begin{proof}
    This follows from Lemma~\ref{lem:gauge_qubits_of_code_switching} and Corollary~\ref{cor:subsystem_code_contextual_iff_g_ge_2} (and Definition~\ref{def:code-switching-contextual}).
\end{proof}

\begin{remark}\label{rem:RM_code_switching}
    The $\llbracket7,1,3\rrbracket$ code and the $\llbracket15,1,3\rrbracket$ code discussed in this section are examples of quantum Reed--Muller codes. The fault-tolerant conversion between quantum Reed--Muller codes is discussed in \cite{anderson2014fault}. Specifically, they described a method to fault-tolerantly convert between two quantum Reed--Muller codes $\mathrm{QRM}(m)$ and $\mathrm{QRM}(m+1)$ (for $m \ge 3$), and explained that their protocol can be described in terms of a subsystem code with $m$ gauge qubits. Since $m \ge 3$, their code-switching protocol is also always contextual in a partial closure by Corollary~\ref{cor:subsystem_code_contextual_iff_g_ge_2}. For $m=3$, this coincides with the switching between the $\llbracket7,1,3\rrbracket$ code (namely, $\mathrm{QRM}(3)$) and the $\llbracket15,1,3\rrbracket$ code (namely, $\mathrm{QRM}(4)$) described earlier.
\end{remark}

\begin{remark}\label{rem:10qubit_code_switching}
    A fault-tolerant code-switching protocol between a $\llbracket 7,1,3 \rrbracket$ code and a $\llbracket 10,1,2 \rrbracket$ code (which is derived from the $\llbracket 15,1,3 \rrbracket$ code by \emph{code morphing}) is also described in \cite[Section VI]{butt2024fault}. As explained in \cite[Section VI.D]{butt2024fault}, this can also be viewed as a subsystem code whose gauge group can be generated up to phase by $12$ Pauli operators (but no less) and whose stabilizer group has rank $6$. Thus, there are $g = \frac{1}{2}(12-6)=3$ gauge qubits, which implies that this code-switching protocol is also strongly contextual in a partial closure by Corollary~\ref{cor:subsystem_code_contextual_iff_g_ge_2}.
\end{remark}

\subsection{Example 2: Code-Switching Between Doubled Color Codes}

The protocol described above is a specific instance of a broader strategy to achieve universal fault-tolerant computation by switching between codes with complementary transversal gate sets and thus allowing us to bypass the Eastin--Knill theorem~\cite{eastin2009restrictions}. Bravyi and Cross introduced three such families of protocols based on a recursive construction called \emph{doubled color codes}~\cite{bravyi2015doubled}. For any desired odd distance $d=2t+1$ (for $t \ge 1$), these protocols provide a method for switching between a ``$C$-code'' with transversal Clifford gates and a ``$T$-code'' with a transversal $T$ gate. Following~\cite{bravyi2015doubled}, we will refer to the underlying subsystem code defining this code-switching protocol as the \emph{base code}. We show that the base codes (and hence the protocols) for all three families are strongly contextual.

The protocols in all three families follow a similar structure, which we will now describe. Each base code is an $\llbracket n_t, 1, g_t, d_t = 2t+1\rrbracket$ subsystem code, with both $n_t$ and $d_t$ odd, and the single logical $X$- and $Z$-operators represented by the transversal $X^{\otimes n}$- and $Z^{\otimes n}$-operators acting on all the qubits. To define the stabilizer and gauge groups, we introduce some additional notation from~\cite{bravyi2015doubled}. For a binary vector $f \in \F_2^n$, let $X(f)$ be the tensor product of $X$-operators on the qubits in the support of $f$. We use $Z(f)$ analogously. For two linear subspaces $\mathcal{A}, \mathcal{B} \subseteq \mathbb{F}_2^n$, $\mathrm{CSS}(\mathcal{A}, \mathcal{B})$ denotes the subgroup of the Pauli group $G_n$ generated by $\{X(a) : a \in \mathcal{A}\}$ and $\{Z(b) : b \in \mathcal{B}\}$. For a subspace $\mathcal{W} \subseteq \F_2^n$, we also define $\dot{\mathcal{W}} \subseteq \mathcal{W}^\perp$ as the subspace of all \emph{even-weight} vectors orthogonal to $\mathcal{W}$.

Each code-switching protocol is defined by a pair of nested subspaces $\mathcal{W}_T \subseteq \mathcal{W}_C$ (whose elements all have even weight) that satisfy the inclusion relation $\mathcal{W}_T \subseteq \mathcal{W}_C \subseteq \dot{\mathcal{W}}_C \subseteq \dot{\mathcal{W}}_T$. The protocol switches between two codes:
\begin{itemize}
    \item the $T$-code, which is a regular CSS (stabilizer) code with stabilizer group $\mathrm{CSS}(\mathcal{W}_T, \dot{\mathcal{W}}_T)$ and gauge group generated by $iI$ and $\mathrm{CSS}(\mathcal{W}_T, \dot{\mathcal{W}}_T)$, and
    \item the $C$-code, which is a subsystem code with stabilizer group $\mathrm{CSS}(\mathcal{W}_C, \mathcal{W}_C)$ and gauge group generated by $iI$ and $\mathrm{CSS}(\dot{\mathcal{W}}_C, \dot{\mathcal{W}}_C)$.
\end{itemize}
The base code is therefore the subsystem code whose stabilizer group is $\mathcal{S} = \mathrm{CSS}(\mathcal{W}_T, \mathcal{W}_C)$ (the intersection of the two stabilizer groups), and whose gauge group $\mathcal{G}$ is generated by $iI$ and $\mathrm{CSS}(\dot{\mathcal{W}}_C, \dot{\mathcal{W}}_T)$.

The three families of doubled color codes are constructed iteratively, starting with a basic construction whose check measurements may be nonlocal, and progressively adding ancillary qubits and additional gauge check measurements to replace all the nonlocal checks with local ones:
\begin{itemize}
    \item \textbf{Family 1 (Basic construction)}: This family is derived from the standard 2D color codes~\cite{bombin2006topological} on hexagonal lattices. Each color code lattice $\Lambda_t$ has $m_t = 3t^2+3t+1$ sites, for $t \ge 1$. The face operators of this lattice span a subspace $\mathcal{S}_t \subseteq \mathbb{F}_2^{m_t}$ of dimension $(m_t-1)/2$. Let $n_t = 1 + 2\sum_{j=1}^t m_j = 2t^3 + 6t^2 + 6t + 1$; this is the total number of physical qubits in the base code. The defining subspaces for the protocol, $\mathcal{W}_T = \mathcal{T}_t$ and $\mathcal{W}_C = \mathcal{C}_t$ (in $\F_2^{n_t}$), are constructed recursively from the subspaces $\{\mathcal{S}_j\}_{j \le t}$ via a ``doubling map.'' In this construction, we have $\mathcal{C}_t = \dot{\mathcal{C}}_t$, so the $C$-code is also a stabilizer code. The resulting code-switching protocol involves nonlocal check operators.

    \item \textbf{Family 2 (Intermediate construction):} This family extends the first by adding ancillary qubits and new local gauge generators to decompose the nonlocal checks. A total of $\sum_{j=2}^t 2j = t^2 + t - 2$ physical qubits are added in this construction\footnote{This is initially given as $\sum_{j=1}^t 2j = t^2 + t$ in~Equation (85) of \cite{bravyi2015doubled}, but it was later noted at the very end of Section 12 in the paper that this number can be reduced by $2$.}, so the total number of physical qubits becomes $n_t = 2t^3 + 7t^2 + 7t - 1$. The new defining subspaces, $\mathcal{W}_T = \mathcal{U}_t$ and $\mathcal{W}_C =\mathcal{D}_t$, are obtained from their respective counterparts $\mathcal{T}_t$ and $\mathcal{C}_t$ by appending zeros to the vectors, so there is no change to the dimensions of these subspaces from their counterparts. The resulting protocol involves mostly local checks, with the exception of some nonlocal checks of weight two.

    \item \textbf{Family 3 (Final construction):} A final ``subdivision gadget'' adds more ancillas to make all gauge generators spatially local on a 2D lattice. This yields the final subspaces, $\mathcal{W}_T = \mathcal{V}_t$ and $\mathcal{W}_C = \mathcal{F}_t$, which again have the same dimensions as $\mathcal{U}_t$ and $\mathcal{D}_t$ respectively. A total of $\sum_{j=2}^t (2j-2) = t^2 - t$ physical qubits are added at this final step, bringing the total number of physical qubits to $n_t = 2t^3 + 8t^2 + 6t - 1$.
\end{itemize}

\begin{lemma}\label{lem:gauge_qubits_of_doubled_color_code}
    The base code for the $t$-th member of the family for the basic, intermediate and final construction has $\tfrac{1}{2}(t^3 + 3t^2 + 2t)$, $\tfrac{1}{2}(t^3 + 5t^2 + 4t - 4)$ and $\tfrac{1}{2}(t^3 + 7t^2 + 2t - 4)$ gauge qubits respectively.
\end{lemma}
\begin{proof}
    Let $\mathcal{H} \subseteq \F_2^{n_t}$ be the subspace of even-weight vectors, so $\dim \mathcal{H} = n_t - 1$. Since $\mathcal{W}_T$ and $\mathcal{W}_C$ are both contained in $\mathcal{H}$, then $\dot{\mathcal{W}}_T$ and $\dot{\mathcal{W}}_C$ are the orthogonal complements of $\mathcal{W}_T$ and $\mathcal{W}_C$ in $\mathcal{H}$ respectively. As the (regular) inner product is nondegenerate on $\mathcal{H}$ (this requires that $n_t$ is odd, so the all-ones vector is not in $\mathcal{H}$),
    \begin{equation} \label{eq:orthogonality_color_code}
        \dim \mathcal{H} = \dim \mathcal{W}_T + \dim \dot{\mathcal{W}}_T = \dim \mathcal{W}_C + \dim \dot{\mathcal{W}}_C.
    \end{equation}
    The rank $s$ of the stabilizer group $\mathcal{S} = \mathrm{CSS}(\mathcal{W}_T, \mathcal{W}_C)$ is $\dim \mathcal{W}_T + \dim \mathcal{W}_C$. The gauge group $\mathcal{G}$ is generated by $iI$ and $\mathrm{CSS}(\dot{\mathcal{W}}_C, \dot{\mathcal{W}}_T)$, so by Equation~\eqref{eq:orthogonality_color_code}, there are $g_t = \tfrac{1}{2}\left((\dim \dot{\mathcal{W}}_C + \dim \dot{\mathcal{W}}_T) - (\dim \mathcal{W}_T + \dim \mathcal{W}_C)\right) = \dim \mathcal{H} - \dim \mathcal{W}_C - \dim \mathcal{W}_T$ gauge qubits. We now specialize to each of the three constructions.

    For the basic construction, the subspace $\mathcal{T}_t$ has dimension $t + \sum_{j=1}^t \dim \mathcal{S}_j$, by Equation~(71) of \cite{bravyi2015doubled}. Since $\dim \mathcal{S}_j = \tfrac{1}{2}(m_t - 1) = \tfrac{3}{2}(t^2 + t)$, then using the formula for sums of consecutive integers and squares, we have $\dim \mathcal{T}_t = \tfrac{1}{2}(t^3 + 3t^2 + 4t)$. Similarly, Equation~(81) of \cite{bravyi2015doubled} implies that $\dim \mathcal{C}_t = t + 2\sum_{j=1}^t \dim \mathcal{S}_j = t^3 + 3t^2 + 3t$. Therefore,
    \begin{align*}
        g_t &= \dim \mathcal{H} - \dim \mathcal{C}_t - \dim \mathcal{T}_t \\
        & = (2t^3 + 6t^2 + 6t) - (t^3 + 3t^2 + 3t) - \tfrac{1}{2}(t^3 + 3t^2 + 4t) \\
        &= \tfrac{1}{2}(t^3 + 3t^2 + 2t).
    \end{align*}

    For the intermediate construction, the value of $n_t$ (and hence of $\dim \mathcal{H}$) is increased by $t^2 + t - 2$ compared to the basic construction, while the dimensions of the other two subspaces $\mathcal{D}_t$ and $\mathcal{U}_t$ remain the same as their counterparts $\mathcal{C}_t$ and $\mathcal{T}_t$. Thus, the number of gauge qubits is increased by $t^2 + t - 2$, giving us $g_t = \tfrac{1}{2}(t^3 + 3t^2 + 2t) + t^2 + t - 2 = \tfrac{1}{2}(t^3 + 5t^2 + 4t - 4)$.

    For the final construction, the value of $n_t$ (and hence of $\dim \mathcal{H}$) is increased by $t^2 - t$ compared to the basic construction, while the dimensions of the other two subspaces $\mathcal{F}_t$ and $\mathcal{V}_t$ remain the same as their counterparts $\mathcal{D}_t$ and $\mathcal{U}_t$. Thus, the number of gauge qubits is increased by $t^2 - t$, giving us $g_t = \tfrac{1}{2}(t^3 + 5t^2 + 4t - 4) + t^2 - t = \tfrac{1}{2}(t^3 + 7t^2 + 2t - 4)$.
\end{proof}

\begin{corollary}\label{cor:doubled_color_code_contextuality}
    All the code-switching protocols in any of the three families described above are strongly contextual in a partial closure.
\end{corollary}
\begin{proof}
    From Lemma~\ref{lem:gauge_qubits_of_doubled_color_code}, we see that (for $t \ge 1$) the number of gauge qubits in the base code for the basic, intermediate, and final constructions must satisfy $g_t = \tfrac{1}{2}(t^3 + 3t^2 + 2t) \ge \tfrac{1}{2}(1+3+2) = 3$, $g_t = \tfrac{1}{2}(t^3 + 5t^2 + 4t - 4) \ge \tfrac{1}{2}(1+5+4-4) = 3$, and $g_t = \tfrac{1}{2}(t^3 + 7t^2 + 2t - 4) \ge \tfrac{1}{2}(1+7+2-4)=3$ respectively. The statement of the corollary then follows from Corollary~\ref{cor:subsystem_code_contextual_iff_g_ge_2} (and Definition~\ref{def:code-switching-contextual}).
\end{proof}

\subsection{Necessity of Strong Contextuality for Universal Transversal Gates from Code-Switching}

We have seen that the code-switching protocols discussed in this section have all been strongly contextual in a partial closure (Corollary~\ref{cor:code_switching_contextuality}, Remark~\ref{rem:RM_code_switching}, Remark~\ref{rem:10qubit_code_switching}, Corollary~\ref{cor:doubled_color_code_contextuality}). This certainly need not be the case for all code-switching protocols, since we can always begin with a subsystem code with $g = 1$ (which is noncontextual, by Corollary~\ref{cor:subsystem_code_contextual_iff_g_ge_2}) and fix the gauge qubit in two different ways to get the two stabilizer codes for a noncontextual code-switching protocol. What appears to be unique to each of the protocols we analyzed is that they involve switching between a pair of subsystem codes that together admit a universal transversal gate set. This suggests a possible broader principle: any code-switching protocol that achieves transversal universality must exhibit strong contextuality in its partial closure. We formalize this as the following conjecture, which may be viewed as a no-go conjecture for noncontextual routes to achieving transversal universality via code-switching.

\begin{conjecture} \label{conj:code-switching-strongly-contextual-partial-closure}
    If a fault-tolerant code-switching protocol between subsystem stabilizer codes $\mathcal{Q}_1$ and $\mathcal{Q}_2$ is such that the union of the set of  transversal gates for $\mathcal{Q}_1$ and that for $\mathcal{Q}_2$ forms a universal gate set, then the code-switching protocol must be strongly contextual in a partial closure.
\end{conjecture}

By Corollary~\ref{cor:subsystem_code_contextual_iff_g_ge_2}, this is equivalent to saying that the underlying subsystem code must have at least two gauge qubits. Since the Eastin--Knill theorem rules out the possibility that $g = 0$ (which would have implied that the code $\mathcal{Q}_1 = \mathcal{Q}_2$ admits a universal transversal gate set), Conjecture~\ref{conj:code-switching-strongly-contextual-partial-closure} would hold if one can show that $g \ne 1$ under the given hypotheses.

If this conjecture is true, it would serve as an evidence that contextuality may be a critical resource for achieving universal quantum computation via code-switching and transversal gates, complementing past results from, e.g., \cite{howard2014contextuality} that contextuality is necessary for achieving universal quantum computation via magic state distillation.

While we do not have a complete resolution of Conjecture~\ref{conj:code-switching-strongly-contextual-partial-closure}, we can prove a version of it under several assumptions commonly satisfied by code-switching constructions achieving universal fault tolerance, including by the basic construction of the doubled color codes family (see Remark~\ref{rem:code_switching_transversal_Tdagger}). For instance, a common universal gate set used to achieve universal quantum computation is $\{H, T, \CNOT\}$, which motivates our assumption that one of the two codes is a stabilizer code admitting logical transversal $\CNOT$ and $T$ gates, while the other is a subsystem code admitting a logical transversal $H$ gate. We also assume the logical gates are implemented using \emph{physical} transversal gates of the same type, i.e., the logical gate $U$ is implemented by applying the $U$ gate to all of the physical qubits. (This last condition can be slightly relaxed; see Remark~\ref{rem:code_switching_relaxation}.)

\begin{theorem} \label{thm:code-switching-contextual}
    Consider a code-switching protocol between two $n$-qubit subsystem codes $\mathcal{Q}_1$ and $\mathcal{Q}_2$ (each with distance greater than $1$),
    such that the underlying $n$-qubit subsystem code $\mathcal{Q}$ encodes one logical qubit, with the $X$- and $Z$-logical operators given by $\overline{X} = \prod_{j=1}^n X_j$ and $\overline{Z} = \prod_{j=1}^n Z_j$. Suppose that:
    \begin{itemize}
        \item $\mathcal{Q}_1$ is a stabilizer code admitting logical $\CNOT$ and $T$ gates via physical transversal $\CNOT$ and $T$ gates respectively, and
        \item $\mathcal{Q}_2$ is a subsystem stabilizer code admitting the logical $H$ gate via the physical transversal $H$ gate.
    \end{itemize}
    Then $\mathcal{Q}$ must have at least $3$ gauge qubits, and the code-switching protocol is therefore strongly contextual in a partial closure.
\end{theorem}

\begin{proof}
    Before presenting the formal proof, we outline its logical structure, which proceeds in four main steps. First, we use the condition that the physical transversal $\CNOT$ gates realize a logical operator on $\mathcal{Q}_1$ to establish that it must be a CSS code. Second, we leverage the requirement that the CSS code $\mathcal{Q}_1$ has physical transversal $T$ gates realizing the logical $T$ to apply the characterization of such ``CSS-T'' codes proven in~\cite{rengaswamy2020optimal}, which imposes strong constraints on the generator matrix of this CSS code. Third, these constraints allow us to prove a key result: the number of independent pure $Z$-type stabilizers must exceed the number of independent pure $X$-type stabilizers by at least six. Finally, using the symplectic representation, we show how the physical transversal $H$ gates realizing a logical operator on $\mathcal{Q}_2$ translates this imbalance in the stabilizer group of $\mathcal{Q}_1$ into a lower bound on the number of gauge operators, forcing the parent code $\mathcal{Q}$ to have at least three gauge qubits.
    
    We first consider the structure of $\mathcal{Q}_1$. It is well known that a stabilizer code admits a logical $\CNOT$ via physical transversal $\CNOT$ gates if and only if it is a CSS code (i.e., a stabilizer code whose stabilizer group is generated by pure $X$-type and pure $Z$-type generators, with phase $+1$).  Thus, $\mathcal{Q}_1$ is a CSS code.
    
    Having established that $\mathcal{Q}_1$ is a CSS code, we now use the transversal $T$ gate condition to constrain its underlying classical codes. We can apply Theorem 14 in~\cite{rengaswamy2020optimal} to characterize the stabilizer group of $\mathcal{Q}_1$. Suppose that the CSS code $\mathcal{Q}_1$ is obtained from two classical binary codes $C_1, C_2$ with $0 \ne C_2 \subseteq C_1 \subseteq \F_2^n$, with the stabilizer group $\mathcal{S}_1$ generated by the $X$-type Pauli operators $X(f)$ for each $f \in C_2$ and the $Z$-type Pauli operators $Z(f)$ for each $f \in C_1^\perp$. We can write a generator matrix for $C_1$ as $G_{C_1} = \begin{bmatrix} G_{C_1/C_2} \\ G_{C_2} \end{bmatrix}$, where $G_{C_2}$ is the generator matrix of $C_2$ and $G_{C_1/C_2}$ is the generator matrix of $C_1/C_2$ using representatives in $C_1$. Note that $G_{C_1/C_2}$ generates the logical $X$ group, so by assumption this must be a single row containing the all-ones vector $\bar{1} \in \F_2^n$. By Theorem 14 in~\cite{rengaswamy2020optimal}, the matrix $G_{C_1}$ must be \emph{triorthogonal}, in the sense that the intersection of the support of any pair or triples of distinct rows of $G_{C_1}$ is of even cardinality, and for all $a \in C_2$:
    \begin{itemize}
        \item the weight of $a$ is congruent to $0$ modulo $8$, and
        \item the weight of $a + \bar{1}$ (which is $n$ minus the weight of $a$) is congruent to $1$ modulo $8$.
    \end{itemize}
    From the last two conditions, we see that the weight of every row in $G_{C_2}$ must be even (in fact, $0$ modulo $8$) and $n$ must be odd (in fact, $1$ modulo $8$). Moreover, since the support of every row in $C_2$ has even overlap with that of every row in $C_1$ (by triorthogonality), we have $C_2 \subseteq C_1^\perp$, which means $Z(f) \in \mathcal{S}_1$ if $X(f) \in \mathcal{S}_1$, for any $f \in \F_2^n$.

    Our next step is to use this characterization of $C_1$ and $C_2$ to demonstrate an asymmetry between its pure $X$-type and $Z$-type generators. We will prove a lower bound on how many ``excess'' pure $Z$-type operators there are compared to pure $X$-type operators in $\mathcal{S}_1$; specifically, we will show that $\dim C_1^\perp - \dim C_2 \ge 6$. First, observe that $C_2$ satisfies a stronger triorthogonality condition, in the sense that the intersection of the supports of any three (not necessarily distinct) elements in $C_2$ is of even cardinality. To see this, note that this property is exactly saying that $F(x,y,z) \coloneqq \sum_{j=1}^n x_j y_j z_j = 0 \in \F_2$ for every $x=(x_j)_{j=1}^n, y=(y_j)_{j=1}^n, z=(z_j)_{j=1}^n \in C_2$, and this follows from linearity of $F$ in each variable and the fact that it holds for all rows of $G_{C_2}$ (as the $G_{C_2}$ is triorthogonal and all its rows have even weight). Let $\supp(v) \subseteq \{1, \cdots, n\}$ denote the support of $v \in \F_2^n$. The strong triorthogonality condition ensures that $\supp(v)$, $\supp(v) \cap \supp(v')$ and $\supp(v) \cap \supp(v') \cap \supp(v'')$ are all of even cardinality for distinct $v,v',v'' \in C_2$ (noting that the condition works for triples of elements that may not be pairwise distinct). Moreover, by the inclusion-exclusion principle,
    \begin{equation} \label{eq:inclusion_exclusion_supports}
        \left|\bigcup_{v \in C_2} \supp(v)\right| =  \sum_{\emptyset \neq A \subseteq C_2} (-1)^{|A|-1} \left|\bigcap_{v \in A} \supp(v)\right|.
    \end{equation}
    (In other words, the size of the union of the supports of elements in $C_2$ is the sum of the sizes of each of the supports, minus the sum of the sizes of the intersection of two distinct supports, plus the sum of the sizes of the intersection of three distinct supports, and so on.) The union of the supports of elements in $C_2$ must be all of $\{1, \cdots, n\}$; otherwise, some $j \in \{1, \cdots, n\}$ must not be in the support of any element in $C_2$, i.e., no $X$-operator has any support on qubit $j$, and so the weight-$1$ operator $Z_j$ will commute with every element in the stabilizer group, making $\mathcal{Q}_1$ a distance-$1$ code. Thus, the left-hand side of Equation~\eqref{eq:inclusion_exclusion_supports} is equal to $|C_2| = n$, which is odd. Therefore, at least one of the summands on the right-hand side is odd, i.e., there is some nonempty $A \subseteq C_2$ such that $\left|\bigcap_{v \in A} \supp(v)\right|$ is odd. Choose such an $A$ with the smallest cardinality, so the intersection of the supports of any $m$ distinct elements in $C_2$ is even whenever $m \le |A|-1$. By the strong triorthogonality condition, $|A| \ge 4$.
    In the language of~\cite{sudakov2018two}, the set $\{\supp(v) : v \in C_2\}$ forms a strong $(|A|-1)$-wise eventown on $\{1, \cdots, n\}$ but not a strong $|A|$-wise eventown\footnote{A collection of sets $\{A_1,\dots,A_\ell\}$ is a strong $k$-wise eventown if $|\bigcap_{j\in S} A_j| = 0 \bmod 2$ for any $S\subseteq\{1,\dots,\ell\}$ with $|S|=m$ for all $1\leq m\leq k$. So, $\{\supp(v) : v \in C_2\}$ is a strong $(|A|-1)$-wise eventown since $|\bigcap_{v\in S} \supp(v)| = 0 \bmod 2$ for all $S\subseteq\{1,\dots,n\}$ with $1\leq|S|\leq|A|-1$, by the minimality of $|A|$. However it is not a $|A|$-wise eventown since $|\bigcap_{v\in A} \supp(v)| = 1 \bmod 2$, by the definition of $A$.}. Theorem 3 of~\cite{sudakov2018two} then says that $|C_2| \le 2^{\floor{n/2} - (2^{|A|-1} - |A| - 1)}$, i.e., $\dim C_2 \le \floor{n/2} - (2^{|A|-1} - |A| - 1)$. Since $|A| \ge 4$ and the expression $2^{|A|-1} - |A| - 1$ is increasing for $|A| \ge 4$, then $\dim C_2 \le \floor{n/2} - (2^{4-1} - 4 - 1) = (n-7)/2$ (noting $n$ is odd). By nondegeneracy of the inner product on $\F_2^n$,
    \begin{align*}
        \dim C_1^\perp - \dim C_2 &= (n - \dim C_1) - \dim C_2 \\
        &= n - (\dim C_2 + 1) - \dim C_2 \\
        &= n - 1 - 2 \cdot \dim C_2 \\
        &\ge n - 1 - (n-7) \ge 6,
    \end{align*}
    as we wanted.

    Having determined the structure of the stabilizer group $\mathcal{S}_1$ of $\mathcal{Q}_1$, we can now turn our attention to the stabilizer group $\mathcal{S}_2$ of $\mathcal{Q}_2$ and the intersection $\mathcal{S}_1 \cap \mathcal{S}_2$. For this, we shall work with the usual symplectic representation of $G_n$: to each Pauli operator $\left(\Pi_{j=1}^n X_j^{x_j}\right) \left(\Pi_{j=1}^n Z_j^{z_j}\right)$ in $G_n$ (modulo its phase), we associate the vector $(x | z) \in \F_2^{2n}$ for $x = (x_j)_{j=1}^n, z = (z_j)_{j=1}^n \in \F_2^n$. Let $W_1, W_2 \subseteq \F_2^{2n}$ be the subspaces for $\mathcal{S}_1$ and $\mathcal{S}_2$. The intersection $\mathcal{S}_1 \cap \mathcal{S}_2$ is then represented by $W_1 \cap W_2$, so the number of gauge operators fixed when switching to the code $\mathcal{Q}_j$ is in fact $\dim (W_j/(W_1 \cap W_2))$. We will now calculate this for each of $\mathcal{Q}_1$ and $\mathcal{Q}_2$.
    
    Since $C_2 \subseteq C_1^\perp$, we can find a subspace $C_3 \subseteq C_1^\perp$ such that $C_1^\perp = C_2 \oplus C_3$ as a direct sum of vector spaces, so $\dim C_3 = \dim C_1^\perp - \dim C_2 \ge 6$. Define $V = \{(0|z) : z \in C_3\}$ to be the subspace of $W_1$ consisting of all pure $Z$-type operators whose supports are represented by the vectors in $C_3$, so $\dim V = \dim C_3 \ge 6$. 
    Composing\footnote{The inclusion map $V \hookrightarrow W_1$ maps $v\in V$ to $v\in W_1$. The quotient map $W_1 \twoheadrightarrow W_1/(W_1 \cap W_2)$ maps $w\in W_1$ to coset $w+(W_1+W_2)$.} the inclusion $V \hookrightarrow W_1$ with the quotient $W_1 \twoheadrightarrow W_1/(W_1 \cap W_2)$ gives the linear map $V \to W_1/(W_1 \cap W_2)$ with kernel $V \cap W_2$, so
    \begin{equation*} \label{eq:code_switching_general_W1_bound}
        \dim \frac{W_1}{W_1 \cap W_2} \ge \dim \frac{V}{V \cap W_2} = \dim V - \dim (V \cap W_2).
    \end{equation*}
    Let $\eta \colon \F_2^{2n} \to \F_2^{2n}$ represent the physical transversal $H$ gate, so $\eta((x|z)) = (z|x)$ for all $x,z \in \F_2^n$. Since the physical transversal $H$ gate is a logical gate for $\mathcal{Q}_2$, then it must preserve $\mathcal{S}_2$, i.e., $\eta(W_2) \subseteq W_2$. Therefore, $\eta(V \cap W_2) \subseteq \eta(V) \cap W_2$. But $\eta(V) = \{(z|0) : z \in C_3\}$ must intersect $W_1$ only trivially, since all vectors in $W_1$ of the form $(z|0)$ are such that $z \in C_2$, but $C_3$ is chosen to intersect $C_2$ only trivially. Hence $\eta(V \cap W_2)$ must intersect $W_1$ only trivially. Composing the inclusion $\eta(V \cap W_2) \hookrightarrow W_2$ with the quotient $W_2 \twoheadrightarrow W_2/(W_1 \cap W_2)$ gives the linear map $\eta(V \cap W_2) \to W_2/(W_1 \cap W_2)$ with kernel $\eta(V \cap W_2) \cap W_1 = \{0\}$, so
    \begin{equation*} \label{eq:code_switching_general_W2_bound}
        \dim \frac{W_2}{W_1 \cap W_2} \ge \dim \eta(V \cap W_2) = \dim (V \cap W_2),
    \end{equation*}
    since $\eta$ is an automorphism.

    The number of gauge qubits $g$ is greater than or equal to the number of gauge operators fixed when switching to either $\mathcal{Q}_1$ or $\mathcal{Q}_2$, so a lower bound for $g$ is
    \[\min\left(\dim \frac{W_1}{W_1 \cap W_2}, \dim \frac{W_2}{W_1 \cap W_2} \right).\]
    We showed that this is at least
    \[\min(\dim V - \dim(V\cap W_2), \dim (V \cap W_2)),\]
    which must be at least $\ceil{\tfrac{1}{2} \dim V}$. Since $\dim V \ge 6$, we get $g \ge 3$.

    Finally, since $\mathcal{Q}$ has at least $3$ gauge qubits, then the code-switching protocol is indeed strongly contextual in a partial closure by Corollary \ref{cor:subsystem_code_contextual_iff_g_ge_2}.
\end{proof}

This theorem establishes a restricted form of Conjecture~\ref{conj:code-switching-strongly-contextual-partial-closure}: noncontextuality of a subsystem code strictly limits the possibility of assembling a universal transversal gate set through code-switching. The general case remains open, but Theorem~\ref{thm:code-switching-contextual} demonstrates that Conjecture~\ref{conj:code-switching-strongly-contextual-partial-closure} holds in a nontrivial and widely studied subclass of code-switching protocols (see Remark~\ref{rem:code_switching_transversal_Tdagger}).

\begin{remark} \label{rem:code_switching_relaxation}
    We can relax the conditions of Theorem~\ref{thm:code-switching-contextual} to allow for Pauli corrections on the physical transversal gates. This means the logical $T$ gate may be implemented as $P (\Pi_{j=1}^n T_j) P^{\dagger}$ for some $n$-qubit Pauli $P$ (where $T_j$ is the $T$ operator applied to qubit $j$), and similarly for the logical $H$ gate, while the logical $\CNOT$ gate (between two logical qubits) may be implemented as $(P \otimes P) (\Pi_{j=1}^n \CNOT_j) (P \otimes P)^{\dagger}$ for some $n$-qubit Pauli $P$ (where $\CNOT_j$ is the $\CNOT$ operator applied between qubit $j$ of the first logical qubit and qubit $j$ of the second logical qubit). To see why this is the case, we examine how the proof changes when we allow for Pauli corrections to the physical transversal gates. We first used the fact that the physical transversal $\CNOT$ realizes a logical gate on $\mathcal{Q}_1$ to establish that $\mathcal{Q}_1$ is a CSS code. If instead we only know that the physical transversal $\CNOT$ realizes a logical gate on $\mathcal{Q}_1$ up to some Pauli corrections, then we know that it is a CSS code up to some Pauli corrections, i.e., the stabilizer group of $\mathcal{Q}_1$ may need to be conjugated by some $n$-qubit Pauli $P$ for $\mathcal{Q}_1$ to be a CSS code. Note, however, that such a conjugation only flips the signs of some of the stabilizers (and stabilizer generators). The next place in the proof that would change is the use of Theorem 14 in~\cite{rengaswamy2020optimal}. Instead of applying the theorem to $\mathcal{Q}_1$, we can instead apply it to an appropriately Pauli-corrected version of $\mathcal{Q}_1$ such that it admits the logical $T$ gate via physical transversal $T$ gates. Note that this will be some Pauli-corrected version of a CSS code, which is still a CSS code by the convention in~\cite{rengaswamy2020optimal}. Thus, all the results regarding the structure of $\mathcal{Q}_1$ continue to hold, with the exception that some of the stabilizers in $\mathcal{S}_1$ might have their sign flipped. Finally, we used the fact that the physical transversal $H$ is a logical gate on $\mathcal{Q}_2$ to deduce that $\eta(W_2) \subseteq W_2$. This is now only true for a Pauli-corrected version of $\mathcal{Q}_2$, but it will still be true that $\eta(W_2) \subseteq W_2$, because we have dropped the phases in the symplectic representation of the stabilizers.    
\end{remark}

\begin{remark} \label{rem:code_switching_transversal_Tdagger}
    The weakening of the hypothesis of Theorem~\ref{thm:code-switching-contextual} as highlighted in Remark~\ref{rem:code_switching_relaxation} means we can allow the logical $T$ gate for $\mathcal{Q}_1$ to be implemented using physical transversal $T$ or $T^{\dagger}$ gates, which can differ between qubits\footnote{This notion of transversality is often referred to as ``weak transversality'' where a logical $U$ gate on an $n$-qubit code is given by $\Bar{U} = \bigotimes_{i=1}^n V_i$ where the physical unitaries $V_i,V_j$ for $i\neq j$ need not be the same (nor equal to $U$). This is a relaxation of transversality (which is sometimes referred to as ``strict transverality'') where $\Bar{U} = V^{\otimes n}$.} (noting that $XTX^{\dagger} = e^{i\pi/4} T^{\dagger}$). In particular, this will cover the switching between the extended $\llbracket 7,1,3 \rrbracket$ and the $\llbracket 15,1,3 \rrbracket$, where the logical $T$ gate for the $\llbracket 15,1,3 \rrbracket$ code is realized by a physical transversal $T^{\dagger}$ gate. This also covers the basic construction of the family of code-switching protocols in~\cite{bravyi2015doubled}, because the logical $T$ gate for the $T$ code in that case is realized by a physical $T$ or $T^{\dagger}$ gate on \emph{every} qubit (see Lemma 1 and Lemma 4 of~\cite{bravyi2015doubled}).
\end{remark}

In summary, this section has demonstrated through several widely studied examples and Theorem~\ref{thm:code-switching-contextual} that many commonly encountered fault-tolerant code-switching protocols enabling universal transversal gate sets are strongly contextual in a partial closure. Our results show that this is not a coincidence but a consequence of the underlying structure required to enable universal fault tolerance. Specifically, the need to accommodate complementary gate sets forces the parent subsystem code to have a sufficient number of gauge qubits, which is precisely the condition for strong contextuality established in Corollary~\ref{cor:subsystem_code_contextual_iff_g_ge_2}. This provides compelling evidence for Conjecture~\ref{conj:code-switching-strongly-contextual-partial-closure} and solidifies the role of contextuality as an essential structural feature for universal fault-tolerant quantum computation via code-switching protocols.

\section{Discussion and Open Problems} \label{sec:discussion_open_problems}
In this work, we forged a direct link between quantum error correction and quantum contextuality, with our contributions spanning three key aspects: fundamental, mathematical, and practical. Fundamentally, we developed a rigorous framework to characterize contextuality in QEC codes, proving that a subsystem stabilizer code is strongly contextual in its partial closure if and only if it has at least two gauge qubits (Corollary~\ref{cor:subsystem_code_contextual_iff_g_ge_2}). This finding establishes a clear physical criterion for when contextuality manifests. Mathematically, we proved the equivalence of several distinct definitions of contextuality---including the sheaf-theoretic and tree-based approaches---for the partial closure of Pauli measurement sets (Corollary~\ref{cor:strong-contextuality-equivalences}). This unification resolves a recent conjecture by Kim and Abramsky~\cite{kim2023stateindependentallversusnothingarguments}. Practically, we established that under mild, common assumptions, achieving universal fault-tolerant computation via code-switching requires the protocol to be strongly contextual in its partial closure (Theorem~\ref{thm:code-switching-contextual}). Ultimately, these findings synthesize our fundamental and mathematical insights into the practical conclusion that the structural complexity demanded of code-switching protocols that achieve universal computation is precisely what makes them strongly contextual under partial closure. An immediate design implication is that contextuality acts as a litmus test: if a code is noncontextual, then no transversal-universality scheme based on it can succeed. Conversely, the presence of contextuality signals the potential for universal fault tolerance. This positions contextuality as a complementary resource to entanglement and magic for realizing universality through code-switching protocols and transversal gates.

Our research lays a strong foundation for further exploration. In this study, we have already established a classification of quantum error-correcting codes as either contextual or noncontextual, providing key insights into their structure. While this binary classification captures essential distinctions, a natural next step is to quantify the degree of contextuality within a code (before taking partial closure) to refine our understanding further. A straightforward way to achieve this is through contextual separation~\cite{kirby2019contextuality}. Another promising approach might involve linear programming techniques~\cite{abramsky2016quantifying,boyd2004convex}, which have been successfully employed in related settings to quantify contextuality~\cite{abramsky2016quantifying}.  

While we have identified contextuality as a property of quantum error-correcting codes, its practical implications beyond its role in universal fault-tolerant computation remain unexplored. Future work should investigate how contextuality affects the performance, resource requirements, or error thresholds of a code. For instance, does greater contextuality enhance a code's ability to correct specific types of errors, or does it impose limitations on code construction? Understanding these operational consequences could guide the design of new codes optimized for specific quantum computational tasks.

In this work, using a sheaf-theoretic framework, we define the contextuality of quantum error-correcting codes. Recently, sheaf-theoretic ideas have been applied to quantum error correction~\cite{panteleev2024maximally,lin2024transversal,dinur2024expansion,kalachev2025maximally}. Exploring possible connections between our work and these developments~\cite{panteleev2024maximally,lin2024transversal,dinur2024expansion,kalachev2025maximally} could provide deeper insights into the role of contextuality in quantum error correction.

Our framework focused on static quantum error-correcting codes, without considering the role of measurement schedules or adaptivity. However, in dynamical quantum error correction~\cite{hastings2021dynamically,davydova2023floquet,fahimniya2023hyperbolic,gidney2021fault,haah2022boundaries,vuillot2021planar,paetznick2023performance,gidney2022benchmarking,hilaire2024enhanced,higgott2023constructions,aasen2023measurement,zhang2023x,dua2024engineering,alam2024dynamicallogicalqubitsbaconshor,bombin2023unifying,tanggara2024strategic,townsend2023floquetifying,ellison2023floquet,sullivan2023floquet,tanggara2024simple}, such as Floquet codes, the scheduling of measurements is critical to the code's functionality. Extending the notion of contextuality to incorporate these temporal aspects could provide a richer understanding of how contextuality manifests in time-dependent quantum codes.

Finally, given the strong contextuality in a partial closure of the code-switching protocols exhibited in Section \ref{sec:contextuality_code_switching}, one may ask whether every fault-tolerantly implementable code-switching protocol whose constituent codes jointly admit a universal transversal gate set---such as the one described in \cite[Section VI]{butt2024fault} and highlighted in Remark \ref{rem:10qubit_code_switching}---must necessarily be strongly contextual in a partial closure when viewed as a subsystem code. We formulate this as Conjecture~\ref{conj:code-switching-strongly-contextual-partial-closure}, which may be viewed as a contextual extension of the Eastin--Knill theorem~\cite{eastin2009restrictions}: while the latter theorem restricts the power of a \emph{single} transversal gate set, our conjecture posits that overcoming this limit via \emph{switching} between gate sets mandates the introduction of a fundamental nonclassical feature, namely strong contextuality. By Corollary \ref{cor:subsystem_code_contextual_iff_g_ge_2}, this is equivalent to asking if the subsystem codes associated to such code-switching protocols must have at least two gauge qubits. Our Theorem~\ref{thm:code-switching-contextual} provides a partial proof of this conjecture under assumptions satisfied by many widely studied code-switching protocols. Resolving this conjecture in its full generality would confirm the role of contextuality as a fundamental requirement for universal fault-tolerant quantum computation via code-switching protocols, thereby extending its established role in magic state distillation~\cite{howard2014contextuality} to a more general principle.

\begin{acknowledgments}
This research is supported by Q.InC Strategic Research and Translational Thrust. AT is supported by CQT PhD scholarship and Google PhD fellowship program. We thank Atul Singh Arora, Tobias Haug, Zlatko Minev, Varun Narsimhachar, and Yunlong Xiao for interesting discussions.
\end{acknowledgments}

\bibliography{references}

\begin{thebibliography}{100}%
\makeatletter
\providecommand \@ifxundefined [1]{%
 \@ifx{#1\undefined}
}%
\providecommand \@ifnum [1]{%
 \ifnum #1\expandafter \@firstoftwo
 \else \expandafter \@secondoftwo
 \fi
}%
\providecommand \@ifx [1]{%
 \ifx #1\expandafter \@firstoftwo
 \else \expandafter \@secondoftwo
 \fi
}%
\providecommand \natexlab [1]{#1}%
\providecommand \enquote  [1]{``#1''}%
\providecommand \bibnamefont  [1]{#1}%
\providecommand \bibfnamefont [1]{#1}%
\providecommand \citenamefont [1]{#1}%
\providecommand \href@noop [0]{\@secondoftwo}%
\providecommand \href [0]{\begingroup \@sanitize@url \@href}%
\providecommand \@href[1]{\@@startlink{#1}\@@href}%
\providecommand \@@href[1]{\endgroup#1\@@endlink}%
\providecommand \@sanitize@url [0]{\catcode `\\12\catcode `\$12\catcode `\&12\catcode `\#12\catcode `\^12\catcode `\_12\catcode `\%12\relax}%
\providecommand \@@startlink[1]{}%
\providecommand \@@endlink[0]{}%
\providecommand \url  [0]{\begingroup\@sanitize@url \@url }%
\providecommand \@url [1]{\endgroup\@href {#1}{\urlprefix }}%
\providecommand \urlprefix  [0]{URL }%
\providecommand \Eprint [0]{\href }%
\providecommand \doibase [0]{https://doi.org/}%
\providecommand \selectlanguage [0]{\@gobble}%
\providecommand \bibinfo  [0]{\@secondoftwo}%
\providecommand \bibfield  [0]{\@secondoftwo}%
\providecommand \translation [1]{[#1]}%
\providecommand \BibitemOpen [0]{}%
\providecommand \bibitemStop [0]{}%
\providecommand \bibitemNoStop [0]{.\EOS\space}%
\providecommand \EOS [0]{\spacefactor3000\relax}%
\providecommand \BibitemShut  [1]{\csname bibitem#1\endcsname}%
\let\auto@bib@innerbib\@empty
%</preamble>
\bibitem [{\citenamefont {Shor}(1995)}]{shor1995scheme}%
  \BibitemOpen
  \bibfield  {author} {\bibinfo {author} {\bibfnamefont {P.~W.}\ \bibnamefont {Shor}},\ }\bibfield  {title} {\bibinfo {title} {Scheme for reducing decoherence in quantum computer memory},\ }\href {https://doi.org/10.1103/physreva.52.r2493} {\bibfield  {journal} {\bibinfo  {journal} {Physical Review A}\ }\textbf {\bibinfo {volume} {52}},\ \bibinfo {pages} {R2493} (\bibinfo {year} {1995})}\BibitemShut {NoStop}%
\bibitem [{\citenamefont {Gottesman}(1997)}]{gottesman1997stabilizer}%
  \BibitemOpen
  \bibfield  {author} {\bibinfo {author} {\bibfnamefont {D.}~\bibnamefont {Gottesman}},\ }\emph {\bibinfo {title} {Stabilizer Codes and Quantum Error Correction}},\ \href {https://thesis.library.caltech.edu/2900/2/THESIS.pdf} {Ph.D. thesis},\ \bibinfo  {school} {California Institute of Technology} (\bibinfo {year} {1997})\BibitemShut {NoStop}%
\bibitem [{\citenamefont {Dennis}\ \emph {et~al.}(2002)\citenamefont {Dennis}, \citenamefont {Kitaev}, \citenamefont {Landahl},\ and\ \citenamefont {Preskill}}]{dennis2002topological}%
  \BibitemOpen
  \bibfield  {author} {\bibinfo {author} {\bibfnamefont {E.}~\bibnamefont {Dennis}}, \bibinfo {author} {\bibfnamefont {A.}~\bibnamefont {Kitaev}}, \bibinfo {author} {\bibfnamefont {A.}~\bibnamefont {Landahl}},\ and\ \bibinfo {author} {\bibfnamefont {J.}~\bibnamefont {Preskill}},\ }\bibfield  {title} {\bibinfo {title} {Topological quantum memory},\ }\href {https://doi.org/10.1063/1.1499754} {\bibfield  {journal} {\bibinfo  {journal} {Journal of Mathematical Physics}\ }\textbf {\bibinfo {volume} {43}},\ \bibinfo {pages} {4452} (\bibinfo {year} {2002})}\BibitemShut {NoStop}%
\bibitem [{\citenamefont {Kitaev}(2003)}]{kitaev2003fault}%
  \BibitemOpen
  \bibfield  {author} {\bibinfo {author} {\bibfnamefont {A.~Y.}\ \bibnamefont {Kitaev}},\ }\bibfield  {title} {\bibinfo {title} {Fault-tolerant quantum computation by anyons},\ }\href {https://doi.org/10.1016/s0003-4916(02)00018-0} {\bibfield  {journal} {\bibinfo  {journal} {Annals of Physics}\ }\textbf {\bibinfo {volume} {303}},\ \bibinfo {pages} {2} (\bibinfo {year} {2003})}\BibitemShut {NoStop}%
\bibitem [{\citenamefont {Lidar}\ and\ \citenamefont {Brun}(2013)}]{lidar2013quantum}%
  \BibitemOpen
  \bibinfo {editor} {\bibfnamefont {P.~A.}\ \bibnamefont {Lidar}}\ and\ \bibinfo {editor} {\bibfnamefont {T.~A.}\ \bibnamefont {Brun}},\ eds.,\ \href {https://doi.org/10.1017/cbo9781139034807} {\emph {\bibinfo {title} {Quantum Error Correction}}}\ (\bibinfo  {publisher} {Cambridge University Press},\ \bibinfo {year} {2013})\BibitemShut {NoStop}%
\bibitem [{\citenamefont {Terhal}(2015)}]{terhal2015quantum}%
  \BibitemOpen
  \bibfield  {author} {\bibinfo {author} {\bibfnamefont {B.~M.}\ \bibnamefont {Terhal}},\ }\bibfield  {title} {\bibinfo {title} {Quantum error correction for quantum memories},\ }\href {https://doi.org/10.1103/revmodphys.87.307} {\bibfield  {journal} {\bibinfo  {journal} {Reviews of Modern Physics}\ }\textbf {\bibinfo {volume} {87}},\ \bibinfo {pages} {307} (\bibinfo {year} {2015})}\BibitemShut {NoStop}%
\bibitem [{\citenamefont {Bennett}\ \emph {et~al.}(1996)\citenamefont {Bennett}, \citenamefont {DiVincenzo}, \citenamefont {Smolin},\ and\ \citenamefont {Wootters}}]{bennett1996mixed}%
  \BibitemOpen
  \bibfield  {author} {\bibinfo {author} {\bibfnamefont {C.~H.}\ \bibnamefont {Bennett}}, \bibinfo {author} {\bibfnamefont {D.~P.}\ \bibnamefont {DiVincenzo}}, \bibinfo {author} {\bibfnamefont {J.~A.}\ \bibnamefont {Smolin}},\ and\ \bibinfo {author} {\bibfnamefont {W.~K.}\ \bibnamefont {Wootters}},\ }\bibfield  {title} {\bibinfo {title} {Mixed-state entanglement and quantum error correction},\ }\href {https://doi.org/10.1103/physreva.54.3824} {\bibfield  {journal} {\bibinfo  {journal} {Physical Review A}\ }\textbf {\bibinfo {volume} {54}},\ \bibinfo {pages} {3824} (\bibinfo {year} {1996})}\BibitemShut {NoStop}%
\bibitem [{\citenamefont {Brun}\ \emph {et~al.}(2006)\citenamefont {Brun}, \citenamefont {Devetak},\ and\ \citenamefont {Hsieh}}]{brun2006correcting}%
  \BibitemOpen
  \bibfield  {author} {\bibinfo {author} {\bibfnamefont {T.}~\bibnamefont {Brun}}, \bibinfo {author} {\bibfnamefont {I.}~\bibnamefont {Devetak}},\ and\ \bibinfo {author} {\bibfnamefont {M.-H.}\ \bibnamefont {Hsieh}},\ }\bibfield  {title} {\bibinfo {title} {Correcting quantum errors with entanglement},\ }\href {https://doi.org/10.1126/science.1131563} {\bibfield  {journal} {\bibinfo  {journal} {Science}\ }\textbf {\bibinfo {volume} {314}},\ \bibinfo {pages} {436} (\bibinfo {year} {2006})}\BibitemShut {NoStop}%
\bibitem [{\citenamefont {D{\"u}r}\ and\ \citenamefont {Briegel}(2007)}]{dur2007entanglement}%
  \BibitemOpen
  \bibfield  {author} {\bibinfo {author} {\bibfnamefont {W.}~\bibnamefont {D{\"u}r}}\ and\ \bibinfo {author} {\bibfnamefont {H.~J.}\ \bibnamefont {Briegel}},\ }\bibfield  {title} {\bibinfo {title} {Entanglement purification and quantum error correction},\ }\href {https://doi.org/10.1088/0034-4885/70/8/r03} {\bibfield  {journal} {\bibinfo  {journal} {Reports on Progress in Physics}\ }\textbf {\bibinfo {volume} {70}},\ \bibinfo {pages} {1381} (\bibinfo {year} {2007})}\BibitemShut {NoStop}%
\bibitem [{\citenamefont {Bravyi}\ \emph {et~al.}(2025)\citenamefont {Bravyi}, \citenamefont {Lee}, \citenamefont {Li},\ and\ \citenamefont {Yoshida}}]{bravyi2024much}%
  \BibitemOpen
  \bibfield  {author} {\bibinfo {author} {\bibfnamefont {S.}~\bibnamefont {Bravyi}}, \bibinfo {author} {\bibfnamefont {D.}~\bibnamefont {Lee}}, \bibinfo {author} {\bibfnamefont {Z.}~\bibnamefont {Li}},\ and\ \bibinfo {author} {\bibfnamefont {B.}~\bibnamefont {Yoshida}},\ }\bibfield  {title} {\bibinfo {title} {How much entanglement is needed for quantum error correction?},\ }\href {https://doi.org/10.1103/physrevlett.134.210602} {\bibfield  {journal} {\bibinfo  {journal} {Physical Review Letters}\ }\textbf {\bibinfo {volume} {134}},\ \bibinfo {pages} {210602} (\bibinfo {year} {2025})}\BibitemShut {NoStop}%
\bibitem [{\citenamefont {Anshu}\ \emph {et~al.}(2023)\citenamefont {Anshu}, \citenamefont {Breuckmann},\ and\ \citenamefont {Nirkhe}}]{anshu2023nlts}%
  \BibitemOpen
  \bibfield  {author} {\bibinfo {author} {\bibfnamefont {A.}~\bibnamefont {Anshu}}, \bibinfo {author} {\bibfnamefont {N.~P.}\ \bibnamefont {Breuckmann}},\ and\ \bibinfo {author} {\bibfnamefont {C.}~\bibnamefont {Nirkhe}},\ }\bibfield  {title} {\bibinfo {title} {{NLTS} {H}amiltonians from good quantum codes},\ }in\ \href {https://doi.org/10.1145/3564246.3585114} {\emph {\bibinfo {booktitle} {Proceedings of the 55th Annual ACM Symposium on Theory of Computing}}},\ \bibinfo {series and number} {STOC ’23}\ (\bibinfo  {publisher} {ACM},\ \bibinfo {year} {2023})\ pp.\ \bibinfo {pages} {1090--1096}\BibitemShut {NoStop}%
\bibitem [{\citenamefont {Bravyi}\ and\ \citenamefont {Kitaev}(2005)}]{bravyi2005universal}%
  \BibitemOpen
  \bibfield  {author} {\bibinfo {author} {\bibfnamefont {S.}~\bibnamefont {Bravyi}}\ and\ \bibinfo {author} {\bibfnamefont {A.}~\bibnamefont {Kitaev}},\ }\bibfield  {title} {\bibinfo {title} {Universal quantum computation with ideal {C}lifford gates and noisy ancillas},\ }\href {https://doi.org/10.1103/physreva.71.022316} {\bibfield  {journal} {\bibinfo  {journal} {Physical Review A}\ }\textbf {\bibinfo {volume} {71}},\ \bibinfo {pages} {022316} (\bibinfo {year} {2005})}\BibitemShut {NoStop}%
\bibitem [{\citenamefont {Veitch}\ \emph {et~al.}(2014)\citenamefont {Veitch}, \citenamefont {Hamed~Mousavian}, \citenamefont {Gottesman},\ and\ \citenamefont {Emerson}}]{veitch2014resource}%
  \BibitemOpen
  \bibfield  {author} {\bibinfo {author} {\bibfnamefont {V.}~\bibnamefont {Veitch}}, \bibinfo {author} {\bibfnamefont {S.~A.}\ \bibnamefont {Hamed~Mousavian}}, \bibinfo {author} {\bibfnamefont {D.}~\bibnamefont {Gottesman}},\ and\ \bibinfo {author} {\bibfnamefont {J.}~\bibnamefont {Emerson}},\ }\bibfield  {title} {\bibinfo {title} {The resource theory of stabilizer quantum computation},\ }\href {https://doi.org/10.1088/1367-2630/16/1/013009} {\bibfield  {journal} {\bibinfo  {journal} {New Journal of Physics}\ }\textbf {\bibinfo {volume} {16}},\ \bibinfo {pages} {013009} (\bibinfo {year} {2014})}\BibitemShut {NoStop}%
\bibitem [{\citenamefont {Gidney}\ \emph {et~al.}(2024)\citenamefont {Gidney}, \citenamefont {Shutty},\ and\ \citenamefont {Jones}}]{gidney2024magic}%
  \BibitemOpen
  \bibfield  {author} {\bibinfo {author} {\bibfnamefont {C.}~\bibnamefont {Gidney}}, \bibinfo {author} {\bibfnamefont {N.}~\bibnamefont {Shutty}},\ and\ \bibinfo {author} {\bibfnamefont {C.}~\bibnamefont {Jones}},\ }\href@noop {} {\bibinfo {title} {Magic state cultivation: growing {T} states as cheap as {CNOT} gates}} (\bibinfo {year} {2024}),\ \Eprint {https://arxiv.org/abs/2409.17595} {arXiv:2409.17595 [quant-ph]} \BibitemShut {NoStop}%
\bibitem [{\citenamefont {Eastin}\ and\ \citenamefont {Knill}(2009)}]{eastin2009restrictions}%
  \BibitemOpen
  \bibfield  {author} {\bibinfo {author} {\bibfnamefont {B.}~\bibnamefont {Eastin}}\ and\ \bibinfo {author} {\bibfnamefont {E.}~\bibnamefont {Knill}},\ }\bibfield  {title} {\bibinfo {title} {Restrictions on transversal encoded quantum gate sets},\ }\href {https://doi.org/10.1103/physrevlett.102.110502} {\bibfield  {journal} {\bibinfo  {journal} {Physical Review Letters}\ }\textbf {\bibinfo {volume} {102}},\ \bibinfo {pages} {110502} (\bibinfo {year} {2009})}\BibitemShut {NoStop}%
\bibitem [{\citenamefont {Kochen}\ and\ \citenamefont {Specker}(1967)}]{Koc}%
  \BibitemOpen
  \bibfield  {author} {\bibinfo {author} {\bibfnamefont {S.}~\bibnamefont {Kochen}}\ and\ \bibinfo {author} {\bibfnamefont {E.~P.}\ \bibnamefont {Specker}},\ }\bibfield  {title} {\bibinfo {title} {The problem of hidden variables in quantum mechanics},\ }\href {http://www.jstor.org/stable/24902153} {\bibfield  {journal} {\bibinfo  {journal} {Journal of Mathematics and Mechanics}\ }\textbf {\bibinfo {volume} {17}},\ \bibinfo {pages} {59} (\bibinfo {year} {1967})}\BibitemShut {NoStop}%
\bibitem [{\citenamefont {Specker}(1960)}]{Specker1960-bg}%
  \BibitemOpen
  \bibfield  {author} {\bibinfo {author} {\bibfnamefont {E.}~\bibnamefont {Specker}},\ }\bibfield  {title} {\bibinfo {title} {Die {L}ogik nicht gleichzeitig entscheidbarer {A}ussagen},\ }\href {http://www.jstor.org/stable/42964323} {\bibfield  {journal} {\bibinfo  {journal} {Dialectica}\ }\textbf {\bibinfo {volume} {14}},\ \bibinfo {pages} {239} (\bibinfo {year} {1960})}\BibitemShut {NoStop}%
\bibitem [{\citenamefont {Budroni}\ \emph {et~al.}(2022)\citenamefont {Budroni}, \citenamefont {Cabello}, \citenamefont {G{\"u}hne}, \citenamefont {Kleinmann},\ and\ \citenamefont {Larsson}}]{budroni2022kochen}%
  \BibitemOpen
  \bibfield  {author} {\bibinfo {author} {\bibfnamefont {C.}~\bibnamefont {Budroni}}, \bibinfo {author} {\bibfnamefont {A.}~\bibnamefont {Cabello}}, \bibinfo {author} {\bibfnamefont {O.}~\bibnamefont {G{\"u}hne}}, \bibinfo {author} {\bibfnamefont {M.}~\bibnamefont {Kleinmann}},\ and\ \bibinfo {author} {\bibfnamefont {J.-{\AA}.}\ \bibnamefont {Larsson}},\ }\bibfield  {title} {\bibinfo {title} {Kochen-{S}pecker contextuality},\ }\href {https://doi.org/10.1103/revmodphys.94.045007} {\bibfield  {journal} {\bibinfo  {journal} {Reviews of Modern Physics}\ }\textbf {\bibinfo {volume} {94}},\ \bibinfo {pages} {045007} (\bibinfo {year} {2022})}\BibitemShut {NoStop}%
\bibitem [{\citenamefont {Klyachko}\ \emph {et~al.}(2008)\citenamefont {Klyachko}, \citenamefont {Can}, \citenamefont {Binicio{\u{g}}lu},\ and\ \citenamefont {Shumovsky}}]{klyachko2008simple}%
  \BibitemOpen
  \bibfield  {author} {\bibinfo {author} {\bibfnamefont {A.~A.}\ \bibnamefont {Klyachko}}, \bibinfo {author} {\bibfnamefont {M.~A.}\ \bibnamefont {Can}}, \bibinfo {author} {\bibfnamefont {S.}~\bibnamefont {Binicio{\u{g}}lu}},\ and\ \bibinfo {author} {\bibfnamefont {A.~S.}\ \bibnamefont {Shumovsky}},\ }\bibfield  {title} {\bibinfo {title} {Simple test for hidden variables in spin-1 systems},\ }\href {https://doi.org/10.1103/physrevlett.101.020403} {\bibfield  {journal} {\bibinfo  {journal} {Physical Review Letters}\ }\textbf {\bibinfo {volume} {101}},\ \bibinfo {pages} {020403} (\bibinfo {year} {2008})}\BibitemShut {NoStop}%
\bibitem [{\citenamefont {Raussendorf}(2013)}]{raussendorf2013contextuality}%
  \BibitemOpen
  \bibfield  {author} {\bibinfo {author} {\bibfnamefont {R.}~\bibnamefont {Raussendorf}},\ }\bibfield  {title} {\bibinfo {title} {Contextuality in measurement-based quantum computation},\ }\href {https://doi.org/10.1103/physreva.88.022322} {\bibfield  {journal} {\bibinfo  {journal} {Physical Review A}\ }\textbf {\bibinfo {volume} {88}},\ \bibinfo {pages} {022322} (\bibinfo {year} {2013})}\BibitemShut {NoStop}%
\bibitem [{\citenamefont {Lapkiewicz}\ \emph {et~al.}(2011)\citenamefont {Lapkiewicz}, \citenamefont {Li}, \citenamefont {Schaeff}, \citenamefont {Langford}, \citenamefont {Ramelow}, \citenamefont {Wie{\'s}niak},\ and\ \citenamefont {Zeilinger}}]{lapkiewicz2011experimental}%
  \BibitemOpen
  \bibfield  {author} {\bibinfo {author} {\bibfnamefont {R.}~\bibnamefont {Lapkiewicz}}, \bibinfo {author} {\bibfnamefont {P.}~\bibnamefont {Li}}, \bibinfo {author} {\bibfnamefont {C.}~\bibnamefont {Schaeff}}, \bibinfo {author} {\bibfnamefont {N.~K.}\ \bibnamefont {Langford}}, \bibinfo {author} {\bibfnamefont {S.}~\bibnamefont {Ramelow}}, \bibinfo {author} {\bibfnamefont {M.}~\bibnamefont {Wie{\'s}niak}},\ and\ \bibinfo {author} {\bibfnamefont {A.}~\bibnamefont {Zeilinger}},\ }\bibfield  {title} {\bibinfo {title} {Experimental non-classicality of an indivisible quantum system},\ }\href {https://doi.org/10.1038/nature10119} {\bibfield  {journal} {\bibinfo  {journal} {Nature}\ }\textbf {\bibinfo {volume} {474}},\ \bibinfo {pages} {490} (\bibinfo {year} {2011})}\BibitemShut {NoStop}%
\bibitem [{\citenamefont {Um}\ \emph {et~al.}(2013)\citenamefont {Um}, \citenamefont {Zhang}, \citenamefont {Zhang}, \citenamefont {Wang}, \citenamefont {Yangchao}, \citenamefont {Deng}, \citenamefont {Duan},\ and\ \citenamefont {Kim}}]{um2013experimental}%
  \BibitemOpen
  \bibfield  {author} {\bibinfo {author} {\bibfnamefont {M.}~\bibnamefont {Um}}, \bibinfo {author} {\bibfnamefont {X.}~\bibnamefont {Zhang}}, \bibinfo {author} {\bibfnamefont {J.}~\bibnamefont {Zhang}}, \bibinfo {author} {\bibfnamefont {Y.}~\bibnamefont {Wang}}, \bibinfo {author} {\bibfnamefont {S.}~\bibnamefont {Yangchao}}, \bibinfo {author} {\bibfnamefont {D.~L.}\ \bibnamefont {Deng}}, \bibinfo {author} {\bibfnamefont {L.-M.}\ \bibnamefont {Duan}},\ and\ \bibinfo {author} {\bibfnamefont {K.}~\bibnamefont {Kim}},\ }\bibfield  {title} {\bibinfo {title} {Experimental certification of random numbers via quantum contextuality},\ }\href {https://doi.org/10.1038/srep01627} {\bibfield  {journal} {\bibinfo  {journal} {Scientific Reports}\ }\textbf {\bibinfo {volume} {3}},\ \bibinfo {pages} {1627} (\bibinfo {year} {2013})}\BibitemShut {NoStop}%
\bibitem [{\citenamefont {Jerger}\ \emph {et~al.}(2016)\citenamefont {Jerger} \emph {et~al.}}]{jerger2016contextuality}%
  \BibitemOpen
  \bibfield  {author} {\bibinfo {author} {\bibfnamefont {M.}~\bibnamefont {Jerger}} \emph {et~al.},\ }\bibfield  {title} {\bibinfo {title} {Contextuality without nonlocality in a superconducting quantum system},\ }\href {https://doi.org/10.1038/ncomms12930} {\bibfield  {journal} {\bibinfo  {journal} {Nature Communications}\ }\textbf {\bibinfo {volume} {7}},\ \bibinfo {pages} {12930} (\bibinfo {year} {2016})}\BibitemShut {NoStop}%
\bibitem [{\citenamefont {Zhan}\ \emph {et~al.}(2017)\citenamefont {Zhan}, \citenamefont {Kurzy{\'n}ski}, \citenamefont {Kaszlikowski}, \citenamefont {Wang}, \citenamefont {Bian}, \citenamefont {Zhang},\ and\ \citenamefont {Xue}}]{zhan2017experimental}%
  \BibitemOpen
  \bibfield  {author} {\bibinfo {author} {\bibfnamefont {X.}~\bibnamefont {Zhan}}, \bibinfo {author} {\bibfnamefont {P.}~\bibnamefont {Kurzy{\'n}ski}}, \bibinfo {author} {\bibfnamefont {D.}~\bibnamefont {Kaszlikowski}}, \bibinfo {author} {\bibfnamefont {K.}~\bibnamefont {Wang}}, \bibinfo {author} {\bibfnamefont {Z.}~\bibnamefont {Bian}}, \bibinfo {author} {\bibfnamefont {Y.}~\bibnamefont {Zhang}},\ and\ \bibinfo {author} {\bibfnamefont {P.}~\bibnamefont {Xue}},\ }\bibfield  {title} {\bibinfo {title} {Experimental detection of information deficit in a photonic contextuality scenario},\ }\href {https://doi.org/10.1103/physrevlett.119.220403} {\bibfield  {journal} {\bibinfo  {journal} {Physical Review Letters}\ }\textbf {\bibinfo {volume} {119}},\ \bibinfo {pages} {220403} (\bibinfo {year} {2017})}\BibitemShut {NoStop}%
\bibitem [{\citenamefont {Malinowski}\ \emph {et~al.}(2018)\citenamefont {Malinowski}, \citenamefont {Zhang}, \citenamefont {Leupold}, \citenamefont {Cabello}, \citenamefont {Alonso},\ and\ \citenamefont {Home}}]{malinowski2018probing}%
  \BibitemOpen
  \bibfield  {author} {\bibinfo {author} {\bibfnamefont {M.}~\bibnamefont {Malinowski}}, \bibinfo {author} {\bibfnamefont {C.}~\bibnamefont {Zhang}}, \bibinfo {author} {\bibfnamefont {F.~M.}\ \bibnamefont {Leupold}}, \bibinfo {author} {\bibfnamefont {A.}~\bibnamefont {Cabello}}, \bibinfo {author} {\bibfnamefont {J.}~\bibnamefont {Alonso}},\ and\ \bibinfo {author} {\bibfnamefont {J.~P.}\ \bibnamefont {Home}},\ }\bibfield  {title} {\bibinfo {title} {Probing the limits of correlations in an indivisible quantum system},\ }\href {https://doi.org/10.1103/physreva.98.050102} {\bibfield  {journal} {\bibinfo  {journal} {Physical Review A}\ }\textbf {\bibinfo {volume} {98}},\ \bibinfo {pages} {050102} (\bibinfo {year} {2018})}\BibitemShut {NoStop}%
\bibitem [{\citenamefont {Leupold}\ \emph {et~al.}(2018)\citenamefont {Leupold}, \citenamefont {Malinowski}, \citenamefont {Zhang}, \citenamefont {Negnevitsky}, \citenamefont {Alonso}, \citenamefont {Home},\ and\ \citenamefont {Cabello}}]{leupold2018sustained}%
  \BibitemOpen
  \bibfield  {author} {\bibinfo {author} {\bibfnamefont {F.~M.}\ \bibnamefont {Leupold}}, \bibinfo {author} {\bibfnamefont {M.}~\bibnamefont {Malinowski}}, \bibinfo {author} {\bibfnamefont {C.}~\bibnamefont {Zhang}}, \bibinfo {author} {\bibfnamefont {V.}~\bibnamefont {Negnevitsky}}, \bibinfo {author} {\bibfnamefont {J.}~\bibnamefont {Alonso}}, \bibinfo {author} {\bibfnamefont {J.~P.}\ \bibnamefont {Home}},\ and\ \bibinfo {author} {\bibfnamefont {A.}~\bibnamefont {Cabello}},\ }\bibfield  {title} {\bibinfo {title} {Sustained state-independent quantum contextual correlations from a single ion},\ }\href {https://doi.org/10.1103/physrevlett.120.180401} {\bibfield  {journal} {\bibinfo  {journal} {Physical Review Letters}\ }\textbf {\bibinfo {volume} {120}},\ \bibinfo {pages} {180401} (\bibinfo {year} {2018})}\BibitemShut {NoStop}%
\bibitem [{\citenamefont {Zhang}\ \emph {et~al.}(2019)\citenamefont {Zhang}, \citenamefont {Xu}, \citenamefont {Xie}, \citenamefont {Zhang}, \citenamefont {Smith}, \citenamefont {Kim},\ and\ \citenamefont {Zhang}}]{zhang2019experimental}%
  \BibitemOpen
  \bibfield  {author} {\bibinfo {author} {\bibfnamefont {A.}~\bibnamefont {Zhang}}, \bibinfo {author} {\bibfnamefont {H.}~\bibnamefont {Xu}}, \bibinfo {author} {\bibfnamefont {J.}~\bibnamefont {Xie}}, \bibinfo {author} {\bibfnamefont {H.}~\bibnamefont {Zhang}}, \bibinfo {author} {\bibfnamefont {B.~J.}\ \bibnamefont {Smith}}, \bibinfo {author} {\bibfnamefont {M.~S.}\ \bibnamefont {Kim}},\ and\ \bibinfo {author} {\bibfnamefont {L.}~\bibnamefont {Zhang}},\ }\bibfield  {title} {\bibinfo {title} {Experimental test of contextuality in quantum and classical systems},\ }\href {https://doi.org/10.1103/physrevlett.122.080401} {\bibfield  {journal} {\bibinfo  {journal} {Physical Review Letters}\ }\textbf {\bibinfo {volume} {122}},\ \bibinfo {pages} {080401} (\bibinfo {year} {2019})}\BibitemShut {NoStop}%
\bibitem [{\citenamefont {Um}\ \emph {et~al.}(2020)\citenamefont {Um}, \citenamefont {Zhao}, \citenamefont {Zhang}, \citenamefont {Wang}, \citenamefont {Wang}, \citenamefont {Qiao}, \citenamefont {Zhou}, \citenamefont {Ma},\ and\ \citenamefont {Kim}}]{um2020randomness}%
  \BibitemOpen
  \bibfield  {author} {\bibinfo {author} {\bibfnamefont {M.}~\bibnamefont {Um}}, \bibinfo {author} {\bibfnamefont {Q.}~\bibnamefont {Zhao}}, \bibinfo {author} {\bibfnamefont {J.}~\bibnamefont {Zhang}}, \bibinfo {author} {\bibfnamefont {P.}~\bibnamefont {Wang}}, \bibinfo {author} {\bibfnamefont {Y.}~\bibnamefont {Wang}}, \bibinfo {author} {\bibfnamefont {M.}~\bibnamefont {Qiao}}, \bibinfo {author} {\bibfnamefont {H.}~\bibnamefont {Zhou}}, \bibinfo {author} {\bibfnamefont {X.}~\bibnamefont {Ma}},\ and\ \bibinfo {author} {\bibfnamefont {K.}~\bibnamefont {Kim}},\ }\bibfield  {title} {\bibinfo {title} {Randomness expansion secured by quantum contextuality},\ }\href {https://doi.org/10.1103/physrevapplied.13.034077} {\bibfield  {journal} {\bibinfo  {journal} {Physical Review Applied}\ }\textbf {\bibinfo {volume} {13}},\ \bibinfo {pages} {034077} (\bibinfo {year} {2020})}\BibitemShut {NoStop}%
\bibitem [{\citenamefont {Wang}\ \emph {et~al.}(2022)\citenamefont {Wang}, \citenamefont {Zhang}, \citenamefont {Luan}, \citenamefont {Um}, \citenamefont {Wang}, \citenamefont {Qiao}, \citenamefont {Xie}, \citenamefont {Zhang}, \citenamefont {Cabello},\ and\ \citenamefont {Kim}}]{wang2022significant}%
  \BibitemOpen
  \bibfield  {author} {\bibinfo {author} {\bibfnamefont {P.}~\bibnamefont {Wang}}, \bibinfo {author} {\bibfnamefont {J.}~\bibnamefont {Zhang}}, \bibinfo {author} {\bibfnamefont {C.-Y.}\ \bibnamefont {Luan}}, \bibinfo {author} {\bibfnamefont {M.}~\bibnamefont {Um}}, \bibinfo {author} {\bibfnamefont {Y.}~\bibnamefont {Wang}}, \bibinfo {author} {\bibfnamefont {M.}~\bibnamefont {Qiao}}, \bibinfo {author} {\bibfnamefont {T.}~\bibnamefont {Xie}}, \bibinfo {author} {\bibfnamefont {J.-N.}\ \bibnamefont {Zhang}}, \bibinfo {author} {\bibfnamefont {A.}~\bibnamefont {Cabello}},\ and\ \bibinfo {author} {\bibfnamefont {K.}~\bibnamefont {Kim}},\ }\bibfield  {title} {\bibinfo {title} {Significant loophole-free test of {K}ochen-{S}pecker contextuality using two species of atomic ions},\ }\bibfield  {journal} {\bibinfo  {journal} {Science Advances}\ }\textbf {\bibinfo {volume} {8}},\ \href {https://doi.org/10.1126/sciadv.abk1660} {10.1126/sciadv.abk1660} (\bibinfo {year} {2022})\BibitemShut {NoStop}%
\bibitem [{\citenamefont {Hu}\ \emph {et~al.}(2023)\citenamefont {Hu} \emph {et~al.}}]{hu2023self}%
  \BibitemOpen
  \bibfield  {author} {\bibinfo {author} {\bibfnamefont {X.-M.}\ \bibnamefont {Hu}} \emph {et~al.},\ }\bibfield  {title} {\bibinfo {title} {Self-testing of a single quantum system from theory to experiment},\ }\href {https://doi.org/10.1038/s41534-023-00769-7} {\bibfield  {journal} {\bibinfo  {journal} {npj Quantum Information}\ }\textbf {\bibinfo {volume} {9}},\ \bibinfo {pages} {103} (\bibinfo {year} {2023})}\BibitemShut {NoStop}%
\bibitem [{\citenamefont {Liu}\ \emph {et~al.}(2023)\citenamefont {Liu}, \citenamefont {Meng}, \citenamefont {Xu}, \citenamefont {Zhou}, \citenamefont {Chen}, \citenamefont {Xu}, \citenamefont {Li}, \citenamefont {Guo},\ and\ \citenamefont {Cabello}}]{liu2023experimental}%
  \BibitemOpen
  \bibfield  {author} {\bibinfo {author} {\bibfnamefont {Z.-H.}\ \bibnamefont {Liu}}, \bibinfo {author} {\bibfnamefont {H.-X.}\ \bibnamefont {Meng}}, \bibinfo {author} {\bibfnamefont {Z.-P.}\ \bibnamefont {Xu}}, \bibinfo {author} {\bibfnamefont {J.}~\bibnamefont {Zhou}}, \bibinfo {author} {\bibfnamefont {J.-L.}\ \bibnamefont {Chen}}, \bibinfo {author} {\bibfnamefont {J.-S.}\ \bibnamefont {Xu}}, \bibinfo {author} {\bibfnamefont {C.-F.}\ \bibnamefont {Li}}, \bibinfo {author} {\bibfnamefont {G.-C.}\ \bibnamefont {Guo}},\ and\ \bibinfo {author} {\bibfnamefont {A.}~\bibnamefont {Cabello}},\ }\bibfield  {title} {\bibinfo {title} {Experimental test of high-dimensional quantum contextuality based on contextuality concentration},\ }\href {https://doi.org/10.1103/physrevlett.130.240202} {\bibfield  {journal} {\bibinfo  {journal} {Physical Review Letters}\ }\textbf {\bibinfo {volume} {130}},\ \bibinfo {pages} {240202} (\bibinfo {year} {2023})}\BibitemShut {NoStop}%
\bibitem [{\citenamefont {Bharti}\ \emph {et~al.}(2019{\natexlab{a}})\citenamefont {Bharti}, \citenamefont {Ray}, \citenamefont {Varvitsiotis}, \citenamefont {Warsi}, \citenamefont {Cabello},\ and\ \citenamefont {Kwek}}]{bharti2019robust}%
  \BibitemOpen
  \bibfield  {author} {\bibinfo {author} {\bibfnamefont {K.}~\bibnamefont {Bharti}}, \bibinfo {author} {\bibfnamefont {M.}~\bibnamefont {Ray}}, \bibinfo {author} {\bibfnamefont {A.}~\bibnamefont {Varvitsiotis}}, \bibinfo {author} {\bibfnamefont {N.~A.}\ \bibnamefont {Warsi}}, \bibinfo {author} {\bibfnamefont {A.}~\bibnamefont {Cabello}},\ and\ \bibinfo {author} {\bibfnamefont {L.-C.}\ \bibnamefont {Kwek}},\ }\bibfield  {title} {\bibinfo {title} {Robust self-testing of quantum systems via noncontextuality inequalities},\ }\href {https://doi.org/10.1103/physrevlett.122.250403} {\bibfield  {journal} {\bibinfo  {journal} {Physical Review Letters}\ }\textbf {\bibinfo {volume} {122}},\ \bibinfo {pages} {250403} (\bibinfo {year} {2019}{\natexlab{a}})}\BibitemShut {NoStop}%
\bibitem [{\citenamefont {Bharti}\ \emph {et~al.}(2019{\natexlab{b}})\citenamefont {Bharti}, \citenamefont {Ray}, \citenamefont {Varvitsiotis}, \citenamefont {Cabello},\ and\ \citenamefont {Kwek}}]{bharti2019local}%
  \BibitemOpen
  \bibfield  {author} {\bibinfo {author} {\bibfnamefont {K.}~\bibnamefont {Bharti}}, \bibinfo {author} {\bibfnamefont {M.}~\bibnamefont {Ray}}, \bibinfo {author} {\bibfnamefont {A.}~\bibnamefont {Varvitsiotis}}, \bibinfo {author} {\bibfnamefont {A.}~\bibnamefont {Cabello}},\ and\ \bibinfo {author} {\bibfnamefont {L.-C.}\ \bibnamefont {Kwek}},\ }\href@noop {} {\bibinfo {title} {Local certification of programmable quantum devices of arbitrary high dimensionality}} (\bibinfo {year} {2019}{\natexlab{b}}),\ \Eprint {https://arxiv.org/abs/1911.09448} {arXiv:1911.09448 [quant-ph]} \BibitemShut {NoStop}%
\bibitem [{\citenamefont {Xu}\ \emph {et~al.}(2024)\citenamefont {Xu}, \citenamefont {Saha}, \citenamefont {Bharti},\ and\ \citenamefont {Cabello}}]{xu2024certifying}%
  \BibitemOpen
  \bibfield  {author} {\bibinfo {author} {\bibfnamefont {Z.-P.}\ \bibnamefont {Xu}}, \bibinfo {author} {\bibfnamefont {D.}~\bibnamefont {Saha}}, \bibinfo {author} {\bibfnamefont {K.}~\bibnamefont {Bharti}},\ and\ \bibinfo {author} {\bibfnamefont {A.}~\bibnamefont {Cabello}},\ }\bibfield  {title} {\bibinfo {title} {Certifying sets of quantum observables with any full-rank state},\ }\href {https://doi.org/10.1103/physrevlett.132.140201} {\bibfield  {journal} {\bibinfo  {journal} {Physical Review Letters}\ }\textbf {\bibinfo {volume} {132}},\ \bibinfo {pages} {140201} (\bibinfo {year} {2024})}\BibitemShut {NoStop}%
\bibitem [{\citenamefont {Saha}\ \emph {et~al.}(2020)\citenamefont {Saha}, \citenamefont {Santos},\ and\ \citenamefont {Augusiak}}]{saha2020sum}%
  \BibitemOpen
  \bibfield  {author} {\bibinfo {author} {\bibfnamefont {D.}~\bibnamefont {Saha}}, \bibinfo {author} {\bibfnamefont {R.}~\bibnamefont {Santos}},\ and\ \bibinfo {author} {\bibfnamefont {R.}~\bibnamefont {Augusiak}},\ }\bibfield  {title} {\bibinfo {title} {Sum-of-squares decompositions for a family of noncontextuality inequalities and self-testing of quantum devices},\ }\href {https://doi.org/10.22331/q-2020-08-03-302} {\bibfield  {journal} {\bibinfo  {journal} {Quantum}\ }\textbf {\bibinfo {volume} {4}},\ \bibinfo {pages} {302} (\bibinfo {year} {2020})}\BibitemShut {NoStop}%
\bibitem [{\citenamefont {Liu}(2023)}]{liu2023exploring}%
  \BibitemOpen
  \bibfield  {author} {\bibinfo {author} {\bibfnamefont {Z.-H.}\ \bibnamefont {Liu}},\ }\href {https://doi.org/10.1007/978-981-99-6167-2} {\emph {\bibinfo {title} {Exploring Quantum Contextuality with Photons}}}\ (\bibinfo  {publisher} {Springer Nature Singapore},\ \bibinfo {year} {2023})\BibitemShut {NoStop}%
\bibitem [{\citenamefont {Arora}\ \emph {et~al.}(2024)\citenamefont {Arora}, \citenamefont {Bharti}, \citenamefont {Cojocaru},\ and\ \citenamefont {Coladangelo}}]{arora2024computational}%
  \BibitemOpen
  \bibfield  {author} {\bibinfo {author} {\bibfnamefont {A.~S.}\ \bibnamefont {Arora}}, \bibinfo {author} {\bibfnamefont {K.}~\bibnamefont {Bharti}}, \bibinfo {author} {\bibfnamefont {A.}~\bibnamefont {Cojocaru}},\ and\ \bibinfo {author} {\bibfnamefont {A.}~\bibnamefont {Coladangelo}},\ }\bibfield  {title} {\bibinfo {title} {A computational test of contextuality and, even simpler proofs of quantumness},\ }in\ \href {https://doi.org/10.1109/focs61266.2024.00073} {\emph {\bibinfo {booktitle} {2024 IEEE 65th Annual Symposium on Foundations of Computer Science (FOCS)}}}\ (\bibinfo  {publisher} {IEEE},\ \bibinfo {year} {2024})\ pp.\ \bibinfo {pages} {1106--1125}\BibitemShut {NoStop}%
\bibitem [{\citenamefont {Bowles}\ \emph {et~al.}(2023)\citenamefont {Bowles}, \citenamefont {Wright}, \citenamefont {Farkas}, \citenamefont {Killoran},\ and\ \citenamefont {Schuld}}]{bowles2023contextuality}%
  \BibitemOpen
  \bibfield  {author} {\bibinfo {author} {\bibfnamefont {J.}~\bibnamefont {Bowles}}, \bibinfo {author} {\bibfnamefont {V.~J.}\ \bibnamefont {Wright}}, \bibinfo {author} {\bibfnamefont {M.}~\bibnamefont {Farkas}}, \bibinfo {author} {\bibfnamefont {N.}~\bibnamefont {Killoran}},\ and\ \bibinfo {author} {\bibfnamefont {M.}~\bibnamefont {Schuld}},\ }\href@noop {} {\bibinfo {title} {Contextuality and inductive bias in quantum machine learning}} (\bibinfo {year} {2023}),\ \Eprint {https://arxiv.org/abs/2302.01365} {arXiv:2302.01365 [quant-ph]} \BibitemShut {NoStop}%
\bibitem [{\citenamefont {Gao}\ \emph {et~al.}(2022)\citenamefont {Gao}, \citenamefont {Anschuetz}, \citenamefont {Wang}, \citenamefont {Cirac},\ and\ \citenamefont {Lukin}}]{gao2022enhancing}%
  \BibitemOpen
  \bibfield  {author} {\bibinfo {author} {\bibfnamefont {X.}~\bibnamefont {Gao}}, \bibinfo {author} {\bibfnamefont {E.~R.}\ \bibnamefont {Anschuetz}}, \bibinfo {author} {\bibfnamefont {S.-T.}\ \bibnamefont {Wang}}, \bibinfo {author} {\bibfnamefont {J.~I.}\ \bibnamefont {Cirac}},\ and\ \bibinfo {author} {\bibfnamefont {M.~D.}\ \bibnamefont {Lukin}},\ }\bibfield  {title} {\bibinfo {title} {Enhancing generative models via quantum correlations},\ }\href {https://doi.org/10.1103/physrevx.12.021037} {\bibfield  {journal} {\bibinfo  {journal} {Physical Review X}\ }\textbf {\bibinfo {volume} {12}},\ \bibinfo {pages} {021037} (\bibinfo {year} {2022})}\BibitemShut {NoStop}%
\bibitem [{\citenamefont {Hameedi}\ \emph {et~al.}(2017)\citenamefont {Hameedi}, \citenamefont {Tavakoli}, \citenamefont {Marques},\ and\ \citenamefont {Bourennane}}]{hameedi2017communication}%
  \BibitemOpen
  \bibfield  {author} {\bibinfo {author} {\bibfnamefont {A.}~\bibnamefont {Hameedi}}, \bibinfo {author} {\bibfnamefont {A.}~\bibnamefont {Tavakoli}}, \bibinfo {author} {\bibfnamefont {B.}~\bibnamefont {Marques}},\ and\ \bibinfo {author} {\bibfnamefont {M.}~\bibnamefont {Bourennane}},\ }\bibfield  {title} {\bibinfo {title} {Communication games reveal preparation contextuality},\ }\href {https://doi.org/10.1103/physrevlett.119.220402} {\bibfield  {journal} {\bibinfo  {journal} {Physical Review Letters}\ }\textbf {\bibinfo {volume} {119}},\ \bibinfo {pages} {220402} (\bibinfo {year} {2017})}\BibitemShut {NoStop}%
\bibitem [{\citenamefont {Saha}\ \emph {et~al.}(2019)\citenamefont {Saha}, \citenamefont {Horodecki},\ and\ \citenamefont {Paw{\l}owski}}]{saha2019state}%
  \BibitemOpen
  \bibfield  {author} {\bibinfo {author} {\bibfnamefont {D.}~\bibnamefont {Saha}}, \bibinfo {author} {\bibfnamefont {P.}~\bibnamefont {Horodecki}},\ and\ \bibinfo {author} {\bibfnamefont {M.}~\bibnamefont {Paw{\l}owski}},\ }\bibfield  {title} {\bibinfo {title} {State independent contextuality advances one-way communication},\ }\href {https://doi.org/10.1088/1367-2630/ab4149} {\bibfield  {journal} {\bibinfo  {journal} {New Journal of Physics}\ }\textbf {\bibinfo {volume} {21}},\ \bibinfo {pages} {093057} (\bibinfo {year} {2019})}\BibitemShut {NoStop}%
\bibitem [{\citenamefont {Gupta}\ \emph {et~al.}(2023)\citenamefont {Gupta}, \citenamefont {Saha}, \citenamefont {Xu}, \citenamefont {Cabello},\ and\ \citenamefont {Majumdar}}]{gupta2023quantum}%
  \BibitemOpen
  \bibfield  {author} {\bibinfo {author} {\bibfnamefont {S.}~\bibnamefont {Gupta}}, \bibinfo {author} {\bibfnamefont {D.}~\bibnamefont {Saha}}, \bibinfo {author} {\bibfnamefont {Z.-P.}\ \bibnamefont {Xu}}, \bibinfo {author} {\bibfnamefont {A.}~\bibnamefont {Cabello}},\ and\ \bibinfo {author} {\bibfnamefont {A.~S.}\ \bibnamefont {Majumdar}},\ }\bibfield  {title} {\bibinfo {title} {Quantum contextuality provides communication complexity advantage},\ }\href {https://doi.org/10.1103/physrevlett.130.080802} {\bibfield  {journal} {\bibinfo  {journal} {Physical Review Letters}\ }\textbf {\bibinfo {volume} {130}},\ \bibinfo {pages} {080802} (\bibinfo {year} {2023})}\BibitemShut {NoStop}%
\bibitem [{\citenamefont {Mermin}(1990)}]{mermin1990extreme}%
  \BibitemOpen
  \bibfield  {author} {\bibinfo {author} {\bibfnamefont {N.~D.}\ \bibnamefont {Mermin}},\ }\bibfield  {title} {\bibinfo {title} {Extreme quantum entanglement in a superposition of macroscopically distinct states},\ }\href {https://doi.org/10.1103/physrevlett.65.1838} {\bibfield  {journal} {\bibinfo  {journal} {Physical Review Letters}\ }\textbf {\bibinfo {volume} {65}},\ \bibinfo {pages} {1838} (\bibinfo {year} {1990})}\BibitemShut {NoStop}%
\bibitem [{\citenamefont {Abramsky}\ and\ \citenamefont {Brandenburger}(2011)}]{abramsky2011sheaf}%
  \BibitemOpen
  \bibfield  {author} {\bibinfo {author} {\bibfnamefont {S.}~\bibnamefont {Abramsky}}\ and\ \bibinfo {author} {\bibfnamefont {A.}~\bibnamefont {Brandenburger}},\ }\bibfield  {title} {\bibinfo {title} {The sheaf-theoretic structure of non-locality and contextuality},\ }\href {https://doi.org/10.1088/1367-2630/13/11/113036} {\bibfield  {journal} {\bibinfo  {journal} {New Journal of Physics}\ }\textbf {\bibinfo {volume} {13}},\ \bibinfo {pages} {113036} (\bibinfo {year} {2011})}\BibitemShut {NoStop}%
\bibitem [{\citenamefont {Abramsky}\ \emph {et~al.}(2015)\citenamefont {Abramsky}, \citenamefont {Soares~Barbosa}, \citenamefont {Kishida}, \citenamefont {Lal},\ and\ \citenamefont {Mansfield}}]{abramsky2015contextuality}%
  \BibitemOpen
  \bibfield  {author} {\bibinfo {author} {\bibfnamefont {S.}~\bibnamefont {Abramsky}}, \bibinfo {author} {\bibfnamefont {R.}~\bibnamefont {Soares~Barbosa}}, \bibinfo {author} {\bibfnamefont {K.}~\bibnamefont {Kishida}}, \bibinfo {author} {\bibfnamefont {R.}~\bibnamefont {Lal}},\ and\ \bibinfo {author} {\bibfnamefont {S.}~\bibnamefont {Mansfield}},\ }\bibfield  {title} {\bibinfo {title} {Contextuality, cohomology and paradox},\ }in\ \href {https://doi.org/10.4230/LIPICS.CSL.2015.211} {\emph {\bibinfo {booktitle} {24th EACSL Annual Conference on Computer Science Logic (CSL 2015)}}}\ (\bibinfo  {publisher} {Schloss Dagstuhl – Leibniz-Zentrum f{\"u}r Informatik},\ \bibinfo {year} {2015})\BibitemShut {NoStop}%
\bibitem [{\citenamefont {Kirby}\ and\ \citenamefont {Love}(2019)}]{kirby2019contextuality}%
  \BibitemOpen
  \bibfield  {author} {\bibinfo {author} {\bibfnamefont {W.~M.}\ \bibnamefont {Kirby}}\ and\ \bibinfo {author} {\bibfnamefont {P.~J.}\ \bibnamefont {Love}},\ }\bibfield  {title} {\bibinfo {title} {Contextuality test of the nonclassicality of variational quantum eigensolvers},\ }\href {https://doi.org/10.1103/physrevlett.123.200501} {\bibfield  {journal} {\bibinfo  {journal} {Physical Review Letters}\ }\textbf {\bibinfo {volume} {123}},\ \bibinfo {pages} {200501} (\bibinfo {year} {2019})}\BibitemShut {NoStop}%
\bibitem [{\citenamefont {Kim}\ and\ \citenamefont {Abramsky}(2023)}]{kim2023stateindependentallversusnothingarguments}%
  \BibitemOpen
  \bibfield  {author} {\bibinfo {author} {\bibfnamefont {B.}~\bibnamefont {Kim}}\ and\ \bibinfo {author} {\bibfnamefont {S.}~\bibnamefont {Abramsky}},\ }\href@noop {} {\bibinfo {title} {State-independent all-versus-nothing arguments}} (\bibinfo {year} {2023}),\ \Eprint {https://arxiv.org/abs/2311.11218} {arXiv:2311.11218 [quant-ph]} \BibitemShut {NoStop}%
\bibitem [{\citenamefont {Butt}\ \emph {et~al.}(2024)\citenamefont {Butt}, \citenamefont {Heu{\ss}en}, \citenamefont {Rispler},\ and\ \citenamefont {M{\"u}ller}}]{butt2024fault}%
  \BibitemOpen
  \bibfield  {author} {\bibinfo {author} {\bibfnamefont {F.}~\bibnamefont {Butt}}, \bibinfo {author} {\bibfnamefont {S.}~\bibnamefont {Heu{\ss}en}}, \bibinfo {author} {\bibfnamefont {M.}~\bibnamefont {Rispler}},\ and\ \bibinfo {author} {\bibfnamefont {M.}~\bibnamefont {M{\"u}ller}},\ }\bibfield  {title} {\bibinfo {title} {Fault-tolerant code-switching protocols for near-term quantum processors},\ }\href {https://doi.org/10.1103/prxquantum.5.020345} {\bibfield  {journal} {\bibinfo  {journal} {PRX Quantum}\ }\textbf {\bibinfo {volume} {5}},\ \bibinfo {pages} {020345} (\bibinfo {year} {2024})}\BibitemShut {NoStop}%
\bibitem [{\citenamefont {Bravyi}\ and\ \citenamefont {Cross}(2015)}]{bravyi2015doubled}%
  \BibitemOpen
  \bibfield  {author} {\bibinfo {author} {\bibfnamefont {S.}~\bibnamefont {Bravyi}}\ and\ \bibinfo {author} {\bibfnamefont {A.}~\bibnamefont {Cross}},\ }\href@noop {} {\bibinfo {title} {Doubled color codes}} (\bibinfo {year} {2015}),\ \Eprint {https://arxiv.org/abs/1509.03239} {arXiv:1509.03239 [quant-ph]} \BibitemShut {NoStop}%
\bibitem [{\citenamefont {Chitambar}\ and\ \citenamefont {Gour}(2019)}]{chitambar2019quantum}%
  \BibitemOpen
  \bibfield  {author} {\bibinfo {author} {\bibfnamefont {E.}~\bibnamefont {Chitambar}}\ and\ \bibinfo {author} {\bibfnamefont {G.}~\bibnamefont {Gour}},\ }\bibfield  {title} {\bibinfo {title} {Quantum resource theories},\ }\href {https://doi.org/10.1103/revmodphys.91.025001} {\bibfield  {journal} {\bibinfo  {journal} {Reviews of Modern Physics}\ }\textbf {\bibinfo {volume} {91}},\ \bibinfo {pages} {025001} (\bibinfo {year} {2019})}\BibitemShut {NoStop}%
\bibitem [{\citenamefont {Bell}(1964)}]{bell1964einstein}%
  \BibitemOpen
  \bibfield  {author} {\bibinfo {author} {\bibfnamefont {J.~S.}\ \bibnamefont {Bell}},\ }\bibfield  {title} {\bibinfo {title} {On the {E}instein {P}odolsky {R}osen paradox},\ }\href {https://doi.org/10.1103/physicsphysiquefizika.1.195} {\bibfield  {journal} {\bibinfo  {journal} {Physics Physique Fizika}\ }\textbf {\bibinfo {volume} {1}},\ \bibinfo {pages} {195} (\bibinfo {year} {1964})}\BibitemShut {NoStop}%
\bibitem [{\citenamefont {Greenberger}\ \emph {et~al.}(1989)\citenamefont {Greenberger}, \citenamefont {Horne},\ and\ \citenamefont {Zeilinger}}]{greenberger1989going}%
  \BibitemOpen
  \bibfield  {author} {\bibinfo {author} {\bibfnamefont {D.~M.}\ \bibnamefont {Greenberger}}, \bibinfo {author} {\bibfnamefont {M.~A.}\ \bibnamefont {Horne}},\ and\ \bibinfo {author} {\bibfnamefont {A.}~\bibnamefont {Zeilinger}},\ }\bibinfo {title} {Going beyond {B}ell's theorem},\ in\ \href {https://doi.org/10.1007/978-94-017-0849-4_10} {\emph {\bibinfo {booktitle} {Bell’s Theorem, Quantum Theory and Conceptions of the Universe}}}\ (\bibinfo  {publisher} {Springer Netherlands},\ \bibinfo {year} {1989})\ pp.\ \bibinfo {pages} {69--72}\BibitemShut {NoStop}%
\bibitem [{\citenamefont {Greenberger}\ \emph {et~al.}(1990)\citenamefont {Greenberger}, \citenamefont {Horne}, \citenamefont {Shimony},\ and\ \citenamefont {Zeilinger}}]{greenberger1990bell}%
  \BibitemOpen
  \bibfield  {author} {\bibinfo {author} {\bibfnamefont {D.~M.}\ \bibnamefont {Greenberger}}, \bibinfo {author} {\bibfnamefont {M.~A.}\ \bibnamefont {Horne}}, \bibinfo {author} {\bibfnamefont {A.}~\bibnamefont {Shimony}},\ and\ \bibinfo {author} {\bibfnamefont {A.}~\bibnamefont {Zeilinger}},\ }\bibfield  {title} {\bibinfo {title} {Bell's theorem without inequalities},\ }\href {https://doi.org/10.1119/1.16243} {\bibfield  {journal} {\bibinfo  {journal} {American Journal of Physics}\ }\textbf {\bibinfo {volume} {58}},\ \bibinfo {pages} {1131} (\bibinfo {year} {1990})}\BibitemShut {NoStop}%
\bibitem [{\citenamefont {DiVincenzo}\ and\ \citenamefont {Peres}(1997)}]{divincenzo1997quantum}%
  \BibitemOpen
  \bibfield  {author} {\bibinfo {author} {\bibfnamefont {D.~P.}\ \bibnamefont {DiVincenzo}}\ and\ \bibinfo {author} {\bibfnamefont {A.}~\bibnamefont {Peres}},\ }\bibfield  {title} {\bibinfo {title} {Quantum code words contradict local realism},\ }\href {https://doi.org/10.1103/physreva.55.4089} {\bibfield  {journal} {\bibinfo  {journal} {Physical Review A}\ }\textbf {\bibinfo {volume} {55}},\ \bibinfo {pages} {4089} (\bibinfo {year} {1997})}\BibitemShut {NoStop}%
\bibitem [{\citenamefont {Cabello}\ \emph {et~al.}(2008)\citenamefont {Cabello}, \citenamefont {G{\"u}hne},\ and\ \citenamefont {Rodr{\'i}guez}}]{cabello2008mermin}%
  \BibitemOpen
  \bibfield  {author} {\bibinfo {author} {\bibfnamefont {A.}~\bibnamefont {Cabello}}, \bibinfo {author} {\bibfnamefont {O.}~\bibnamefont {G{\"u}hne}},\ and\ \bibinfo {author} {\bibfnamefont {D.}~\bibnamefont {Rodr{\'i}guez}},\ }\bibfield  {title} {\bibinfo {title} {Mermin inequalities for perfect correlations},\ }\href {https://doi.org/10.1103/physreva.77.062106} {\bibfield  {journal} {\bibinfo  {journal} {Physical Review A}\ }\textbf {\bibinfo {volume} {77}},\ \bibinfo {pages} {062106} (\bibinfo {year} {2008})}\BibitemShut {NoStop}%
\bibitem [{\citenamefont {Abramsky}\ \emph {et~al.}(2017)\citenamefont {Abramsky}, \citenamefont {Barbosa}, \citenamefont {Car{\`u}},\ and\ \citenamefont {Perdrix}}]{abramsky2017complete}%
  \BibitemOpen
  \bibfield  {author} {\bibinfo {author} {\bibfnamefont {S.}~\bibnamefont {Abramsky}}, \bibinfo {author} {\bibfnamefont {R.~S.}\ \bibnamefont {Barbosa}}, \bibinfo {author} {\bibfnamefont {G.}~\bibnamefont {Car{\`u}}},\ and\ \bibinfo {author} {\bibfnamefont {S.}~\bibnamefont {Perdrix}},\ }\bibfield  {title} {\bibinfo {title} {A complete characterization of all-versus-nothing arguments for stabilizer states},\ }\href {https://doi.org/10.1098/rsta.2016.0385} {\bibfield  {journal} {\bibinfo  {journal} {Philosophical Transactions of the Royal Society A: Mathematical, Physical and Engineering Sciences}\ }\textbf {\bibinfo {volume} {375}},\ \bibinfo {pages} {20160385} (\bibinfo {year} {2017})}\BibitemShut {NoStop}%
\bibitem [{\citenamefont {Howard}\ \emph {et~al.}(2014)\citenamefont {Howard}, \citenamefont {Wallman}, \citenamefont {Veitch},\ and\ \citenamefont {Emerson}}]{howard2014contextuality}%
  \BibitemOpen
  \bibfield  {author} {\bibinfo {author} {\bibfnamefont {M.}~\bibnamefont {Howard}}, \bibinfo {author} {\bibfnamefont {J.}~\bibnamefont {Wallman}}, \bibinfo {author} {\bibfnamefont {V.}~\bibnamefont {Veitch}},\ and\ \bibinfo {author} {\bibfnamefont {J.}~\bibnamefont {Emerson}},\ }\bibfield  {title} {\bibinfo {title} {Contextuality supplies the `magic' for quantum computation},\ }\href {https://doi.org/10.1038/nature13460} {\bibfield  {journal} {\bibinfo  {journal} {Nature}\ }\textbf {\bibinfo {volume} {510}},\ \bibinfo {pages} {351} (\bibinfo {year} {2014})}\BibitemShut {NoStop}%
\bibitem [{\citenamefont {Cabello}\ \emph {et~al.}(2010)\citenamefont {Cabello}, \citenamefont {Severini},\ and\ \citenamefont {Winter}}]{cabello2010noncontextuality}%
  \BibitemOpen
  \bibfield  {author} {\bibinfo {author} {\bibfnamefont {A.}~\bibnamefont {Cabello}}, \bibinfo {author} {\bibfnamefont {S.}~\bibnamefont {Severini}},\ and\ \bibinfo {author} {\bibfnamefont {A.}~\bibnamefont {Winter}},\ }\href@noop {} {\bibinfo {title} {{(Non-)Contextuality} of physical theories as an axiom}} (\bibinfo {year} {2010}),\ \Eprint {https://arxiv.org/abs/1010.2163} {arXiv:1010.2163 [quant-ph]} \BibitemShut {NoStop}%
\bibitem [{\citenamefont {Cabello}\ \emph {et~al.}(2014)\citenamefont {Cabello}, \citenamefont {Severini},\ and\ \citenamefont {Winter}}]{cabello2014graph}%
  \BibitemOpen
  \bibfield  {author} {\bibinfo {author} {\bibfnamefont {A.}~\bibnamefont {Cabello}}, \bibinfo {author} {\bibfnamefont {S.}~\bibnamefont {Severini}},\ and\ \bibinfo {author} {\bibfnamefont {A.}~\bibnamefont {Winter}},\ }\bibfield  {title} {\bibinfo {title} {Graph-theoretic approach to quantum correlations},\ }\href {https://doi.org/10.1103/physrevlett.112.040401} {\bibfield  {journal} {\bibinfo  {journal} {Physical Review Letters}\ }\textbf {\bibinfo {volume} {112}},\ \bibinfo {pages} {040401} (\bibinfo {year} {2014})}\BibitemShut {NoStop}%
\bibitem [{\citenamefont {Kirby}\ \emph {et~al.}(2021)\citenamefont {Kirby}, \citenamefont {Tranter},\ and\ \citenamefont {Love}}]{kirby2021contextual}%
  \BibitemOpen
  \bibfield  {author} {\bibinfo {author} {\bibfnamefont {W.~M.}\ \bibnamefont {Kirby}}, \bibinfo {author} {\bibfnamefont {A.}~\bibnamefont {Tranter}},\ and\ \bibinfo {author} {\bibfnamefont {P.~J.}\ \bibnamefont {Love}},\ }\bibfield  {title} {\bibinfo {title} {Contextual subspace variational quantum eigensolver},\ }\href {https://doi.org/10.22331/q-2021-05-14-456} {\bibfield  {journal} {\bibinfo  {journal} {Quantum}\ }\textbf {\bibinfo {volume} {5}},\ \bibinfo {pages} {456} (\bibinfo {year} {2021})}\BibitemShut {NoStop}%
\bibitem [{\citenamefont {Abramsky}\ \emph {et~al.}(2024)\citenamefont {Abramsky}, \citenamefont {Cercelescu},\ and\ \citenamefont {Constantin}}]{abramsky2024commutation}%
  \BibitemOpen
  \bibfield  {author} {\bibinfo {author} {\bibfnamefont {S.}~\bibnamefont {Abramsky}}, \bibinfo {author} {\bibfnamefont {{\c{S}}.-I.}\ \bibnamefont {Cercelescu}},\ and\ \bibinfo {author} {\bibfnamefont {C.-M.}\ \bibnamefont {Constantin}},\ }\bibfield  {title} {\bibinfo {title} {{Commutation Groups and State-Independent Contextuality}},\ }in\ \href {https://doi.org/10.4230/LIPIcs.FSCD.2024.28} {\emph {\bibinfo {booktitle} {9th International Conference on Formal Structures for Computation and Deduction (FSCD 2024)}}},\ \bibinfo {series} {Leibniz International Proceedings in Informatics (LIPIcs)}, Vol.\ \bibinfo {volume} {299},\ \bibinfo {editor} {edited by\ \bibinfo {editor} {\bibfnamefont {J.}~\bibnamefont {Rehof}}}\ (\bibinfo  {publisher} {Schloss Dagstuhl -- Leibniz-Zentrum f{\"u}r Informatik},\ \bibinfo {address} {Dagstuhl, Germany},\ \bibinfo {year} {2024})\ pp.\ \bibinfo {pages} {28:1--28:20}\BibitemShut {NoStop}%
\bibitem [{\citenamefont {Gottesman}(2024)}]{gottesman2016surviving}%
  \BibitemOpen
  \bibfield  {author} {\bibinfo {author} {\bibfnamefont {D.}~\bibnamefont {Gottesman}},\ }\bibfield  {title} {\bibinfo {title} {Surviving as a quantum computer in a classical world}} (\bibinfo {year} {2024}),\ \bibinfo {note} {textbook manuscript preprint available at \url{https://www.cs.umd.edu/class/spring2024/cmsc858G/QECCbook-2024-ch1-15.pdf}}\BibitemShut {NoStop}%
\bibitem [{\citenamefont {Kribs}\ \emph {et~al.}(2006)\citenamefont {Kribs}, \citenamefont {Laflamme}, \citenamefont {Poulin},\ and\ \citenamefont {Lesosky}}]{kribs2005operator}%
  \BibitemOpen
  \bibfield  {author} {\bibinfo {author} {\bibfnamefont {D.}~\bibnamefont {Kribs}}, \bibinfo {author} {\bibfnamefont {R.}~\bibnamefont {Laflamme}}, \bibinfo {author} {\bibfnamefont {D.}~\bibnamefont {Poulin}},\ and\ \bibinfo {author} {\bibfnamefont {M.}~\bibnamefont {Lesosky}},\ }\bibfield  {title} {\bibinfo {title} {Operator quantum error correction},\ }\href {https://doi.org/10.26421/qic6.4-5-6} {\bibfield  {journal} {\bibinfo  {journal} {Quantum Information and Computation}\ }\textbf {\bibinfo {volume} {6}},\ \bibinfo {pages} {382} (\bibinfo {year} {2006})}\BibitemShut {NoStop}%
\bibitem [{\citenamefont {Kribs}\ \emph {et~al.}(2005)\citenamefont {Kribs}, \citenamefont {Laflamme},\ and\ \citenamefont {Poulin}}]{kribs2005unified}%
  \BibitemOpen
  \bibfield  {author} {\bibinfo {author} {\bibfnamefont {D.}~\bibnamefont {Kribs}}, \bibinfo {author} {\bibfnamefont {R.}~\bibnamefont {Laflamme}},\ and\ \bibinfo {author} {\bibfnamefont {D.}~\bibnamefont {Poulin}},\ }\bibfield  {title} {\bibinfo {title} {Unified and generalized approach to quantum error correction},\ }\href {https://doi.org/10.1103/physrevlett.94.180501} {\bibfield  {journal} {\bibinfo  {journal} {Physical Review Letters}\ }\textbf {\bibinfo {volume} {94}},\ \bibinfo {pages} {180501} (\bibinfo {year} {2005})}\BibitemShut {NoStop}%
\bibitem [{\citenamefont {Poulin}(2005)}]{poulin2005stabilizer}%
  \BibitemOpen
  \bibfield  {author} {\bibinfo {author} {\bibfnamefont {D.}~\bibnamefont {Poulin}},\ }\bibfield  {title} {\bibinfo {title} {Stabilizer formalism for operator quantum error correction},\ }\href {https://doi.org/10.1103/physrevlett.95.230504} {\bibfield  {journal} {\bibinfo  {journal} {Physical Review Letters}\ }\textbf {\bibinfo {volume} {95}},\ \bibinfo {pages} {230504} (\bibinfo {year} {2005})}\BibitemShut {NoStop}%
\bibitem [{\citenamefont {Suchara}\ \emph {et~al.}(2011)\citenamefont {Suchara}, \citenamefont {Bravyi},\ and\ \citenamefont {Terhal}}]{suchara2011constructions}%
  \BibitemOpen
  \bibfield  {author} {\bibinfo {author} {\bibfnamefont {M.}~\bibnamefont {Suchara}}, \bibinfo {author} {\bibfnamefont {S.}~\bibnamefont {Bravyi}},\ and\ \bibinfo {author} {\bibfnamefont {B.}~\bibnamefont {Terhal}},\ }\bibfield  {title} {\bibinfo {title} {Constructions and noise threshold of topological subsystem codes},\ }\href {https://doi.org/10.1088/1751-8113/44/15/155301} {\bibfield  {journal} {\bibinfo  {journal} {Journal of Physics A: Mathematical and Theoretical}\ }\textbf {\bibinfo {volume} {44}},\ \bibinfo {pages} {155301} (\bibinfo {year} {2011})}\BibitemShut {NoStop}%
\bibitem [{\citenamefont {Higgott}\ and\ \citenamefont {Breuckmann}(2021)}]{higgott2021subsystem}%
  \BibitemOpen
  \bibfield  {author} {\bibinfo {author} {\bibfnamefont {O.}~\bibnamefont {Higgott}}\ and\ \bibinfo {author} {\bibfnamefont {N.~P.}\ \bibnamefont {Breuckmann}},\ }\bibfield  {title} {\bibinfo {title} {Subsystem codes with high thresholds by gauge fixing and reduced qubit overhead},\ }\href {https://doi.org/10.1103/physrevx.11.031039} {\bibfield  {journal} {\bibinfo  {journal} {Physical Review X}\ }\textbf {\bibinfo {volume} {11}},\ \bibinfo {pages} {031039} (\bibinfo {year} {2021})}\BibitemShut {NoStop}%
\bibitem [{\citenamefont {Shaw}\ \emph {et~al.}(2008)\citenamefont {Shaw}, \citenamefont {Wilde}, \citenamefont {Oreshkov}, \citenamefont {Kremsky},\ and\ \citenamefont {Lidar}}]{shaw2008encoding}%
  \BibitemOpen
  \bibfield  {author} {\bibinfo {author} {\bibfnamefont {B.}~\bibnamefont {Shaw}}, \bibinfo {author} {\bibfnamefont {M.~M.}\ \bibnamefont {Wilde}}, \bibinfo {author} {\bibfnamefont {O.}~\bibnamefont {Oreshkov}}, \bibinfo {author} {\bibfnamefont {I.}~\bibnamefont {Kremsky}},\ and\ \bibinfo {author} {\bibfnamefont {D.~A.}\ \bibnamefont {Lidar}},\ }\bibfield  {title} {\bibinfo {title} {Encoding one logical qubit into six physical qubits},\ }\href {https://doi.org/10.1103/physreva.78.012337} {\bibfield  {journal} {\bibinfo  {journal} {Physical Review A}\ }\textbf {\bibinfo {volume} {78}},\ \bibinfo {pages} {012337} (\bibinfo {year} {2008})}\BibitemShut {NoStop}%
\bibitem [{ecz(2024)}]{eczoo_bacon_shor_9}%
  \BibitemOpen
  \bibfield  {title} {\bibinfo {title} {$[[9,1,4,3]]$ {N}ine-qubit {B}acon-{S}hor code},\ }in\ \href {https://errorcorrectionzoo.org/c/bacon_shor_9} {\emph {\bibinfo {booktitle} {The Error Correction Zoo}}},\ \bibinfo {editor} {edited by\ \bibinfo {editor} {\bibfnamefont {V.~V.}\ \bibnamefont {Albert}}\ and\ \bibinfo {editor} {\bibfnamefont {P.}~\bibnamefont {Faist}}}\ (\bibinfo {year} {2024})\BibitemShut {NoStop}%
\bibitem [{\citenamefont {Hastings}\ and\ \citenamefont {Haah}(2021)}]{hastings2021dynamically}%
  \BibitemOpen
  \bibfield  {author} {\bibinfo {author} {\bibfnamefont {M.~B.}\ \bibnamefont {Hastings}}\ and\ \bibinfo {author} {\bibfnamefont {J.}~\bibnamefont {Haah}},\ }\bibfield  {title} {\bibinfo {title} {Dynamically generated logical qubits},\ }\href {https://doi.org/10.22331/q-2021-10-19-564} {\bibfield  {journal} {\bibinfo  {journal} {Quantum}\ }\textbf {\bibinfo {volume} {5}},\ \bibinfo {pages} {564} (\bibinfo {year} {2021})}\BibitemShut {NoStop}%
\bibitem [{\citenamefont {Chermak}(2013)}]{chermak2013fusion}%
  \BibitemOpen
  \bibfield  {author} {\bibinfo {author} {\bibfnamefont {A.}~\bibnamefont {Chermak}},\ }\bibfield  {title} {\bibinfo {title} {Fusion systems and localities},\ }\href {https://doi.org/10.1007/s11511-013-0099-5} {\bibfield  {journal} {\bibinfo  {journal} {Acta Mathematica}\ }\textbf {\bibinfo {volume} {211}},\ \bibinfo {pages} {47} (\bibinfo {year} {2013})}\BibitemShut {NoStop}%
\bibitem [{\citenamefont {Anderson}\ \emph {et~al.}(2014)\citenamefont {Anderson}, \citenamefont {Duclos-Cianci},\ and\ \citenamefont {Poulin}}]{anderson2014fault}%
  \BibitemOpen
  \bibfield  {author} {\bibinfo {author} {\bibfnamefont {J.~T.}\ \bibnamefont {Anderson}}, \bibinfo {author} {\bibfnamefont {G.}~\bibnamefont {Duclos-Cianci}},\ and\ \bibinfo {author} {\bibfnamefont {D.}~\bibnamefont {Poulin}},\ }\bibfield  {title} {\bibinfo {title} {Fault-tolerant conversion between the {S}teane and {R}eed-{M}uller quantum codes},\ }\href {https://doi.org/10.1103/physrevlett.113.080501} {\bibfield  {journal} {\bibinfo  {journal} {Physical Review Letters}\ }\textbf {\bibinfo {volume} {113}},\ \bibinfo {pages} {080501} (\bibinfo {year} {2014})}\BibitemShut {NoStop}%
\bibitem [{\citenamefont {Bombin}\ and\ \citenamefont {Martin-Delgado}(2006)}]{bombin2006topological}%
  \BibitemOpen
  \bibfield  {author} {\bibinfo {author} {\bibfnamefont {H.}~\bibnamefont {Bombin}}\ and\ \bibinfo {author} {\bibfnamefont {M.~A.}\ \bibnamefont {Martin-Delgado}},\ }\bibfield  {title} {\bibinfo {title} {Topological quantum distillation},\ }\href {https://doi.org/10.1103/physrevlett.97.180501} {\bibfield  {journal} {\bibinfo  {journal} {Physical Review Letters}\ }\textbf {\bibinfo {volume} {97}},\ \bibinfo {pages} {180501} (\bibinfo {year} {2006})}\BibitemShut {NoStop}%
\bibitem [{\citenamefont {Rengaswamy}\ \emph {et~al.}(2020)\citenamefont {Rengaswamy}, \citenamefont {Calderbank}, \citenamefont {Newman},\ and\ \citenamefont {Pfister}}]{rengaswamy2020optimal}%
  \BibitemOpen
  \bibfield  {author} {\bibinfo {author} {\bibfnamefont {N.}~\bibnamefont {Rengaswamy}}, \bibinfo {author} {\bibfnamefont {R.}~\bibnamefont {Calderbank}}, \bibinfo {author} {\bibfnamefont {M.}~\bibnamefont {Newman}},\ and\ \bibinfo {author} {\bibfnamefont {H.~D.}\ \bibnamefont {Pfister}},\ }\bibfield  {title} {\bibinfo {title} {On optimality of {CSS} codes for transversal {$T$}},\ }\href {https://doi.org/10.1109/jsait.2020.3012914} {\bibfield  {journal} {\bibinfo  {journal} {IEEE Journal on Selected Areas in Information Theory}\ }\textbf {\bibinfo {volume} {1}},\ \bibinfo {pages} {499} (\bibinfo {year} {2020})}\BibitemShut {NoStop}%
\bibitem [{\citenamefont {Sudakov}\ and\ \citenamefont {Vieira}(2018)}]{sudakov2018two}%
  \BibitemOpen
  \bibfield  {author} {\bibinfo {author} {\bibfnamefont {B.}~\bibnamefont {Sudakov}}\ and\ \bibinfo {author} {\bibfnamefont {P.}~\bibnamefont {Vieira}},\ }\bibfield  {title} {\bibinfo {title} {Two remarks on eventown and oddtown problems},\ }\href {https://doi.org/10.1137/16m1100666} {\bibfield  {journal} {\bibinfo  {journal} {SIAM Journal on Discrete Mathematics}\ }\textbf {\bibinfo {volume} {32}},\ \bibinfo {pages} {280} (\bibinfo {year} {2018})}\BibitemShut {NoStop}%
\bibitem [{\citenamefont {Abramsky}\ \emph {et~al.}(2016)\citenamefont {Abramsky}, \citenamefont {Barbosa},\ and\ \citenamefont {Mansfield}}]{abramsky2016quantifying}%
  \BibitemOpen
  \bibfield  {author} {\bibinfo {author} {\bibfnamefont {S.}~\bibnamefont {Abramsky}}, \bibinfo {author} {\bibfnamefont {R.}~\bibnamefont {Barbosa}},\ and\ \bibinfo {author} {\bibfnamefont {S.}~\bibnamefont {Mansfield}},\ }\bibfield  {title} {\bibinfo {title} {Quantifying contextuality via linear programming},\ }in\ \href {http://qpl2016.cis.strath.ac.uk/} {\emph {\bibinfo {booktitle} {13th International Conference on Quantum Physics and Logic (QPL 2016)}}}\ (\bibinfo {year} {2016})\BibitemShut {NoStop}%
\bibitem [{\citenamefont {Boyd}\ and\ \citenamefont {Vandenberghe}(2004)}]{boyd2004convex}%
  \BibitemOpen
  \bibfield  {author} {\bibinfo {author} {\bibfnamefont {S.}~\bibnamefont {Boyd}}\ and\ \bibinfo {author} {\bibfnamefont {L.}~\bibnamefont {Vandenberghe}},\ }\href {https://doi.org/10.1017/cbo9780511804441} {\emph {\bibinfo {title} {Convex Optimization}}}\ (\bibinfo  {publisher} {Cambridge University Press},\ \bibinfo {year} {2004})\BibitemShut {NoStop}%
\bibitem [{\citenamefont {Panteleev}\ and\ \citenamefont {Kalachev}(2024)}]{panteleev2024maximally}%
  \BibitemOpen
  \bibfield  {author} {\bibinfo {author} {\bibfnamefont {P.}~\bibnamefont {Panteleev}}\ and\ \bibinfo {author} {\bibfnamefont {G.}~\bibnamefont {Kalachev}},\ }\href@noop {} {\bibinfo {title} {Maximally extendable sheaf codes}} (\bibinfo {year} {2024}),\ \Eprint {https://arxiv.org/abs/2403.03651} {arXiv:2403.03651 [cs.IT]} \BibitemShut {NoStop}%
\bibitem [{\citenamefont {Lin}(2024)}]{lin2024transversal}%
  \BibitemOpen
  \bibfield  {author} {\bibinfo {author} {\bibfnamefont {T.-C.}\ \bibnamefont {Lin}},\ }\href@noop {} {\bibinfo {title} {Transversal non-{C}lifford gates for quantum {LDPC} codes on sheaves}} (\bibinfo {year} {2024}),\ \Eprint {https://arxiv.org/abs/2410.14631} {arXiv:2410.14631 [quant-ph]} \BibitemShut {NoStop}%
\bibitem [{\citenamefont {Dinur}\ \emph {et~al.}(2024)\citenamefont {Dinur}, \citenamefont {Lin},\ and\ \citenamefont {Vidick}}]{dinur2024expansion}%
  \BibitemOpen
  \bibfield  {author} {\bibinfo {author} {\bibfnamefont {I.}~\bibnamefont {Dinur}}, \bibinfo {author} {\bibfnamefont {T.-C.}\ \bibnamefont {Lin}},\ and\ \bibinfo {author} {\bibfnamefont {T.}~\bibnamefont {Vidick}},\ }\href@noop {} {\bibinfo {title} {Expansion of higher-dimensional cubical complexes with application to quantum locally testable codes}} (\bibinfo {year} {2024}),\ \Eprint {https://arxiv.org/abs/2402.07476} {arXiv:2402.07476 [quant-ph]} \BibitemShut {NoStop}%
\bibitem [{\citenamefont {Kalachev}\ and\ \citenamefont {Panteleev}(2025)}]{kalachev2025maximally}%
  \BibitemOpen
  \bibfield  {author} {\bibinfo {author} {\bibfnamefont {G.}~\bibnamefont {Kalachev}}\ and\ \bibinfo {author} {\bibfnamefont {P.}~\bibnamefont {Panteleev}},\ }\href@noop {} {\bibinfo {title} {Maximally extendable product codes are good coboundary expanders}} (\bibinfo {year} {2025}),\ \Eprint {https://arxiv.org/abs/2501.01411} {arXiv:2501.01411 [cs.IT]} \BibitemShut {NoStop}%
\bibitem [{\citenamefont {Davydova}\ \emph {et~al.}(2023)\citenamefont {Davydova}, \citenamefont {Tantivasadakarn},\ and\ \citenamefont {Balasubramanian}}]{davydova2023floquet}%
  \BibitemOpen
  \bibfield  {author} {\bibinfo {author} {\bibfnamefont {M.}~\bibnamefont {Davydova}}, \bibinfo {author} {\bibfnamefont {N.}~\bibnamefont {Tantivasadakarn}},\ and\ \bibinfo {author} {\bibfnamefont {S.}~\bibnamefont {Balasubramanian}},\ }\bibfield  {title} {\bibinfo {title} {Floquet codes without parent subsystem codes},\ }\href {https://doi.org/10.1103/prxquantum.4.020341} {\bibfield  {journal} {\bibinfo  {journal} {PRX Quantum}\ }\textbf {\bibinfo {volume} {4}},\ \bibinfo {pages} {020341} (\bibinfo {year} {2023})}\BibitemShut {NoStop}%
\bibitem [{\citenamefont {Fahimniya}\ \emph {et~al.}(2024)\citenamefont {Fahimniya}, \citenamefont {Dehghani}, \citenamefont {Bharti}, \citenamefont {Mathew}, \citenamefont {Koll{\'a}r}, \citenamefont {Gorshkov},\ and\ \citenamefont {Gullans}}]{fahimniya2023hyperbolic}%
  \BibitemOpen
  \bibfield  {author} {\bibinfo {author} {\bibfnamefont {A.}~\bibnamefont {Fahimniya}}, \bibinfo {author} {\bibfnamefont {H.}~\bibnamefont {Dehghani}}, \bibinfo {author} {\bibfnamefont {K.}~\bibnamefont {Bharti}}, \bibinfo {author} {\bibfnamefont {S.}~\bibnamefont {Mathew}}, \bibinfo {author} {\bibfnamefont {A.~J.}\ \bibnamefont {Koll{\'a}r}}, \bibinfo {author} {\bibfnamefont {A.~V.}\ \bibnamefont {Gorshkov}},\ and\ \bibinfo {author} {\bibfnamefont {M.~J.}\ \bibnamefont {Gullans}},\ }\href@noop {} {\bibinfo {title} {Fault-tolerant hyperbolic {F}loquet quantum error correcting codes}} (\bibinfo {year} {2024}),\ \Eprint {https://arxiv.org/abs/2309.10033} {arXiv:2309.10033 [quant-ph]} \BibitemShut {NoStop}%
\bibitem [{\citenamefont {Gidney}\ \emph {et~al.}(2021)\citenamefont {Gidney}, \citenamefont {Newman}, \citenamefont {Fowler},\ and\ \citenamefont {Broughton}}]{gidney2021fault}%
  \BibitemOpen
  \bibfield  {author} {\bibinfo {author} {\bibfnamefont {C.}~\bibnamefont {Gidney}}, \bibinfo {author} {\bibfnamefont {M.}~\bibnamefont {Newman}}, \bibinfo {author} {\bibfnamefont {A.}~\bibnamefont {Fowler}},\ and\ \bibinfo {author} {\bibfnamefont {M.}~\bibnamefont {Broughton}},\ }\bibfield  {title} {\bibinfo {title} {A fault-tolerant honeycomb memory},\ }\href {https://doi.org/10.22331/q-2021-12-20-605} {\bibfield  {journal} {\bibinfo  {journal} {Quantum}\ }\textbf {\bibinfo {volume} {5}},\ \bibinfo {pages} {605} (\bibinfo {year} {2021})}\BibitemShut {NoStop}%
\bibitem [{\citenamefont {Haah}\ and\ \citenamefont {Hastings}(2022)}]{haah2022boundaries}%
  \BibitemOpen
  \bibfield  {author} {\bibinfo {author} {\bibfnamefont {J.}~\bibnamefont {Haah}}\ and\ \bibinfo {author} {\bibfnamefont {M.~B.}\ \bibnamefont {Hastings}},\ }\bibfield  {title} {\bibinfo {title} {Boundaries for the honeycomb code},\ }\href {https://doi.org/10.22331/q-2022-04-21-693} {\bibfield  {journal} {\bibinfo  {journal} {Quantum}\ }\textbf {\bibinfo {volume} {6}},\ \bibinfo {pages} {693} (\bibinfo {year} {2022})}\BibitemShut {NoStop}%
\bibitem [{\citenamefont {Vuillot}(2021)}]{vuillot2021planar}%
  \BibitemOpen
  \bibfield  {author} {\bibinfo {author} {\bibfnamefont {C.}~\bibnamefont {Vuillot}},\ }\href@noop {} {\bibinfo {title} {Planar {F}loquet codes}} (\bibinfo {year} {2021}),\ \Eprint {https://arxiv.org/abs/2110.05348} {arXiv:2110.05348 [quant-ph]} \BibitemShut {NoStop}%
\bibitem [{\citenamefont {Paetznick}\ \emph {et~al.}(2023)\citenamefont {Paetznick}, \citenamefont {Knapp}, \citenamefont {Delfosse}, \citenamefont {Bauer}, \citenamefont {Haah}, \citenamefont {Hastings},\ and\ \citenamefont {da~Silva}}]{paetznick2023performance}%
  \BibitemOpen
  \bibfield  {author} {\bibinfo {author} {\bibfnamefont {A.}~\bibnamefont {Paetznick}}, \bibinfo {author} {\bibfnamefont {C.}~\bibnamefont {Knapp}}, \bibinfo {author} {\bibfnamefont {N.}~\bibnamefont {Delfosse}}, \bibinfo {author} {\bibfnamefont {B.}~\bibnamefont {Bauer}}, \bibinfo {author} {\bibfnamefont {J.}~\bibnamefont {Haah}}, \bibinfo {author} {\bibfnamefont {M.~B.}\ \bibnamefont {Hastings}},\ and\ \bibinfo {author} {\bibfnamefont {M.~P.}\ \bibnamefont {da~Silva}},\ }\bibfield  {title} {\bibinfo {title} {Performance of planar {F}loquet codes with {M}ajorana-based qubits},\ }\href {https://doi.org/10.1103/prxquantum.4.010310} {\bibfield  {journal} {\bibinfo  {journal} {PRX Quantum}\ }\textbf {\bibinfo {volume} {4}},\ \bibinfo {pages} {010310} (\bibinfo {year} {2023})}\BibitemShut {NoStop}%
\bibitem [{\citenamefont {Gidney}\ \emph {et~al.}(2022)\citenamefont {Gidney}, \citenamefont {Newman},\ and\ \citenamefont {McEwen}}]{gidney2022benchmarking}%
  \BibitemOpen
  \bibfield  {author} {\bibinfo {author} {\bibfnamefont {C.}~\bibnamefont {Gidney}}, \bibinfo {author} {\bibfnamefont {M.}~\bibnamefont {Newman}},\ and\ \bibinfo {author} {\bibfnamefont {M.}~\bibnamefont {McEwen}},\ }\bibfield  {title} {\bibinfo {title} {Benchmarking the planar honeycomb code},\ }\href {https://doi.org/10.22331/q-2022-09-21-813} {\bibfield  {journal} {\bibinfo  {journal} {Quantum}\ }\textbf {\bibinfo {volume} {6}},\ \bibinfo {pages} {813} (\bibinfo {year} {2022})}\BibitemShut {NoStop}%
\bibitem [{\citenamefont {Hilaire}\ \emph {et~al.}(2024)\citenamefont {Hilaire}, \citenamefont {Dessertaine}, \citenamefont {Bourdoncle}, \citenamefont {Denys}, \citenamefont {de~Gliniasty}, \citenamefont {Valent{\'i}-Rojas},\ and\ \citenamefont {Mansfield}}]{hilaire2024enhanced}%
  \BibitemOpen
  \bibfield  {author} {\bibinfo {author} {\bibfnamefont {P.}~\bibnamefont {Hilaire}}, \bibinfo {author} {\bibfnamefont {T.}~\bibnamefont {Dessertaine}}, \bibinfo {author} {\bibfnamefont {B.}~\bibnamefont {Bourdoncle}}, \bibinfo {author} {\bibfnamefont {A.}~\bibnamefont {Denys}}, \bibinfo {author} {\bibfnamefont {G.}~\bibnamefont {de~Gliniasty}}, \bibinfo {author} {\bibfnamefont {G.}~\bibnamefont {Valent{\'i}-Rojas}},\ and\ \bibinfo {author} {\bibfnamefont {S.}~\bibnamefont {Mansfield}},\ }\href@noop {} {\bibinfo {title} {Enhanced fault-tolerance in photonic quantum computing: {F}loquet code outperforms surface code in tailored architecture}} (\bibinfo {year} {2024}),\ \Eprint {https://arxiv.org/abs/2410.07065} {arXiv:2410.07065 [quant-ph]} \BibitemShut {NoStop}%
\bibitem [{\citenamefont {Higgott}\ and\ \citenamefont {Breuckmann}(2024)}]{higgott2023constructions}%
  \BibitemOpen
  \bibfield  {author} {\bibinfo {author} {\bibfnamefont {O.}~\bibnamefont {Higgott}}\ and\ \bibinfo {author} {\bibfnamefont {N.~P.}\ \bibnamefont {Breuckmann}},\ }\bibfield  {title} {\bibinfo {title} {Constructions and performance of hyperbolic and semi-hyperbolic {F}loquet codes},\ }\href {https://doi.org/10.1103/prxquantum.5.040327} {\bibfield  {journal} {\bibinfo  {journal} {PRX Quantum}\ }\textbf {\bibinfo {volume} {5}},\ \bibinfo {pages} {040327} (\bibinfo {year} {2024})}\BibitemShut {NoStop}%
\bibitem [{\citenamefont {Aasen}\ \emph {et~al.}(2023)\citenamefont {Aasen}, \citenamefont {Haah}, \citenamefont {Li},\ and\ \citenamefont {Mong}}]{aasen2023measurement}%
  \BibitemOpen
  \bibfield  {author} {\bibinfo {author} {\bibfnamefont {D.}~\bibnamefont {Aasen}}, \bibinfo {author} {\bibfnamefont {J.}~\bibnamefont {Haah}}, \bibinfo {author} {\bibfnamefont {Z.}~\bibnamefont {Li}},\ and\ \bibinfo {author} {\bibfnamefont {R.~S.~K.}\ \bibnamefont {Mong}},\ }\href@noop {} {\bibinfo {title} {Measurement quantum cellular automata and anomalies in {F}loquet codes}} (\bibinfo {year} {2023}),\ \Eprint {https://arxiv.org/abs/2304.01277} {arXiv:2304.01277 [quant-ph]} \BibitemShut {NoStop}%
\bibitem [{\citenamefont {Zhang}\ \emph {et~al.}(2023)\citenamefont {Zhang}, \citenamefont {Aasen},\ and\ \citenamefont {Vijay}}]{zhang2023x}%
  \BibitemOpen
  \bibfield  {author} {\bibinfo {author} {\bibfnamefont {Z.}~\bibnamefont {Zhang}}, \bibinfo {author} {\bibfnamefont {D.}~\bibnamefont {Aasen}},\ and\ \bibinfo {author} {\bibfnamefont {S.}~\bibnamefont {Vijay}},\ }\bibfield  {title} {\bibinfo {title} {{$X$-cube} {F}loquet code: {A} dynamical quantum error correcting code with a subextensive number of logical qubits},\ }\href {https://doi.org/10.1103/physrevb.108.205116} {\bibfield  {journal} {\bibinfo  {journal} {Physical Review B}\ }\textbf {\bibinfo {volume} {108}},\ \bibinfo {pages} {205116} (\bibinfo {year} {2023})}\BibitemShut {NoStop}%
\bibitem [{\citenamefont {Dua}\ \emph {et~al.}(2024)\citenamefont {Dua}, \citenamefont {Tantivasadakarn}, \citenamefont {Sullivan},\ and\ \citenamefont {Ellison}}]{dua2024engineering}%
  \BibitemOpen
  \bibfield  {author} {\bibinfo {author} {\bibfnamefont {A.}~\bibnamefont {Dua}}, \bibinfo {author} {\bibfnamefont {N.}~\bibnamefont {Tantivasadakarn}}, \bibinfo {author} {\bibfnamefont {J.}~\bibnamefont {Sullivan}},\ and\ \bibinfo {author} {\bibfnamefont {T.~D.}\ \bibnamefont {Ellison}},\ }\bibfield  {title} {\bibinfo {title} {Engineering 3{D} {F}loquet codes by rewinding},\ }\href {https://doi.org/10.1103/prxquantum.5.020305} {\bibfield  {journal} {\bibinfo  {journal} {PRX Quantum}\ }\textbf {\bibinfo {volume} {5}},\ \bibinfo {pages} {020305} (\bibinfo {year} {2024})}\BibitemShut {NoStop}%
\bibitem [{\citenamefont {Alam}\ and\ \citenamefont {Rieffel}(2024)}]{alam2024dynamicallogicalqubitsbaconshor}%
  \BibitemOpen
  \bibfield  {author} {\bibinfo {author} {\bibfnamefont {M.~S.}\ \bibnamefont {Alam}}\ and\ \bibinfo {author} {\bibfnamefont {E.}~\bibnamefont {Rieffel}},\ }\href@noop {} {\bibinfo {title} {Dynamical logical qubits in the {B}acon-{S}hor code}} (\bibinfo {year} {2024}),\ \Eprint {https://arxiv.org/abs/2403.03291} {arXiv:2403.03291 [quant-ph]} \BibitemShut {NoStop}%
\bibitem [{\citenamefont {Bombin}\ \emph {et~al.}(2024)\citenamefont {Bombin}, \citenamefont {Litinski}, \citenamefont {Nickerson}, \citenamefont {Pastawski},\ and\ \citenamefont {Roberts}}]{bombin2023unifying}%
  \BibitemOpen
  \bibfield  {author} {\bibinfo {author} {\bibfnamefont {H.}~\bibnamefont {Bombin}}, \bibinfo {author} {\bibfnamefont {D.}~\bibnamefont {Litinski}}, \bibinfo {author} {\bibfnamefont {N.}~\bibnamefont {Nickerson}}, \bibinfo {author} {\bibfnamefont {F.}~\bibnamefont {Pastawski}},\ and\ \bibinfo {author} {\bibfnamefont {S.}~\bibnamefont {Roberts}},\ }\bibfield  {title} {\bibinfo {title} {Unifying flavors of fault tolerance with the {ZX} calculus},\ }\href {https://doi.org/10.22331/q-2024-06-18-1379} {\bibfield  {journal} {\bibinfo  {journal} {Quantum}\ }\textbf {\bibinfo {volume} {8}},\ \bibinfo {pages} {1379} (\bibinfo {year} {2024})}\BibitemShut {NoStop}%
\bibitem [{\citenamefont {Tanggara}\ \emph {et~al.}(2024)\citenamefont {Tanggara}, \citenamefont {Gu},\ and\ \citenamefont {Bharti}}]{tanggara2024strategic}%
  \BibitemOpen
  \bibfield  {author} {\bibinfo {author} {\bibfnamefont {A.}~\bibnamefont {Tanggara}}, \bibinfo {author} {\bibfnamefont {M.}~\bibnamefont {Gu}},\ and\ \bibinfo {author} {\bibfnamefont {K.}~\bibnamefont {Bharti}},\ }\href@noop {} {\bibinfo {title} {Strategic code: {A} unified spatio-temporal framework for quantum error-correction}} (\bibinfo {year} {2024}),\ \Eprint {https://arxiv.org/abs/2405.17567} {arXiv:2405.17567 [quant-ph]} \BibitemShut {NoStop}%
\bibitem [{\citenamefont {Townsend-Teague}\ \emph {et~al.}(2023)\citenamefont {Townsend-Teague}, \citenamefont {Magdalena de~la Fuente},\ and\ \citenamefont {Kesselring}}]{townsend2023floquetifying}%
  \BibitemOpen
  \bibfield  {author} {\bibinfo {author} {\bibfnamefont {A.}~\bibnamefont {Townsend-Teague}}, \bibinfo {author} {\bibfnamefont {J.}~\bibnamefont {Magdalena de~la Fuente}},\ and\ \bibinfo {author} {\bibfnamefont {M.}~\bibnamefont {Kesselring}},\ }\bibfield  {title} {\bibinfo {title} {Floquetifying the colour code},\ }\href {https://doi.org/10.4204/eptcs.384.14} {\bibfield  {journal} {\bibinfo  {journal} {Electronic Proceedings in Theoretical Computer Science}\ }\textbf {\bibinfo {volume} {384}},\ \bibinfo {pages} {265} (\bibinfo {year} {2023})}\BibitemShut {NoStop}%
\bibitem [{\citenamefont {Ellison}\ \emph {et~al.}(2023)\citenamefont {Ellison}, \citenamefont {Sullivan},\ and\ \citenamefont {Dua}}]{ellison2023floquet}%
  \BibitemOpen
  \bibfield  {author} {\bibinfo {author} {\bibfnamefont {T.~D.}\ \bibnamefont {Ellison}}, \bibinfo {author} {\bibfnamefont {J.}~\bibnamefont {Sullivan}},\ and\ \bibinfo {author} {\bibfnamefont {A.}~\bibnamefont {Dua}},\ }\href@noop {} {\bibinfo {title} {Floquet codes with a twist}} (\bibinfo {year} {2023}),\ \Eprint {https://arxiv.org/abs/2306.08027} {arXiv:2306.08027 [quant-ph]} \BibitemShut {NoStop}%
\bibitem [{\citenamefont {Sullivan}\ \emph {et~al.}(2023)\citenamefont {Sullivan}, \citenamefont {Wen},\ and\ \citenamefont {Potter}}]{sullivan2023floquet}%
  \BibitemOpen
  \bibfield  {author} {\bibinfo {author} {\bibfnamefont {J.}~\bibnamefont {Sullivan}}, \bibinfo {author} {\bibfnamefont {R.}~\bibnamefont {Wen}},\ and\ \bibinfo {author} {\bibfnamefont {A.~C.}\ \bibnamefont {Potter}},\ }\bibfield  {title} {\bibinfo {title} {Floquet codes and phases in twist-defect networks},\ }\href {https://doi.org/10.1103/physrevb.108.195134} {\bibfield  {journal} {\bibinfo  {journal} {Physical Review B}\ }\textbf {\bibinfo {volume} {108}},\ \bibinfo {pages} {195134} (\bibinfo {year} {2023})}\BibitemShut {NoStop}%
\bibitem [{\citenamefont {Tanggara}\ \emph {et~al.}(2025)\citenamefont {Tanggara}, \citenamefont {Gu},\ and\ \citenamefont {Bharti}}]{tanggara2024simple}%
  \BibitemOpen
  \bibfield  {author} {\bibinfo {author} {\bibfnamefont {A.}~\bibnamefont {Tanggara}}, \bibinfo {author} {\bibfnamefont {M.}~\bibnamefont {Gu}},\ and\ \bibinfo {author} {\bibfnamefont {K.}~\bibnamefont {Bharti}},\ }\href@noop {} {\bibinfo {title} {Simple construction of qudit {F}loquet codes on a family of lattices}} (\bibinfo {year} {2025}),\ \Eprint {https://arxiv.org/abs/2410.02022} {arXiv:2410.02022 [quant-ph]} \BibitemShut {NoStop}%
\end{thebibliography}%

\appendix

\section{Primer on Sheaf-Theoretic Contextuality} \label{sec:primer_sheaf_theory}

\noindent\textit{Note. This appendix is intended as a pedagogical supplement for readers less familiar with the mathematics; it is not required for following the main proofs in the paper.}

\bigskip

Sheaf theory provides a mathematical framework to study contextuality by modeling measurement outcomes across compatible sets of observables. Here, we introduce the basics of sheaves on a finite, discrete topological space, tailored to the quantum measurement scenarios in this paper.

\subsection{Sheaves on a Finite, Discrete Topological Space}

We consider a finite set $X$ equipped with the discrete topology, where every subset of $X$ is \emph{open}. Informally, a \emph{presheaf} is a structure that assigns data to all (open) subsets of X while providing a way to ``restrict'' this data from bigger to smaller subsets. A \emph{sheaf} is a presheaf that satisfies additional conditions, ensuring that this \emph{local} data can be consistently ``glued'' together to form unique \emph{global} data.

Formally, a \emph{presheaf} $\mathcal{F}$ on $X$ (e.g., a set of Pauli observables) assigns to each subset $U \subseteq X$ a set $\mathcal{F}(U)$ (e.g., possible measurement outcomes for the observables in $U$) and to each inclusion $U \subseteq V \subseteq X$ a \emph{restriction map} $\restr_{V,U} \colon \mathcal{F}(V) \to \mathcal{F}(U)$, satisfying

\begin{itemize}
    \item The restriction from any subset $U$ to itself is the identity, i.e., $\restr_{U,U}$ is the identity. 
    \item For $U \subseteq V \subseteq W$, the restriction from $W$ to $U$ is the composition of restrictions from $W$ to $V$ and $V$ to $U$, i.e., $\restr_{W,U} = \restr_{V,U} \circ \restr_{W,V}$. 
\end{itemize}
The elements $\mathcal{F}(U)$ are called \emph{sections} of $\mathcal{F}$ over $U$. For $V\subseteq U$ and $s\in\mathcal{F}(U)$, we sometimes write $s|_V := \restr_{U,V}(s)$ for convenience. A sheaf $\mathcal{F}$ on $X$ is a presheaf on $X$ that satisfies two additional axioms, which formalize the idea of ``locality'' and ``gluing'':

\begin{itemize}
    \item Locality: If two sections $s_1, s_2 \in \mathcal{F}(U)$ are equal on a \emph{cover} of $U$, then they must be the same section. Specifically, if $U = \cup_i U_i$ and $\restr_{U,U_i}(s_1) = \restr_{U,U_i}(s_2)$ for all $i$, then $s_1 = s_2$. In simpler terms, \emph{a section is uniquely determined by its local pieces}. No two different ``global'' sections look identical everywhere ``locally.''
    \item Gluing: Given any collection of sections $\{s_i \in \mathcal{F}(U_i)\}$ that agree on all the overlaps (i.e., $\restr_{U_i, U_i \cap U_j}(s_i) = \restr_{U_j, U_i \cap U_j}(s_j)$), these local sections can be ``glued together'' to form a unique global section $s \in \mathcal{F}(U)$, where $U = \cup_i U_i$. This unique global section $s$ must restrict to each of the original local sections, i.e., $\restr_{U,U_i}(s) = s_i$. This means that \emph{consistent local data can always be combined to form unique global data}.
\end{itemize}

\subsection{Empirical Models and Contextuality}

In this subsection, $X$ will always refer to a finite set of observables and $O$ will refer to the set of measurement outcomes for each observable, assumed to be the same across all observables (e.g., $O = \{0,1\}$, or equivalently, $O = \{\pm 1\}$ for Pauli measurements). The set $\mathcal{M}$ will refer to the set of all \emph{maximally} commuting subsets of $X$. The elements of $\mathcal{M}$ are called \emph{measurement contexts}.

We begin by defining the \emph{event sheaf} $\mathcal{E}$ on $X$, which assigns to each subset $U$ of observables the set $\mathcal{E}(U) = O^U$ of all functions from $U$ to $O$. Each element (i.e., section) of $\mathcal{E}(U)$ is a function that assigns an outcome in $O$ to every observable in $U$, representing a conceivable joint outcome assignment (whether or not it is empirically realizable in a given model, and whether or not the observables in $U$ commute). For example, if $U = \{X_1, X_2, X_1X_2\}$ and $O = \{0,1\}$, then $\mathcal{E}(U)$ contains $2^3 = 8$ assignments, even though only four of these are physically valid. The restriction maps $\restr_{V,U}$ for any $U \subseteq V$ are the natural ones, sending each $s \colon V \to O$ to its restriction to $U$, $s|_U \colon U \to O$. It can be checked that $\mathcal{E}$ satisfies all the presheaf and sheaf axioms.

Besides the event sheaf $\mathcal{E}$, we can also consider the presheaf $\mathcal{D}_{\R_{\ge 0}} \mathcal{E}$ of probability distributions on the joint measurement outcomes. Instead of assigning each subset $U$ of observables to the set $\mathcal{E}(U) = O^U$, we assign each $U$ to the set of all probability distributions on $\mathcal{E}(U)$ (i.e., $\R_{\ge 0}$-valued functions on $\mathcal{E}(U)$ whose function values sum to $1$). An $\R_{\ge 0}$-valued \emph{empirical model} is a choice of probability distributions $e_C \in \mathcal{D}_{\R_{\ge 0}} \mathcal{E}(C)$ for all the maximally commuting sets $C \in \mathcal{M}$, which are compatible in the sense that for any two maximally commuting sets $C$ and $C'$, their marginal distributions on $O^{C \cap C'}$ agree, i.e., $e_C|_{C \cap C'} = e_{C'}|_{C \cap C'}$. This compatibility condition corresponds to the no-signaling condition, ensuring consistency on overlaps.

Instead of probability distributions on the joint measurement outcomes, we can also consider the so-called ``possibility distributions'' on the outcomes, where we assign $0$ if we deem the joint measurement outcome impossible and $1$ if we deem the joint measurement outcome possible. (Note that every possibility distribution must have at least one joint measurement outcome that is deemed possible; this corresponds to the condition that probabilities in a probability distribution sum to $1$.) This gives us the presheaf $\mathcal{D}_{\B} \mathcal{E}$ instead, where $\B = \{0,1\}$ is the boolean semiring satisfying $1 + 1 = 1$. An $\B$-valued empirical model is a choice of compatible possibility distributions $e_C \in \mathcal{D}_{\B} \mathcal{E}(C)$ for all the maximally commuting sets $C \in \mathcal{M}$. (Compatibility is defined analogously as before, via restrictions on intersections.)

An equivalent description of a $\B$-valued empirical model is a \emph{possibilistic empirical model}, which is a subpresheaf $\mathcal{S}$ of $\mathcal{E}$ (i.e., $\mathcal{S}(U) \subseteq \mathcal{E}(U) = O^{U}$ for all $U \subseteq X$), where each $\mathcal{S}(U)$ represents all joint measurement outcomes that is deemed ``possible.'' Formally, this must satisfy the following properties:
\begin{enumerate} [label={(E\arabic*)}]
    \item $\mathcal{S}(C) \ne \emptyset$ for all $C \in \mathcal{M}$ (i.e., any joint measurement of a set of maximally commuting observables must yield at least one possible joint measurement outcome),
    \item $\mathcal{S}(U') \to \mathcal{S}(U)$ is surjective for all $C \in \mathcal{M}$ and $U \subseteq U' \subseteq C$ (i.e., if a measurement outcome is possible for a given set of commuting observables, such an outcome must still be possible even if we add more commuting observables to the joint measurement; this is the no-signaling condition),
    \item for any family $\{s_C\}_{C \in \mathcal{M}}$ with $s_C \in \mathcal{S}(C)$ satisfying $s_C|_{C\cap C'} = s_{C'}|_{C \cap C'}$ for all $C, C' \in \mathcal{M}$, there is a global section $s \in \mathcal{S}(X)$ satisfying $s|_{C} = s_C$ for all $C \in \mathcal{M}$ (i.e., if we can choose joint measurement outcomes for every set of maximally commuting observables in a compatible way---such that they agree on their pairwise intersections---then these choices can be extended to a consistent assignment of outcomes to all observables in $X$).
\end{enumerate}
Note that the presheaf $\mathcal{S}$ may not be a sheaf. There is a bijection between $B$-valued empirical models and possibilistic empirical models. Given a $B$-valued empirical model, we can use the possibility distributions $e_C$ for the maximally commuting sets of observables $C$ to directly define $\mathcal{S}(U)$ as the set of all possible joint measurement outcomes for each set of commuting observables $U$, and then extend this to general subsets of observables $U$ by defining $\mathcal{S}(U)$ to be all assignments $s \colon U \to O$ that are locally consistent, meaning that for every commuting subset $V \subseteq U$, the restriction $s|_V$ lies in $\mathcal{S}(V)$. Conversely, given a possibilistic empirical model $\mathcal{S}$, we can simply define each possibility distribution $e_C$ to be the indicator function on $C$. Possibilistic empirical models are particularly useful for capturing ``all-or-nothing'' logical contradictions in quantum systems.

The \emph{empirical model defined by a quantum state $\rho$} (on the Hilbert space where the observables in $X$ act) is the $\R_{\ge 0}$-valued empirical model $\{e_C\}_{C \in \mathcal{M}}$ where each $e_C$ is the probability distribution on $O^C$ induced by jointly measuring the commuting observables in $C$ on $\rho$. This satisfies the compatibility (i.e., no-signaling) condition automatically by the principles of quantum mechanics. We can also obtain a $\B$-valued empirical model from $\rho$ by considering \emph{possibility} distributions instead, and the procedure described in the preceding paragraph allows us to derive a possibilistic empirical model, which we call the \emph{possibilistic empirical model defined by $\rho$}.

We now define the relevant notions of contextuality. A possibilistic empirical model $\mathcal{S}$ is \emph{strongly contextual} if it has \emph{no} global section, i.e., $\mathcal{S}(X) = \emptyset$. This means there is no consistent assignment of outcomes to all observables in $X$ that is compatible with the possible local outcomes on every measurement context. An $\R_{\ge 0}$-valued empirical model $e = \{e_C\}_{C \in \mathcal{M}}$ is \emph{contextual} if there is \emph{no} global probability distribution $d$ on $O^X$ such that its marginal distribution on $O^C$ is $e_C$ for each $C \in \mathcal{M}$. It is \emph{strongly contextual} if its associated possibilistic empirical model is strongly contextual. If $e$ is not contextual, it is \emph{noncontextual}.

For an $\R_{\ge 0}$-valued empirical model, strong contextuality implies contextuality (see Theorem \ref{thm:strong-contexuality-implies-contextuality}), but the converse may not hold in general.

This framework allows us to define contextuality for QEC codes by applying it to their check measurements (Definition \ref{def:QECC_contextuality}), capturing nonclassical behavior in error syndrome measurement.

\subsection{A Simple Example}

As an example, consider $X = \{X_1, X_2, Z_1, Z_2\}$, with outcomes $O = \{0,1\}$. The contexts $M$ are the maximal commuting sets, namely $\{X_1, X_2\}$, $\{Z_1, Z_2\}$, $\{X_1, Z_2\}$ and $\{Z_1, X_2\}$.

An $\R_{\ge 0}$-empirical model assigns probabilities to the outcomes in each context (which are $(0,0), (0,1), (1,0), (1,1)$) in a compatible way. Here, compatibility means, for instance, that the probability that the measurement outcome of $(X_1, X_2)$ is $(0,0)$ or $(0,1)$ must be equal to the probability that the measurement outcome of $(X_1, Z_2)$ is $(0,0)$ or $(0,1)$. This is contextual if there is no global probability distribution over the $2^4 = 16$ possible measurement outcomes of $(X_1, X_2, Z_1, Z_2)$ whose marginal distribution when restricted to each of the four contexts coincides with the one provided by the model.

A possibilistic empirical model assigns to each subset of $X$ a set of outcomes which are deemed ``possible'', again in a compatible way. Here, compatibility means, for instance, that if the only possible outcomes for $(X_1, X_2)$ are $(0,0)$ and $(0,1)$, then the only possible outcome for $X_1$ is $0$ and the only possible outcomes for $(X_1, Z_2)$ are (among) $(0,0)$ and $(0,1)$. This is strongly contextual if the set of possible outcomes for $(X_1, X_2, Z_1, Z_2)$ is the empty set, or equivalently, if we cannot choose possible outcomes for each of the four contexts in a compatible way. An example of a strongly contextual possibilistic model $\mathcal{S}$ is one where\footnote{We slightly abuse notation here for the sake of notation clarity: each ordered pair listed in fact refers to function that sends the relevant operators (in the order listed) to the ordered pair. For instance, by $\mathcal{S}(\{X_1, X_2\}) = \{(0,0), (1,1)\}$, we mean that $\mathcal{S}(\{X_1, X_2\})$ consists of two functions: one sends both $X_1$ and $X_2$ to $0$, the other sends both $X_1$ and $X_2$ to $1$.} $\mathcal{S}(\{X_1, X_2\}) = \mathcal{S}(\{X_1, Z_2\}) = \mathcal{S}(\{Z_1, Z_2\}) = \{(0,0), (1,1)\}$ and $ \mathcal{S}(\{Z_1, X_2\}) = \{(0,1), (1,0)\}$. It can be checked that these do extend to a unique possibilistic model $\mathcal{S}$, and this $\mathcal{S}$ satisfies $\mathcal{S}(\{X_1, X_2, Z_1, Z_2\}) = \emptyset$ because there is no way to choose possible outcomes for each of the four contexts in a compatible way. The first three equations require all elements of $\mathcal{S}(\{X_1, X_2, Z_1, Z_2\})$ to be one of $(0,0,0,0)$ or $(1,1,1,1)$, but neither of this is consistent with the final equation. In contrast, had $(0,0)$ also been a section of $\{Z_1, X_2\}$ (i.e., $ \mathcal{S}(\{Z_1, X_2\}) = \{(0,1), (1,0), (0,0)\}$), then $\mathcal{S}(\{X_1, X_2, Z_1, Z_2\}) = \{(0,0,0,0)\} \ne \emptyset$ and $\mathcal{S}$ would not be strongly contextual.

\section{Primer on Partial Groups} \label{sec:primer_partial_groups}

\noindent\textit{Note. This appendix is intended as a pedagogical supplement for readers less familiar with the mathematics; it is not required for following the main proofs in the paper.}

\bigskip

This appendix provides a brief introduction to partial groups, aimed at readers familiar with group theory and Pauli groups but possibly new to this generalization. Partial groups extend standard groups to cases where multiplication is only partially defined (e.g., for commuting Pauli operators). They are used in Section \ref{sec:equivalence_of_contextuality_notions_under_partial_closure} to define partial closures of measurement sets, which in turn allow us to define Kochen--Specker contextuality (Definition \ref{def:kochen_specker}) and prove equivalences between different notions of contextuality (e.g., Theorems \ref{thm:ks-contextuality-iff-siavn} and \ref{thm:siavn-iff-kl-contextual}). For the full formal definitions, Pauli-specific examples, and applications, see Section~\ref{sec:partial_group_pauli_operators}.

\subsection{Motivation}

In the usual group-theoretic setting, any two elements can be multiplied, and the result is again in the group. However, in quantum measurement scenarios we are often interested in multiplying only commuting observables.
\begin{itemize}
    \item For two Pauli operators $P_1, P_2$, the product $P_1 P_2$ is only meaningful in our operational sense when $P_1$ and $P_2$ commute, since their eigenvalues can then be jointly measured.
    \item If $P_1$ and $P_2$ anticommute, the order of measurement matters and $\{P_1, P_2\}$ is not a valid measurement set.
\end{itemize}
This motivates \emph{partial groups}: sets with a multiplication defined only for certain tuples of elements.

\subsection{Partial Groups}

A \emph{partial group} generalizes a group by allowing the product to be defined only for certain ``allowed'' sequences of elements \cite{chermak2013fusion}. Let $\mathcal{P}$ be a nonempty set. Equip it with:
\begin{itemize}
    \item a subset $D \subseteq \mathcal{P}^*$ of ``allowed words'' (where $\mathcal{P}^*$ is the set of all finite-length sequences or ``words'' of elements from $\mathcal{P}$, including the empty word of length $0$, with concatenation of words given by $\circ$),
    \item a product operation $\Pi \colon D \to \mathcal{P}$ which multiplies all the elements in each allowed word to form an element in $\mathcal{P}$, and
    \item an inversion operation $(\blank)^{-1} \colon \mathcal{P} \to \mathcal{P}$.
\end{itemize}
These operations have to satisfy several additional axioms. Intuitively, these axioms ensure group-like behavior where defined: single elements are allowed, products are associative when possible, inverses exist, and there is an identity $\mathbf{1}$. (For the precise axioms, see Definition \ref{def:partial-group} in the main text.)

This induces a partial binary operation on $\mathcal{P}$: for two elements $u, v \in \mathcal{P}$, we can multiply these together (in the order $u, v$) if $u \circ v$ is an allowed word (i.e., lies in $D$). In this case, the product $u \cdot v$ is the result of applying the product operation $\Pi$ on $u \circ v$.  A partial group is \emph{abelian} if products commute where defined. (See Definition \ref{def:abelian-partial-group} in the main text.)

Any ordinary group $G$ is a partial group with all words allowed ($\mathcal{D} = \mathcal{P}^*$) and usual multiplication and inversion. Abelian groups yield abelian partial groups.

For an example of a partial group that is not a group (Example \ref{eg:pauli-partial-group}), consider the set of n-qubit Pauli observables $\mathcal{P}_n$, i.e., operators of the form $\pm P_1 \otimes \cdots \otimes P_n$ for $P_i \in \{I, X, Y, Z\}$. (The phases $\pm i$ are excluded because they do not correspond to observables.) This set is not a group under regular multiplication, because it fails to satisfy the closure property: for instance, $Z_1 X_1 = iY_1$ is not an observable. However, this set is closed under multiplication of \emph{commuting} operators. (To see this, observe that two operators $\pm P_1 \otimes \cdots \otimes P_n$ and $\pm Q_1 \otimes \cdots \otimes Q_n$ commute if and only if there are an even number of $j$ such that $P_j$ and $Q_j$ anticommute. If $P_j$ and $Q_j$ anticommutes---and hence are both distinct, and also distinct from $I$---then $P_j Q_j = i R_j$, where $R_j$ is the operator among $\{X, Y, Z\}$ that is distinct from $P_j$ and $Q_j$. If $P_j$ and $Q_j$ commute, then $P_j Q_j = R_j$ is among $\{I, X, Y, Z\}$. This means that the product of commuting $\pm P_1 \otimes \cdots \otimes P_n$ and $\pm Q_1 \otimes \cdots \otimes Q_n$ must be of the form $\pm R_1 \otimes \cdots \otimes R_n$.) Therefore, $\mathcal{P}_n$ is a partial group under this partial multiplication operation (i.e., restricted to commuting operators).

\subsection{Partial Subgroups and Closures} \label{sec:primer_partial_subgroups_closure}

A \emph{partial subgroup} $\mathcal{Q}$ of a partial group $\mathcal{P}$ is a subset of $\mathcal{P}$ which is closed under inversion and allowed products. In this case, $\mathcal{Q}$ is itself a partial group with the same partial multiplication operation as $\mathcal{P}$.

The \emph{partial closure} $\bar{S}$ of a subset $S$ in a partial group $\mathcal{P}$ is the smallest partial subgroup of $\mathcal{P}$ containing $S$. In practice, if $\mathcal{P}$ is a finite set, $\bar{S}$ can be constructed iteratively: start with $S$, then repeatedly add products of allowed (e.g., commuting) elements until no more new ones can be added.

As an example, consider the set of $2$-qubit Pauli observables $S = \{X_1, X_2, Z_1, Z_2\}$ in $\mathcal{P}_2$. To calculate $\bar{S}$, we do the following:
\begin{itemize}
    \item We first take products of all commuting elements. On top of the four elements already in $S$, we also get five more elements: $I = X_1 X_1$, $X_1 X_2$, $Z_1 X_2$, $Z_1 Z_2$ and $X_1 Z_2$.
    \item Taking products of all commuting elements gives us two new elements: $Y_1 Y_2 = (X_1 Z_2) (Z_1 X_2)$ and $- Y_1 Y_2 = (X_1 X_2) (Z_1 Z_2)$.
    \item Doing this again, we get five new elements: $-I = (- Y_1 Y_2)(Y_1 Y_2)$, $- X_1 Z_2 = (- Y_1 Y_2)(Z_1 X_2)$, $-Z_1 X_2 = (- Y_1 Y_2)(X_1 Z_2)$, $-X_1 X_2 = (Y_1 Y_2)(Z_1 Z_2)$, and $-Z_1 Z_2 = (Y_1 Y_2)(X_1 X_2)$.
    \item Since $-I$ commutes with every Pauli observable, the next step in the process gives us four new elements by multiplying $-I$ with the four original elements in $S$: $-X_1, -X_2, -Z_1, -Z_2$.
    \item Finally, the product of any two commuting elements in the $20$ elements obtained thus far yields another element among the $20$. This resulting set of $20$ elements is thus $\bar{S}$.
\end{itemize}

The notion of partial closure is especially relevant when discussing check measurements of a quantum error-correcting code. This is because operationally, if $P_1$ and $P_2$ are two commuting check measurements of a quantum error-correcting code, then measuring both (in any order) gives us the outcome of $P_1 P_2$ for free, without any additional measurements. Thus, any noncontextuality or contextuality property for the quantum error-correcting code should reasonably account for all such inferable measurements. By formalizing the set of check measurements as a subset of the partial group of Pauli observables and working with its partial closure, we ensure that all relevant jointly measurable observables are included.

\subsection{Partial Group Homomorphisms}

A \emph{partial group homomorphism} $\beta: \mathcal{P} \to \mathcal{Q}$ is a function between partial groups $\mathcal{P}$ and $\mathcal{Q}$ that preserves allowed words and products. We will only be interested in the case where $\mathcal{Q} = \{ \pm 1\}$ is an ordinary group. In this case, this definition simplifies to preserving binary products where defined (see Remark \ref{rem:partial-group-hom-equiv-def}), i.e., $\beta(u \cdot v) = \beta(u) \beta(v)$ for $u, v \in \mathcal{P}$ (where $\cdot$ is the partial multiplication operation on $\mathcal{P}$).

This notion is key for defining Kochen--Specker contextuality (Definition \ref{def:kochen_specker}): a set of $n$-qubit Pauli observables $X \subseteq \mathcal{P}_n$ is contextual in the Kochen--Specker sense if there is no partial group homomorphism $\lambda \colon \bar{X} \to \{\pm 1\}$ that satisfies $\lambda(-I) = -1$ whenever $-I \in \bar{X}$. Informally, this is saying that there is no way to assign $\pm 1$ (measurement) values to the operators in $\bar{X}$ that is consistent with taking commuting products in $\bar{X}$ (and with assigning $-I$ to $-1$, since the measurement of $-I$ is always $-1$).

\end{document}